\documentclass[dvipsnames]{lmcs} 


\usepackage{mathtools}
\usepackage{xcolor}
\usepackage{cmll}
\usepackage{stmaryrd}
\usepackage{amsthm,amssymb,amsmath}
\usepackage{MnSymbol}
\usepackage{tikz}
\usetikzlibrary{calc, arrows, arrows.meta, automata}
\usepackage{subcaption}
\usepackage{ebproof}
\usepackage{stmaryrd}
\usepackage{hyperref}
\usepackage{cleveref}
\usepackage{xifthen}
\usepackage{wrapfig}
\usepackage{tikz-cd}
\usepackage{CJKutf8}
\newcommand{\nat}{\mathbf{N}}
\newcommand{\rel}{\mathbf{Z}}

\DeclarePairedDelimiter\floor{\lfloor}{\rfloor}
\newcommand{\interp}[1]{\llbracket #1 \rrbracket}
\newcommand{\sinterp}[1]{\llangle #1 \rrangle}

\newcommand{\fmll}{\mathcal{F}_{\mathrm{MLL}}}
\newcommand{\stars}[1]{\mathtt{Stars}(#1)}
\newcommand{\terms}[1]{\mathtt{Terms}(#1)}
\newcommand{\idrays}[1]{\mathtt{IdRays}(#1)}
\newcommand{\test}[1]{\mathtt{Test}(#1)}
\newcommand{\tests}[1]{\mathtt{Tests}(#1)}

\newcommand{\axioms}[1]{\mathtt{Ax}(#1)}
\newcommand{\cuts}[1]{\mathtt{Cuts}(#1)}
\newcommand{\atoms}[1]{\mathtt{Atoms}(#1)}
\newcommand{\conclu}[1]{\mathtt{Concl}(#1)}
\newcommand{\idrayspol}[1]{\pm\mathtt{IdRays}(#1)}
\newcommand{\vars}[1]{\mathtt{vars}(#1)}
\newcommand{\freevars}[1]{\mathtt{fv}(#1)}
\newcommand{\colours}[1]{\mathtt{colours}(#1)}
\newcommand{\freerays}[1]{\mathtt{free}(#1)}
\newcommand{\ein}[1]{\ifthenelse{\equal{#1}{}}{\mathtt{in}}{\mathtt{in}(#1)}}
\newcommand{\eout}[1]{\ifthenelse{\equal{#1}{}}{\mathtt{out}}{\mathtt{out}(#1)}}
\newcommand{\diags}{\mathtt{Diags}}
\newcommand{\csatdiags}{\mathtt{CSatDiags}}
\newcommand{\satdiags}{\mathtt{SatDiags}}

\newcommand{\restrees}{\mathtt{RT}}

\newcommand{\alphaeq}{\approx_\alpha}
\newcommand{\eqq}{\stackrel{?}{=}}
\newcommand{\match}{\bowtie}

\newcommand{\lacc}{\mathtt{accept}}
\newcommand{\lrej}{\mathtt{reject}}

\newcommand{\gc}{\mathtt{g}}
\newcommand{\lc}{\mathtt{l}}
\newcommand{\rc}{\mathtt{r}}
\newcommand{\stopc}{\mathtt{s}}
\newcommand{\axl}{\mathrm{ax}}
\newcommand{\cutl}{\mathrm{cut}}
\newcommand{\testl}{\mathrm{test}}
\newcommand{\formatl}{\mathrm{format}}
\newcommand{\emptystar}{[]}

\newcommand{\indax}[1]{\mathtt{Ax}_{#1}}
\newcommand{\indcut}[2]{\mathtt{Cut}^{#1}(#2)}
\newcommand{\indtens}[2]{\mathtt{Tens}^{#1}(#2)}
\newcommand{\indpar}[2]{\mathtt{Par}^{#1}(#2)}
\newcommand{\indunion}[2]{#1 \uplus #2}

\newcommand{\axof}[1]{\Phi_{#1}^\axl}
\newcommand{\cutof}[1]{\Phi_{#1}^\cutl}
\newcommand{\axcutof}[1]{\Phi_{#1}^\axl \uplus \Phi_{#1}^\cutl}
\newcommand{\testof}[2]{\Phi_{#1}^{\testl(#2)}}
\newcommand{\compof}[1]{\Phi_{#1}^\mathtt{comp}}
\newcommand{\formatof}[1]{\Phi_{#1}^\formatl}
\newcommand{\concl}[1]{\overset{\leftarrow}{#1}}
\newcommand{\concr}[1]{\overset{\rightarrow}{#1}}
\newcommand{\str}[1]{\mathtt{str}(#1)}
\newcommand{\undereq}[1]{\ifthenelse{\equal{#1}{}}{\mathtt{eq}}{\mathtt{eq}(#1)}}

\newcommand{\ar}{\mathtt{ar}}
\newcommand{\op}{\mathtt{op}}

\newcommand{\subst}[2]{\{#1:=#2\}}
\newcommand{\substseq}[1]{\{#1\}}
\newcommand{\enat}[1]{\overline{#1}}
\newcommand{\solution}[1]{\mathtt{solution}(#1)}

\newcommand{\dgraph}[2][]{\ifthenelse{\equal{#1}{}}{\mathfrak{D}[#2]}{\mathfrak{D}[#2;#1]}}
\newcommand{\actu}{\mathop{\Downarrow}}
\newcommand{\exec}[1][]{\ifthenelse{\equal{#1}{}}{\mathtt{Ex}}{\mathtt{Ex}^{#1}}}

\newcommand{\roots}[1]{\mathtt{Roots}(#1)}

\newcommand{\wangn}[5]{
    \draw ($(#1)+(0,0)$)--($(#1)+(0.5,0.5)$)--($(#1)+(0,1)$)--cycle;
    \node at ($(#1)+(0.5,0.5)-(0.3, 0)$) {#5};
    \draw ($(#1)+(0,0)$)--($(#1)+(0.5,0.5)$)--($(#1)+(1,0)$)--cycle;
    \node at ($(#1)+(0.5,0.5)-(0, 0.3)$) {#4};
    \draw ($(#1)+(1,1)$)--($(#1)+(0.5,0.5)$)--($(#1)+(1,0)$)--cycle;
    \node at ($(#1)+(0.5,0.5)-(-0.3, 0)$) {#3};
    \draw ($(#1)+(1,1)$)--($(#1)+(0.5,0.5)$)--($(#1)+(0,1)$)--cycle;
    \node at ($(#1)+(0.5,0.5)-(0, -0.3)$) {#2};
}

\newcommand{\addr}{\mathtt{addr}}
\newcommand{\paddr}{\mathtt{pAddr}}

\newcommand{\ucap}{\Cap}

\newcommand{\setofaddresses}[2][x]{\mathrm{Addr}_{#1}(\mathcal{#2})}

\newcommand{\ccol}[1]{{\color{Bittersweet}#1}}
\newcommand{\tcol}[1]{{\color{MidnightBlue}#1}}

\newcommand{\ulocus}[2]{p_{#1}(#2)}
\newcommand{\clocus}[2]{\ccol{+c.\ulocus{#1}{#2}}}
\newcommand{\tlocus}[3][+]{\tcol{#1 t.\ulocus{#2}{#3}}}

\newcommand{\vcut}[2]{\ccol{-c.p_{#1}(#2)}}
\newcommand{\qray}[2][]{\ifthenelse{\equal{#1}{}}{p_{#2}(\gc \cdot X)}{\ccol{#1c.p_{#2}(\gc \cdot X)}}}

\newcommand{\gtensc}[3]{
    [\qray[-]{#1}, \qray[-]{#2}, \qray[+]{#3}]
}
\newcommand{\gparrc}[3]{
    [\qray[-]{#1}] +
    [\qray[-]{#2}, \qray[+]{#3}]
}
\newcommand{\gparlc}[3]{
    [\qray[-]{#1}, \qray[+]{#3}] +
    [\qray[-]{#2}]
}
\newcommand{\gconc}[1]{[\qray[-]{#1}, p_{#1}(X)]}

\newenvironment{Japanese}{%
    \CJKfamily{min}%
    \CJKtilde
    \CJKnospace}{}

\tikzset{
    dot/.style = {fill=black, circle, outer sep=2.5pt, inner sep=1pt},
    every loop/.style={}
}
\theoremstyle{plain}
\crefname{lem}{Lemma}{Lemmas}
\crefname{prop}{Proposition}{Propositions}
\crefname{rem}{Remark}{Remarks}
\crefname{exa}{Example}{Examples}
\crefname{conv}{Convention}{Conventions}
\crefname{defi}{Definition}{Definitions}
\crefname{nota}{Notation}{Notations}

\def\eg{{\em e.g.}\  }
\def\cf{{\em cf.}\  }
\def\ie{{\em i.e.}\  }
\def\wrt{{\em w.r.t.}\  }
\def\resp{{\em resp.}\  }

\begin{document}

\title[MLL from a resolution-based tile system]{Multiplicative linear logic from a resolution-based tile system}

\author[B.~Eng]{Boris Eng}	
\address{LIPN - Universit\'{e} Sorbonne Paris Nord}	
\email{engboris@hotmail.fr}  
\thanks{The authors would like to thank Valentin Maestracci for his useful comments. }	

\author[T.~Seiller]{Thomas Seiller}	
\address{CNRS / LIPN - Universit\'{e} Sorbonne Paris Nord}	
\email{seiller@lipn.fr}  





\begin{abstract}
We present the stellar resolution, a ``flexible'' tile system based on Robinson's first-order resolution. After establishing formal definitions and basic properties of the stellar resolution, we show its Turing-completeness and to illustrate the model, we exhibit how it naturally represents computation with Horn clauses and automata as well as non-deterministic tiling constructions used in DNA computing.
In the second and main part, by using the stellar resolution, we formalise and extend ideas of a new alternative to proof-net theory sketched by Girard in his transcendental syntax programme. In particular, we encode both cut-elimination and logical correctness for the multiplicative fragment of linear logic (MLL).
We finally obtain completeness results for both MLL and MLL extended with the so-called MIX rule.
By extending the ideas of Girard's geometry of interaction, this suggests a first step towards a new understanding of the interplay between logic and computation where linear logic is seen as a (constructed) way to format computation.
\end{abstract}

\maketitle

\section{Introduction}\label{sec:intro}

\paragraph{The evolution of the notion of proof} Among the different materialisations of logic, some remarkable and standard formalisms are Gentzen's natural deduction and sequent calculus \cite{gentzen1, gentzen2} which are attempts at formally representing mathematical reasoning by means of logical rules one applies successively to construct proofs. These rules are represented by means of \emph{sequents} which are expressions $\Gamma \vdash A$ stating that the conclusion $A$ follows from a set of hypotheses $\Gamma$ (\Cref{fig:inference}). Although intuitive and natural, logical rules seem rather arbitrary and for that reason, attempts at justifying them were made in philosophy \cite[Chapter 8]{dummett1991logical}.

In the end of the 20th century, the so-called Curry-Howard correspondence was first discovered by Curry \cite{curry} but then clearly stated by Howard \cite{howard}, during the rise of computer science. It shows a formal correspondence between proofs and programs but also between formulas and types in programming for some logical systems and some (functional) typed programming systems. A common illustration is the correspondence between natural deduction restricted to the implication (also called \emph{implicative minimal logic}) and the simply typed $\lambda$-calculus \cite{church1940formulation}. The mysterious rules of logic were then given a computational meaning. For instance, the inference rule of \emph{modus ponens} corresponds to the typing of function application (\cf \Cref{fig:implication}).

\begin{figure}
    \begin{prooftree}
        \hypo{\Gamma, A, B \vdash C}
        \infer1[$\land$L]{\Gamma, A \land B \vdash C}
    \end{prooftree}
    \hspace{1cm}
    \begin{prooftree}
        \hypo{\Gamma \vdash A}
        \hypo{\Gamma \vdash B}
        \infer2[$\land$R]{\Gamma \vdash A \land B}
    \end{prooftree}
    \caption{Two inference rules for the conjunction in sequent calculus. The first one states that we can combine two hypotheses $A$ and $B$ to build the hypothesis $A \land B$ and the second one that proving $A$ and $B$ gives a proof of $A \land B$.}
    \label{fig:inference}
\end{figure}

\begin{figure}
    \begin{prooftree}
        \hypo{\Gamma \vdash A}
        \hypo{\Delta \vdash A \Rightarrow B}
        \infer2[MD]{\Gamma, \Delta \vdash B}
    \end{prooftree}
    \hspace{2cm}
    \begin{prooftree}
        \hypo{\Gamma \vdash a : A}
        \hypo{\Delta \vdash f : A \rightarrow B}
        \infer2[app]{\Gamma, \Delta \vdash f(a) : B}
    \end{prooftree}
    \caption{Inference rule for modus ponens and the typing of function application where $\Gamma$ is seen as a typing environment. The upper part corresponds to premises and the bottom to the conclusion. Giving an argument $a$ of type $A$ to a function $f$ turning an element of type $A$ to an element of type $B$ indeed produces an element $f(a)$ of type $B$. }
    \label{fig:implication}
\end{figure}

Although mathematical proofs are naturally thought to be purely static objects, this correspondence between proofs and programs shows that they also have a computational or dynamic aspect. The \emph{cut rule} of sequent calculus, defined as follows\footnote{Notice that the cut rule is very close to modus ponens.}:
\[\begin{prooftree}
\hypo{\Gamma \vdash A}
\hypo{\Delta, A \vdash C}
\infer2[cut]{\Gamma, \Delta \vdash C}
\end{prooftree}\]
represents the use of intermediate lemma in a proof. The expression $\Gamma \vdash A$ states that if a statement $A$ is provable from the hypotheses in $\Gamma$ and that $A$ together with some hypotheses $\Delta$ lead to some conclusion $C$ then we can have a ``shortcut" stating that $C$ is a consequence of both $\Gamma$ and $\Delta$. Although essential in mathematical practice, Gentzen shows that it can, in fact, be removed without any loss of meaning. Similarly to how we can \emph{inline} the code of function calls in the main body of a program, Gentzen's cut-elimination theorem shows that there exists a procedure of explicitation of proofs which inlines lemmas within proofs. By the Curry-Howard correspondence, this corresponds to a logical counterpart of program execution: proofs are dynamic entities.

\paragraph{The implicit operations in reasoning and their structure} Inspired by the semantics of $\lambda$-calculus \cite{girard1988normal} and its role in the Curry-Howard correspondence, linear logic \cite{linearlogic} was introduced by Girard as a refinement of intuitionistic logic (the underlying logic of Gentzen's natural deduction). Linear logic can be simply presented as a sequent calculus with a control and explicitation over the duplication and erasure of formulas, making formulas handled as sort of limited resources. In particular, the famous decomposition of the implication presents implication are a composition of two operations: $A \Rightarrow B$ becomes $!A \multimap B$ where $!A$ allows for an arbitrary use of $A$ and $\multimap$ is a linear implication using its premise exactly once.
On top of that, linear logic enjoys a nice involutive linear negation breaking the separation between hypothesis and conclusion. This allows a nice compact sequent calculus which will be presented in \ref{sec:proofs}.

Apart from defining a sequent calculus for linear logic, Girard was led to introduce an alternative syntax akin to natural deduction for linear logic: \emph{proof-nets} (\Cref{fig:proofnet}), which exhibit a non-sequential structure in proofs \cite{girard1996proof}. Initially seen as a mere syntactic convenience, it led to a new understanding of proof theory with a fine-grained analysis of the mathematical meaning of proofs and of how they relate to computation. Danos \cite[Chapter 9-11]{danosPHD} and Regnier's \cite{regnierPHD} thesis, but also other developments of linear logic that came after \cite{di2003proof, accattoli2018proof} illustrate this analysis. Proof-nets are defined from more general alogical graphs called \emph{proof-structures} which constitute a model of computation by themselves (where the graphs are reduced with cut-elimination). It is then possible to assert whether a proof-structure corresponds to a proof-net, \ie is the translation of a sequent calculus proof, by using what we call a \emph{correctness criterion} (\cf\Cref{subsec:correctness}). 
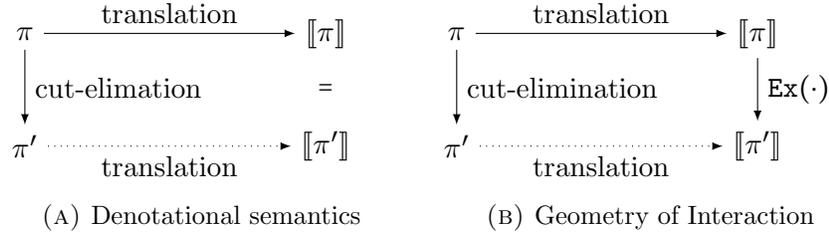
\begin{figure}
	\begin{subfigure}{0.35\textwidth}
		\begin{tikzpicture}
		    \node at (0, 0) (s) {$\pi$};
		    \node at (4, 0) (qs) {$\interp{\pi}$};
		    \node at (0, -1.5) (sp) {$\pi'$};
		    \node at (4, -1.5) (qsp) {$\interp{\pi'}$};
		    \draw[-latex] (s) -- (qs) node[midway, above] {translation};
		    \draw[-latex, dotted] (sp) -- (qsp) node[midway, below] {translation};
		    \draw[-latex] (s) -- (sp) node[midway, right] {cut-elimation};
		    \draw[draw=none] (qs) -- (qsp) node[midway] {$=$};
		\end{tikzpicture}
    	\caption{Denotational semantics}
    \end{subfigure}
    \hspace{0.15cm}
	\begin{subfigure}{0.38\textwidth}
		\begin{tikzpicture}
		    \node at (0, 0) (s) {$\pi$};
		    \node at (4, 0) (qs) {$\interp{\pi}$};
		    \node at (0, -1.5) (sp) {$\pi'$};
		    \node at (4, -1.5) (qsp) {$\interp{\pi'}$};
		    \draw[-latex] (s) -- (qs) node[midway, above] {translation};
		    \draw[-latex, dotted] (sp) -- (qsp) node[midway, below] {translation};
		    \draw[-latex] (s) -- (sp) node[midway, right] {cut-elimination};
		    \draw[-latex] (qs) -- (qsp) node[midway, right] {$\mathtt{Ex(\cdot)}$};
		\end{tikzpicture}
		\caption{Geometry of Interaction}
	\end{subfigure}
	\caption{We associate a proof $\pi$ to a mathematical object $\interp{\pi}$. In denotational semantics, we identify a proof $\pi$ and its cut-elimination $\pi'$ because we consider they have the same meaning, whereas in the GoI, they differ but are linked by computation. In particular, we are interested in simulating the computation linking them.}
	\label{fig:denotational}
\end{figure}

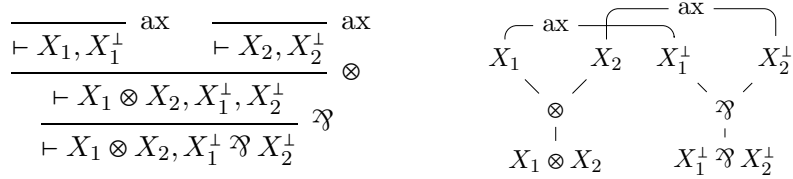
\begin{figure}
	\begin{minipage}{0.4\textwidth}
		\begin{prooftree}
		    \infer0[ax]{\vdash X_1, X_1^\bot}
		    \infer0[ax]{\vdash X_2, X_2^\bot}
		    \infer2[$\otimes$]{\vdash X_1 \otimes X_2, X_1^\bot, X_2^\bot}
		    \infer1[$\parr$]{\vdash X_1 \otimes X_2, X_1^\bot \parr X_2^\bot}
		\end{prooftree}
	\end{minipage}
	\begin{minipage}{0.4\textwidth}
	\scalebox{0.9}{\begin{tikzpicture}
      \node at (-0.75, 0.75) (a) {$X_1$};
      \node at (0.75, 0.75) (b) {$X_2$};
      \node at (0, 0) (tens) {$\otimes$};
      \node at (0, -0.75) (ab) {$X_1 \otimes X_2$};
      \draw[-] (a) -- (tens);
      \draw[-] (b) -- (tens);
      \draw[-] (tens) -- (ab);

      \node at (1.75, 0.75) (ad) {$X_1^\bot$};
      \node at (3.25, 0.75) (bd) {$X_2^\bot$};
      \node at (2.5, 0) (par) {$\parr$};
      \node at (2.5, -0.75) (adbd) {$X_1^\bot \parr X_2^\bot$};
      \draw[-] (ad) -- (par);
      \draw[-] (bd) -- (par);
      \draw[-] (par) -- (adbd);

      \node at (0, 1.25) (ax1) {ax};
      \node at (2, 1.5) (ax2) {ax};
      \draw[-, rounded corners=5pt] (ax1) -| (a);
      \draw[-, rounded corners=5pt] (ax1) -| (ad);
      \draw[-, rounded corners=5pt] (ax2) -| (b);
      \draw[-, rounded corners=5pt] (ax2) -| (bd);
    \end{tikzpicture}}
    \end{minipage}
    \caption{Example of sequent calculus proof and its corresponding proof-net in linear logic. Notice that the order of rules is forgotten and that we have a sort of ``parallel" \cite{girard1996proof} syntax for proofs.}
    \label{fig:proofnet}
  \end{figure}

\paragraph{A semantic-free space for logic} Geometry of Interaction\footnote{Although the expression ``Geometry of Interaction" often refers to methods of static execution of $\lambda$-terms by a token machine \cite{danos1999reversible, asperti1995paths} (inspired by a simplification of Girard's first GoI \cite{danos1995proof}), we only refer to Girard's original programme \cite{goi1, goi2, goi3, goi4, goi5, goi6} in this paper.} (GoI) \cite{goi0, goi1, goi2, goi3, goi4} was originally introduced as a mathematical analysis of cut-elimination for proof-nets. This gave rise to a dynamic semantics which became a major inspiration behind game semantics \cite{blass1992game, abramsky1994games,hyland2000full,abramsky2000full,abramsky2002geometry, haghverdi2006categorical}, distinguishing itself from denotational semantics which usually identify a proof and its reduction by cut-elimination, thus forgetting how they are computationally related, as explained in \Cref{fig:denotational} \cite[Section 1.1]{seiller2017interaction}.

Around the time of the fourth paper on the GoI \cite{goi4}, Girard concurrently introduced ludics \cite{girard2001locus} which instead of proof-nets, starts by forgetting all the inessential parts of sequent calculus in order to obtain very general and alogical objects called \emph{designs}. What the latest works on the GoI and ludics have in common is that they start from alogical computational objects from which proofs and formulas are \emph{defined}, and hence, are no more primitive. In particular, formulas are defined by set of computational objects from which we expect some shared common characteristics, and connectives are ways to construct new sets from other sets depending on how their objects interact with each other\footnote{This way of typing computational entities is called \emph{l'Usage} (the use) in Girard's terminology.}. In particular, the (linear) negation $\mathbf{A}^\bot$ is the set of all objects which interact well with the objects of $\mathbf{A}$, which respects to a binary orthogonality relation $\perp$ opposing objects. Orthogonality relations should be understood as \emph{point of view} on interaction, deciding what a \emph{good} interaction is. Formulas are constructed by \emph{interactive typing} which is reminiscent of realisability interpretations \cite{krivine2009interactive, rieg2014forcing} and Riba's reconstruction of simple types from the untyped $\lambda$-calculus \cite[Section 3]{riba2007strong}.

In some sense, this leads to a sort of \emph{semantic-free} approach to logic since the meaning of formulas is no more defined by an external semantics but rather by computational interaction between objects in a model of computation chosen as a ground for logic. The meaning of objects become their \emph{possible uses}, which is \emph{internal} since related to how they are shaped. This can be illustrated by a comparison between ludics and Schütte's completeness proof \cite{schutte1956system, basaldella2009meaning}: in the latter, we either have a proof (in the syntactic world) of a statement or we can construct a counter-model (in the semantic world) of its negation which refutes it, while in the former, we have a \emph{counter-proof} (in the syntactic world) instead. Girard describes this situation as a \emph{monism}\footnote{In opposition to \emph{dualism}.}, meaning that logic lives in a self-contained homogeneous space where there is nothing but syntactic interaction. In terms of programming, it is like having both a program and its environment expressed as interactive entities of the same kind. This new framework should provide further developments of ideas presented by Curien \cite{curien2003symmetry} and Abramsky \cite{abramsky2008information}, which were already present in linear logic.

This idea of semantic-free typing allows us to define formulas which would be more general than the formulas of usual proof theory in the sense that the space of proofs of the conjunction $A \land B$ would be larger than in usual proof theory. Unusual computational entities may constitute a proof of $A \land B$. Hence, formulas become descriptions of computational behaviours in a chosen computational space.

\paragraph{Sufficient conditions for effective reasoning} This interactive interpretation of logic starts by choosing a model of computation but not all choices are equal. Several GoI models were defined using operator algebras \cite{goi1,goi2,goi4,goi5}, term unification algebras \cite{goi3}, graphs \cite{seiller2012interaction,seiller2016interaction} and graphings \cite{seiller2017interaction,seiller2019interaction,seiller2016interaction2}.
Although all these models did define rich models, that were in particular used to study computational complexity \cite{baillot2001elementary,aubert2016characterizing,aubert2016logarithmic,seiller2018interaction,seiller2020probabilistic}, they had two main drawbacks.
\begin{enumerate}
    \item The objects used to interpret even the most basic proofs were most of the time unnatural and infinite (as in Girard's initial interpretation using operator algebras \cite{goi1}). In particular, finite reasoning is no more possible\footnote{This fact is the starting point of the philosophical motivation of the transcendental syntax: the search for the conditions of possibility of (logical) language.} because in order to assert that an object $\Phi$ is a proof of $\mathbf{A}$ (which is a set), we need to \emph{test} it (by checking orthogonality) against all elements of $\mathbf{A}^\bot$ which may be infinite (because they represent all the potential partners of interaction). A solution would be to take inspiration from the correctness criterion of proof-nets which allows a finite and sufficient checking of logical correctness. Using term unification \cite{goi3}, a simplification with finite combinatorics has been suggested but was limited for a satisfying treatment of correctness criteria and fall into the same problems as for Seiller's interaction graphs \cite{seiller2012interaction,seiller2016interaction} which are less expressive as they could not naturally express the standard Danos-Regnier correctness criterion \cite{danos1989structure}. Therefore, we need to find an appropriate model of computation to start with.
    \item The obtained models did interpret soundly the fragments of linear logic considered, but no completeness results exist\footnote{While this aspect is a failure somehow, it is also a feature as the models are very rich and open other paths of reflection.}.
\end{enumerate}

Recently, Girard published a series of articles \cite{goi6, transyn1, transyn2, transyn3, transyn4} sketching the main lines of a new kind of model called \emph{transcendental syntax}. This kind of model would have the qualities of GoI models with more improvements. In particular, unlike previous models, the transcendental syntax allows for a satisfying treatment of correctness criterion for proof-nets which leads to completeness results by extending the model of flows \cite{goi3}. Inspired by the correctness criterion, the new idea introduced by the transcendental syntax is that in order to check whether a computational object $\Phi$ is a proof of $\mathbf{A}$, we only test it against a primitive finite set $\mathbf{a} \subseteq \mathbf{A}^\bot$ which is sufficient to ensure what we wish for \footnote{This way of typing computational entities is called \emph{l'Usine} (the factory) in Girard's terminology.}. In some sense, it introduces in logic a notion of testing close to the usual program testing of software engineering. Once the objects are certified, we need to mathematically justify that the tests ensure a sound use, a property Girard calls \emph{adequacy} which is similar to the adequacy property in Krivine realisability.

The main problem is that these articles are too inexact in form to serve satisfactorily as the basis of a mathematical theory\footnote{The formulation is borrowed from Church's criticism of von Mises' notion of kollektiv \cite{church1940}.}. The current work is the first step towards a proper formal account of the model.

\subsection{Contributions of the paper}

The contributions of this paper are the following:
\begin{itemize}
    \item In \Cref{sec:stellar}, we formally describe a model of computation we call ``stellar resolution". It is based on Robinson's first-order resolution \cite{robinson1965machine} and extends/corrects the model of computation vaguely described in Girard's transcendental syntax. In \Cref{subsec:propexec}, we prove the main properties of the model and state its Turing-completeness (Proposition~\ref{prop:turing}). In particular, while Girard claimed the failure of the Church-Rosser property for stellar resolution, we show it holds under some conditions (Theorem~\ref{thm:confluence}). We also relate the dynamics of our model to the construction of tiling-based computation (\cf \Cref{subsec:tiles2prolog}).

    \item In \Cref{sec:illustrations}, we encode few standard models of computation such as logic program and Turing machines. In \Cref{subsec:atam}, we relate our formalism to tiling-based models: Wang tiles \cite{wang} and the abstract tile assembly model \cite{winfree1998algorithmic,patitz2014introduction}, which has applications in DNA computing \cite{seeman,woods2019diverse}. In particular, the stellar resolution is seen as a very general model of computation describing a dynamic exchange of information within a hypergraph-like structure, thus subsuming automata and tile systems.

    \item We explain how our model captures the dynamic of cut-elimination for the multiplicative fragment of linear logic (MLL) in Theorem~\ref{thm:dynamics}, and the correctness criterion for proof-structures in Theorem~\ref{thm:correctness2}. This implicitly relates MLL proofs (hence the linear simply typed $\lambda$-calculus) to the models of computation we encode in \Cref{sec:illustrations}. We also correct some minor technical mistakes appearing in Girard's introducing paper on the transcendental syntax \cite{transyn1}.

    \item In \Cref{sec:formulas}, we show how MLL formulas can be defined only from the execution of the stellar resolution and a binary orthogonality relation $\perp$ opposing constellations (the objects of stellar resolution). This reconstruction of types suggests the possibility of speaking about type systems which could be applied to models such as automata, logic programs and tiling models as especially fine-grained specifications. Two typing methods are presented: types as labels in \Cref{subsec:testing} defining types as set of computational entities passing some finite tests and types as computational behaviours in \Cref{subsec:behaviours} representing types as sort of potentially infinite ideals. The section ends with a short discussion about multiplicative units in \Cref{subsec:units}.

    \item We prove soundness and completeness of the model \wrt both MLL (\cf Theorem~\ref{thm:soundness2} and \ref{thm:completeness2}) and MLL+MIX (\cf Theorem~\ref{thm:soundness} and \ref{thm:completeness}), an extension of MLL with the so-called MIX rule \cite{fleury1994mix}. A comment about Girard's adequacy property is given in \Cref{subsec:adequacy}.
\end{itemize}

\section{Stellar Resolution}\label{sec:stellar}

The stellar resolution is a new model of computation introduced in this paper as a computational ground for logic. For pedagogical purposes and for its proximity with these models, we present our model of computation as a tile system which can simulate logic programs by evaluating tilings and comment how it differs from other existing models appearing in the literature of logic programming at the end of \Cref{subsec:propexec}.

We recall definitions of terms, unification and resolution in \Cref{sec:unification} which are essential.

\subsection{From tile systems to logic programs}\label{subsec:tiles2prolog}

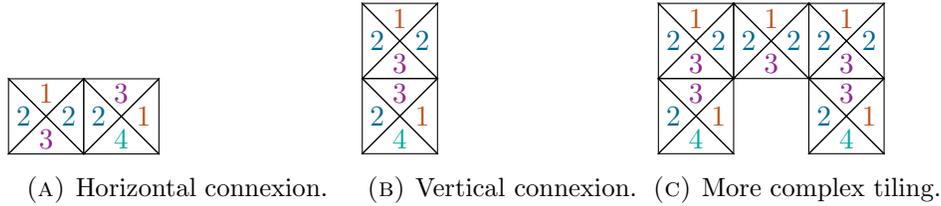
\begin{figure}
    \begin{subfigure}[b]{0.3\textwidth}
        \begin{tikzpicture}
            \wangn{0, 0}{{\color{Bittersweet}1}}{{\color{MidnightBlue}2}}{{\color{Plum}3}}{{\color{MidnightBlue}2}}
            \wangn{1, 0}{{\color{Plum}3}}{{\color{Bittersweet}1}}{{\color{Emerald}4}}{{\color{MidnightBlue}2}}
        \end{tikzpicture}
        \caption{Horizontal connexion.}
        \label{subfib:wang1}
    \end{subfigure}
    \begin{subfigure}[b]{0.25\textwidth}
        \begin{tikzpicture}
            \wangn{0, 0}{{\color{Bittersweet}1}}{{\color{MidnightBlue}2}}{{\color{Plum}3}}{{\color{MidnightBlue}2}}
            \wangn{0, -1}{{\color{Plum}3}}{{\color{Bittersweet}1}}{{\color{Emerald}4}}{{\color{MidnightBlue}2}}
        \end{tikzpicture}
        \caption{Vertical connexion.}
        \label{subfig:wang2}
    \end{subfigure}
    \begin{subfigure}[b]{0.25\textwidth}
        \begin{tikzpicture}
            \wangn{0, 0}{{\color{Bittersweet}1}}{{\color{MidnightBlue}2}}{{\color{Plum}3}}{{\color{MidnightBlue}2}}
            \wangn{0, -1}{{\color{Plum}3}}{{\color{Bittersweet}1}}{{\color{Emerald}4}}{{\color{MidnightBlue}2}}
            \wangn{1, 0}{{\color{Bittersweet}1}}{{\color{MidnightBlue}2}}{{\color{Plum}3}}{{\color{MidnightBlue}2}}
            \wangn{2, 0}{{\color{Bittersweet}1}}{{\color{MidnightBlue}2}}{{\color{Plum}3}}{{\color{MidnightBlue}2}}
            \wangn{2, -1}{{\color{Plum}3}}{{\color{Bittersweet}1}}{{\color{Emerald}4}}{{\color{MidnightBlue}2}}
        \end{tikzpicture}
        \caption{More complex tiling.}
        \label{subfig:wang3}
    \end{subfigure}
    \caption{Tiling with two wang tiles.}
    \label{fig:wang}
\end{figure}

We start with the simple and common Wang tiles \cite{wang} which are very intuitive and present generalisations leading to our model of stellar resolution.

Wang tiles are square bricks with four sides containing a value (usually an integer or a colour when presented graphically). We can construct tilings by placing copies of tiles on the plane $\rel^2$. Two tiles can be connected along two opposite sides when they hold the same value. We give three examples of partial tilings in \Cref{fig:wang}. We usually require the tilings to be \emph{maximal} (it cannot be extended with further copies of a tile) and \emph{connected}. Moreover, the rotation of tiles is usually forbidden (although this condition can be relaxed in some cases).

A possible generalisation is the model of flexible tiles \cite{jonoska2005computational} used as a model for DNA computing with branched junction molecules \cite{chen1991synthesis}. It works with star-like tiles having flexible arms which contain values from a given set $H$. An involution\footnote{A function $f$ is an involution (or is involutive) when $f(f(x)) = x$} $\op : H \rightarrow H$ such that $\op(h) \neq h = \op(\op(h))$ defines complementary values called \emph{Watson-Crick complementaries}. Two flexible arms can be connected if they have complementary values (\cf\Cref{fig:flexible}). Notice that the model is not limited to planarity, unlike Wang tiles. Surprisingly, this model can actually encode Wang tiles or other ``planar" tiling-based models \cite{jonoska2006flexible}.

\begin{figure}
    \begin{tikzpicture}
        \node[circle,draw=black] at (0, 0) (c1) {$t_1$};
        \node[dot, label={180:$h_1$}] at (-1, 0) (b1) {};
        \node[dot, label={90:$h_2$}] at (1, 0) (b2) {};
        \node[dot, label={180:$h_3$}] at (0, -1) (b3) {};
        \draw (c1) to[in=-30, out=150] (b1);
        \draw (c1) to[out=-30, in=150] (b2);
        \draw (c1) to[out=220, in=30] (b3);

        \node[circle,draw=black] at (5, 0) (c2) {$t_2$};
        \node[dot, label={90:$\mathtt{op}(h_2)$}] at (4, 0) (a1) {};
        \node[dot, label={0:$h_4$}] at (6, 0) (a2) {};
        \draw (c2) to[in=-30, out=150] (a1);
        \draw (c2) to[out=-30, in=150] (a2);
        \draw[dashed] (b2) to[out=-30, in=150] (a1);
    \end{tikzpicture}
    \caption{Tiling of flexible tiles connected by two complementary flexible arms.}
    \label{fig:flexible}
\end{figure}
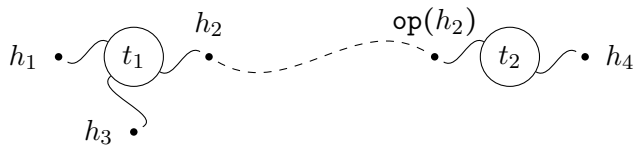

It is possible to generalise even more by considering polarised terms with a head symbol (\eg $+c(X)$ but not $+X$) as values for flexible arms such that two terms can be connected when they are unifiable up to renaming, \ie they can be made equal by a substitution from variables to terms (\cf\Cref{sec:unification}). It is more general than flexible tiles because terms can encode any set and term unification up to renaming potentially involves several possible partners. For instance, $+c(X)$ can be connected with $-c(t)$ for any term $t$ since the substitution $\theta = \{X \mapsto t\}$ is a solution of the equation $X \eqq t$ (with a renaming of $t$ in order to avoid variable conflicts).

These flexible tile sets with terms are equivalent to first-order formulas in prenex conjunctive normal form \cite{robinson1965machine}: flexible arms are first-order atoms $P(t), \lnot P(t)$, tiles are disjunctive clauses $\{A_1, ..., A_n\}$ with bound variables and tile sets are conjunctions of clauses. Robinson's first-order resolution \cite{robinson1965machine} induces a procedure of evaluation of tilings by successive contraction of links. For instance, the two connected clauses of \Cref{fig:resolution} merge and produce the new clause $[g(f(Y)), -b(f(Y)), +c(Y)]$ where the solution $\theta$ of the equation $a(X) \eqq a(f(Y))$ associated to the link is propagated to the neighbours. The evaluation of a whole set of clauses is defined as follows:
\begin{enumerate}
    \item construct all possible connected tilings by connecting clauses along unifiable terms of opposite polarity (non-determinism can happen);
    \item eliminate the unwanted ones (typically, we need maximal tilings which cannot be extended);
    \item try to contract all the tilings by using Robinson's resolution rule:
    \begin{itemize}
        \item if it fails, throw the tiling away;
        \item it if works, it should give you a new star;
    \end{itemize}
    \item at the end, you obtain a new constellation with stars coming from all the contracted tilings.
\end{enumerate}

It corresponds to the computation of all new clauses we can infer from the available ones. At this point, our flexible tile system with terms is indeed a fancy way to present first-order resolution.

\begin{figure}
    \begin{tikzpicture}
        \node at (0, 0) (a1) {$[g(X), {\color{Bittersweet}+a(X)}, -b(X)]$};
        \node at (0, -1.5) (b1) {$[{\color{MidnightBlue}-a(f(Y))}, +c(Y)]$};
        \draw (a1.south) -- (b1.110) node[midway, left] {$\theta = \substseq{X \mapsto {\color{Plum}f(Y)}}$};

        \node at (2.5, -0.75) (arrow) {$\leadsto$};

        \node at (6, -0.75) (result) {$[g({\color{Plum}f(Y)}), -b({\color{Plum}f(Y)}), +c(Y)]$};
    \end{tikzpicture}
    \caption{Robinson's first-order resolution seen as a tiling model. The two terms $+a(X)$ and $-a(f(Y))$ are dual and then can be connected to form a tiling. }
    \label{fig:resolution}
\end{figure}

The model we call stellar resolution is a graph-theoretic variant of the above tile system where we allow additional features such as unpolarised terms and internal polarised terms. It is possible because we are not interested anymore in the logical meaning of our constructions.

The difference between logic programming and our approach is that our model is used with different motivations and with less constraints. In particular, it is alogical so it does not follow logical rules. More details about the stellar resolution and approaches of logic programming are detailed at the end of \Cref{subsec:propexec}.

\subsection{Stars and constellations}\label{subsec:stars}

We use Girard's terminology of \emph{stars} and \emph{constellations} \cite{transyn1}. The tiles are called \emph{stars} and their flexible arms are \emph{rays}. A tile set is called a \emph{constellation}. Rays can contain special polarised symbols called \emph{colours} which are analogous to the colours of Wang tiles. For instance, if $f$ is a symbol, then $+f$ and $-f$ are two \emph{dual} colours. We then expect terms such as $+c(f(X))$, $f(X)$, $Y$, $+d(X, -e)$ and $+c(f(+f(X), Y))$ to be rays. The point lies in the technical ability to switch colours and play with the potential connexions of constellations, \eg turning all colours $-c$ into $-d$.

\begin{defi}[Coloured signature]
\label{def:colouredsignature}
A \emph{coloured signature} is a tuple $\mathbb{C} = (V, C, F, \ar, \op, \floor{\cdot})$ where $V$ is a countable set of variables, $C$ and $F$ are disjoint countable set of function symbols such that $C$ is called the set of \emph{colours} and $\ar : C \uplus F \rightarrow \nat$ associates an arity to function symbols. Colours in $C$ are new function symbols $+f, -f$ constructed by juxtaposing a polarity in $\{-, +\}$ and a function symbol $f \in F$, and $\op$ is an involution defined by $\op(+f) = -f$. The \emph{underlying symbol} $\floor{c}$ of a colour $c \in C$ such that $c \in \{+f, -f\}$ is defined by $\floor{+f} = \floor{-f} = f$ with $f$.
\end{defi}

We assume the existence of a coloured signature $\mathbb{C} = (V, C, F, \ar, \op, \floor{\cdot})$ unless we explicitly use a specific one.

\begin{figure}
    \begin{tikzpicture}
        \node[circle,draw=black] at (0, 0) (s) {$\phi$};
        \node[MidnightBlue] at (-3, 0.5) (i1) {$-c_1(t_1^1, ..., t_n^1)$};
        \node[MidnightBlue] at (-3, 0) (i2) {$\hdots$};
        \node[MidnightBlue] at (-3, -0.5) (i3) {$-c_k(t_1^k, ..., t_m^k)$};
        \draw[-latex,MidnightBlue] (i1) to[in=170, out=0] (s);
        \draw[-latex,MidnightBlue] (i3) to[in=190, out=0] (s);
        \node[Bittersweet] at (3, 0.5) (o1) {$+d_1(u_1^1, ..., u_{n'}^1)$};
        \node[Bittersweet] at (3, 0) (o2) {$\hdots$};
        \node[Bittersweet] at (3, -0.5) (o3) {$+d_{k'}(u_1^{k'}, ..., u_{m'}^{k'})$};
        \draw[-latex,Bittersweet] (s) to[out=10, in=180] (o1);
        \draw[-latex,Bittersweet] (s) to[out=-10, in=180] (o3);
        \node at (-1, -1.5) (u1) {$v_1$};
        \node at (0, -1.5) (u2) {$\hdots$};
        \node at (1, -1.5) (u3) {$v_l$};
        \draw (u1) to[out=90, in=-100] (s);
        \draw (u3) to[out=90, in=-80] (s);
    \end{tikzpicture}
    \caption{Star with prefixed rays seen as either input, output or unpolarised. For general stars, internal colours are allowed as well.}
    \label{fig:star}
\end{figure}
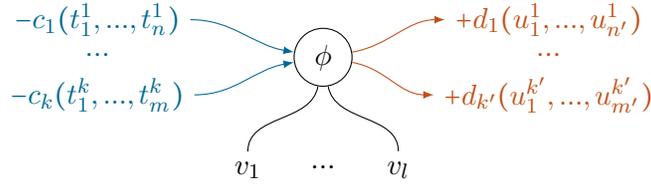

\begin{defi}[Rays]
\label{def:ray}
A \emph{ray} on a signature $\mathbb{C} = (V, C, F, \ar, \op, \floor{\cdot})$ is a term $r \in \terms{\mathbb{C}}$ constructed with variables in $V$ and function symbols in $C \uplus F$ (\cf\Cref{sec:unification}).

A ray $r$ is \emph{coloured} if there is a function symbol $f$ appearing in $r$ such that $f \in C$ and it is \emph{uncoloured} otherwise.

The underlying term of a ray is defined inductively as follows:
\[\floor{x} = x \qquad \floor{c(r_1, ..., r_n)} = \floor{c}(\floor{r_1}, ..., \floor{r_n}) \qquad \floor{f(r_1, ..., r_n)} = f(\floor{r_1}, ..., \floor{r_m})\]
with $x \in V, f \in F$ and $c \in C$.
\end{defi}

Notice that the model is more complex that one would expect from the intuitions given in \Cref{subsec:tiles2prolog} because of the possible presence of internal colours. Actually, in this paper, only rays with colours as prefix will be used, hence rays $c(t_1, ..., t_n)$ such that $c \in C$ and the terms $t_i$ are uncoloured. Such terms looks like atomic first-order formulas $P(t), \lnot P(t)$ where $P$ is a predicate. Although not used here, we choose the keep the more general definition of rays with internal colours in order to anticipate further works and extensions already described by Girard \cite[Section 4.1]{transyn3}.

\begin{defi}[Star, \Cref{fig:star}]
\label{def:star}
A star $\phi$ over a coloured signature $\mathbb{C}$ is a finite indexed family\footnote{Which should be understood as an array indexed by a given set of indexes (typically, natural numbers).} of \emph{rays}, \ie a finite set of indexes $I_\phi$ together with a map $\phi[\cdot]:I_\phi\rightarrow\idrays{\mathbb{C}}$ which given an index gives the corresponding ray. The set of variables appearing in $\phi$ is defined by $\vars{\phi} := \bigcup_{i \in I_\phi} \vars{\phi[i]}$ (\cf\Cref{sec:unification}). For convenience, stars will be written as a clause $[r_1, ..., r_n]$.

The empty star is written $\emptystar$ and is defined by $I_{\emptystar} = \emptyset$.
\end{defi}

\begin{nota}[Substitutions extended to stars]
A \emph{substitution} (\cf\Cref{sec:unification}) is a function replacing variables by terms, within a term.
Given a substitution $\theta$, its action extends to stars by $\theta[r_1, ..., r_n] = [\theta r_1, ..., \theta r_n]$.

A \emph{renaming} is a substitution replacing variables by other variables.
We say that two stars $\phi_1, \phi_2$ are \emph{$\alpha$-equivalent}, written $\phi_1 \alphaeq \phi_2$, when there exists a renaming $\alpha$ such that $\alpha \phi_1 = \phi_2$.
\end{nota}

\begin{conv}
In this paper, stars will be considered up to $\alpha$-equivalence. We therefore define $\stars{\mathbb{C}}$ as the set of all stars over a coloured signature $\mathbb{C}$, quotiented by $\alphaeq$.
\end{conv}

\begin{defi}[Constellation]
\label{def:constellation}
A \emph{constellation} $\Phi$ is a countable indexed family of stars, \ie a countable (possibly infinite) set $I_\Phi$ together with a map $\Phi[\cdot]: I_\Phi \rightarrow \stars{\mathbb{C}}$. For convenience, a finite constellation will be written as a sum of stars $\Phi = \phi_1 + ... + \phi_n$.

We define the set of rays of a constellation $\Phi$ by $\idrays{\Phi} = \{ (i, j) \mid i \in I_\Phi, j \in I_{\Phi[i]}\}$ (we keep track of the star from which rays come) and $\idrayspol{\Phi} := \{r \in \idrays{\Phi} \mid r \text{ is coloured}\}$ by its restriction to coloured rays.

The empty constellation is written $\emptyset$ and is defined by $I_\emptyset = \emptyset$.
\end{defi}

Constellations are meant to be sort of programs. As in logic programming (\eg Prolog) or functional programming (\eg $\lambda$-calculus), variables will be considered bound to their star (which can be seen as sort of declarations), hence the two $x$ in $[+f(x)]+[-f(x), y]$ are unrelated. This is similar to how the two $x$ in the $\lambda$-term $\lambda x. (\lambda x. M)$ are different.

\begin{nota}[Indexed set membership]
We will sometimes write $e \in E$ for an indexed set $E$ (a star or constellation in our case). The intended meaning is that there is an $i \in I_E$ such that $e = E[i]$. 
\end{nota}

Now that all the elementary objects of the stellar resolution are defined, we give a very standard encoding of natural numbers which will be useful through the paper and illustrate the model by some example of constellations.

\begin{defi}[Encoding of natural numbers]
We define the function symbol $\enat{n}$ for a natural number $n \in \nat$ by $\enat{0} = 0$ and $\enat{n+1} = s(\enat{n})$ for a unary symbol $s$ and a constant $0$.
\end{defi}

\begin{exas}
\label{ex:constellation}
We give examples of finite and infinite constellations:
\begin{itemize}
    \item $\Phi_{BH}^c := [-c(X), +c(X)]$;
    \item $\Phi_\nat^+ := [+add(\enat{0}, Y, Y)]+[-add(X, Y, Z), +add(s(X), Y, s(Z))]$ (logic program for addition);
    \item $\Phi_\nat^{n+m} := \Phi_\nat^++[-add(\enat{n}, \enat{m}, R), R]$ (query for the computation of $n+m$);
    \item $\Phi_\nat$ is defined by $I_{\Phi_\nat} = \nat$ and $\Phi_\nat[i] := [-nat(\enat{i}), +nat(\enat{i+1})]$ (infinite chain)
\end{itemize}
over the signature defined by the variables $V = \{X, Y, Z, R\}$, the colours $C = \{\pm add, \pm nat\}$, and $F = \{add, nat, s, 0\}$, $\ar(add) = 3, \ar(nat) = \ar(s) = 1, \ar(0) = 0$. The constellation $\Phi_\nat^{n+m}$ corresponds to the following Horn clauses \cite{tarnlund1977horn} where $Add(X, Y, Z)$ states that $X+Y = Z$: $Add(0, Y, Y)$ and $Add(X, Y, Z) \Rightarrow Add(s(X), Y, s(Z))$.
\end{exas}

\begin{nota}[Disjoint union of constellations]
Let $\Phi_1$ and $\Phi_2$ be two constellations. Their disjoint union $\Phi_1 \uplus \Phi_2$ is a constellation defined by $I_{\Phi_1 \uplus \Phi_2} := I_{\Phi_1} \uplus I_{\Phi_2}$ and the associated copairing $\Phi_1[\cdot] \uplus \Phi_2[\cdot] : I_{\Phi_1} \uplus I_{\Phi_2} \rightarrow \stars{\mathbb{C}}$.
\end{nota}

\subsection{Evaluation of diagrams and execution of constellations}\label{subsec:eval}

We are now interested in the formation of tilings we call \emph{diagrams}. Unlike tilings with Wang tiles or flexible tiles, it is possible to evaluate these diagrams by contracting them with Robinson's resolution rule \cite{robinson1965machine}.

\begin{figure}
    \begin{tikzpicture}
        \node[state, initial] (q0) {$\mathtt{q_0}$};
        \node[state, right of=q0, xshift=20mm] (q1) {$\mathtt{q_1}$};
        \node[state, accepting, right of=q1, xshift=20mm] (q2) {$\mathtt{q_2}$};
        \draw[-latex]
        (q0) edge[loop above] node{1} (q0)
        (q2) edge[loop above] node{0} (q2)
        (q0) edge[above] node{0} (q1)
        (q1) edge[above, bend right] node{1} (q0)
        (q2) edge[below, bend left] node{1} (q0)
        (q1) edge[above] node{0} (q2);

        \node[state, fill=lightgray, below of=q0, yshift=-20mm] (w1) {$1$};
        \node[state, fill=lightgray, right of=w1, xshift=10mm] (w2) {$0$};
        \node[state, fill=lightgray, right of=w2, xshift=10mm] (w3) {$1$};
        \node[state, fill=lightgray, right of=w3, xshift=10mm] (w4) {$0$};
        \node[state, fill=lightgray, right of=w4, xshift=10mm] (w5) {$0$};
        \draw[-latex]
        (w1) edge node{} (w2)
        (w2) edge node{} (w3)
        (w3) edge node{} (w4)
        (w4) edge node{} (w5);
        \draw[-latex, MidnightBlue]
        (w1) edge node{} (q0)
        (w2) edge node{} (q1)
        (w3) edge node{} (q0)
        (w4) edge node{} (q1)
        (w5) edge node{} (q2);
    \end{tikzpicture}
    \caption{Example of finite deterministic automata with a mapping from a word graph to its state graph.  }
    \label{fig:automatapath}
\end{figure}
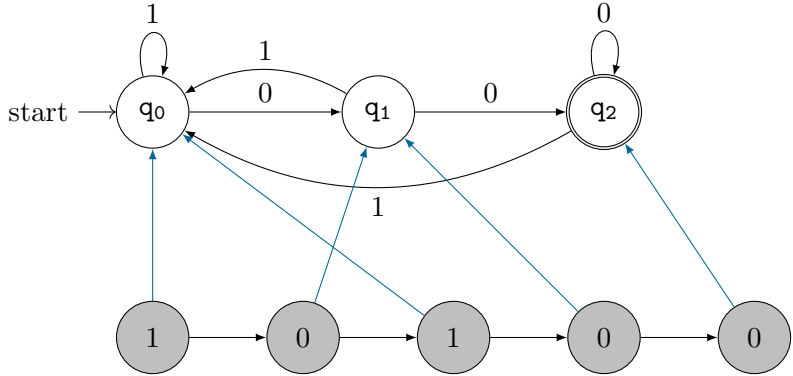

We first define the \emph{dependency graph} of a constellation which defines the allowed connexions between stars along dual rays. A diagram corresponds to an actual plugging of stars along dual rays, following those allowed connexions. The edges linking stars will induce an equation between terms and the whole diagram will induce a unification problem (\cf \Cref{sec:unification}). The \emph{evaluation} of a diagram will correspond to solving its associated unification problem and producing a new star.

In order to approach this idea more intuitively, we explain a common occurrence of it in automata theory. A finite deterministic automaton is a machine reading an input word character by character. It starts from an initial state and moves from a state to another accordingly to the current character it reads. If it ends on the final state, it \emph{accepts} the input word. Otherwise, it rejects the input. An example of automata of final state $q_2$ is given in \Cref{fig:automatapath} (on the top with vertices $q_0, q_1, q_2$). A word can be represented as a linear graph (on the bottom of \Cref{fig:automatapath} with vertices $1, 0, 1, 0$ and $0$) and finally, the reading of a word can be represented as a mapping from characters to states (blue links).

The idea is that the state graph of the automata shows the allowed transitions (where loop can appear) and the word graph is a tiling of states or a traversal of graph which follows those allowed transitions (sometimes by \emph{unfolding} loops).

Our diagrams generalise this idea. The state graph corresponds to a dependency graph and the word graph corresponds to a diagram. The difference is that a dynamics of term unification is present in our dependency graphs and diagrams can be any graph, not necessarily limited to the linear case as for automata. Mathematically, a diagram will be associated to a graph homomorphism between a graph (representing the tiling) and the dependency graph, exactly like how word graphs are related to state graphs in automata theory.

\begin{defi}[Duality between rays]
\label{def:duality}
Let $\op(r)$ the inverse of a ray $r$ defined by inverting polarities, \ie $\op(r)$ is $r$ where all colours $c$ are replaced by $\op(c)$.

Two rays $r$ and $r'$ are \emph{dual} \wrt a set of colours $A \subseteq C$, written $r \match_A r'$, when both have at least one colour $c \in A$ and $\{r \eqq \op(r')\}$ has a solution.
\end{defi}

\begin{prop}
The relation $\match_A$ is symmetric but not reflexive nor transitive.
\end{prop}
\begin{proof}
Notice that only symmetric relations are used in the definition (equality and unification). Hence, it follows that duality is also symmetric. The failure of reflexivity and transitivity comes from the requirement of opposite colours. A ray cannot be dual to another ray where two identical polarities face each other.
\end{proof}

\begin{exa}
We have $+c(X) \match_{\{c\}} -c(0)$ and $-d(X) \match_{\{d\}} +d(f(X))$ but not $+c(X) \match_A f(Y)$ (presence of unpolarised ray), $+c(X) \match -d(X)$ (different head symbol), $+c(X) \match_A +c(f(Y))$ (polarities are not opposite) nor $+c(f(X)) \match_A -c(g(Y))$ (terms are not $\alpha$-unifiable) for any $A$.
\end{exa}

\begin{figure}
    \begin{subfigure}{0.9\textwidth}
        \centering
        \begin{tikzpicture}[every node/.style={scale=0.8}]
        	\node[circle,draw] (e1) at (0,0) {$\phi_1$};
        	\node[circle,draw] (e2) at (6,0) {$\phi_2$};
        	\node[circle,draw] (e3) at (12,0) {$\phi_3$};
            \node (r) at (12.5,0) {$R$};

        	\draw (e1) -- (e2) node[midway,above,sloped] {$+add(\enat{0}, Y, Y) \match -add(X,Y,Z)$};
            \draw (e2) to[loop] (e2) node[above=15mm] {$-add(X,Y,Z) \match +add(s(X), Y, s(Z))$};
        	\draw (e2) -- (e3) node[midway,above,sloped] {$+add(s(X), Y, s(Z)) \match -add(\enat{n},\enat{m},R)$};
        \end{tikzpicture}
        \subcaption{Dependency graph of $\Phi_\nat^{n+m}$.}
    \end{subfigure}
    \begin{subfigure}{0.8\textwidth}
        \centering
        \begin{tikzpicture}[x=2cm,y=1cm,every node/.style={scale=0.8}]
        	\node[circle,draw] (e1) at (0,0) {$n_1$};
        	\node[circle,draw] (e2) at (2,0) {$n_2$};
        	\node (dots) at (4,0) {$\dots$};

        	\draw (e1) -- (e2) node[midway,above] {$-nat(\enat{0}) \match +nat(\enat{1})$};
        	\draw (e2) -- (dots) node[midway,above] {$-nat(\enat{1}) \match +nat(\enat{2})$};
        \end{tikzpicture}
        \subcaption{Dependency graph of $\Phi_\nat$.}
    \end{subfigure}
    \caption{Examples of dependency graphs for constellations of Example~\ref{ex:constellation}.}
    \label{fig:dgraph}
\end{figure}
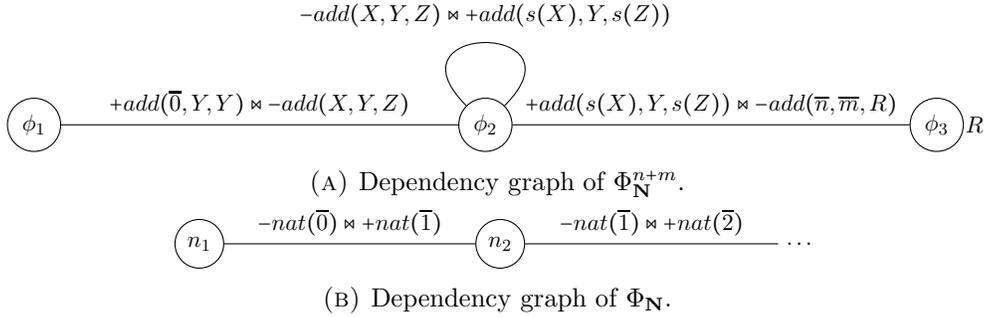

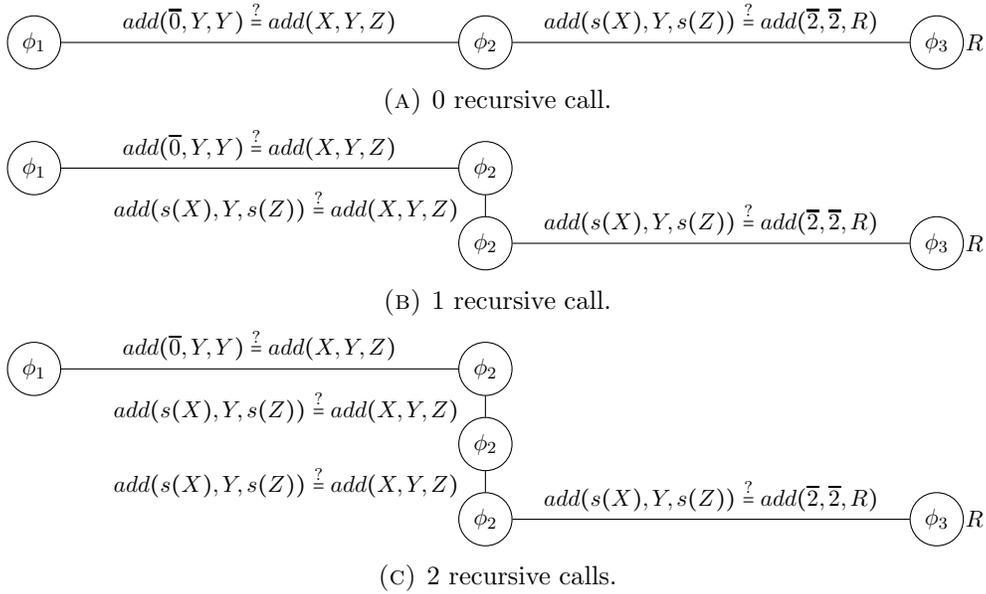
\begin{figure}
    \begin{subfigure}{0.9\textwidth}
        \centering
        \begin{tikzpicture}[every node/.style={scale=0.8}]
        	\node[circle,draw] (e1) at (0,0) {$\phi_1$};
        	\node[circle,draw] (e2) at (6,0) {$\phi_2$};
            \node[circle,draw] (e3) at (12,0) {$\phi_3$};
            \node (r) at (12.5,0) {$R$};

        	\draw (e1) -- (e2) node[midway, above] {$add(\enat{0}, Y, Y) \eqq add(X, Y, Z)$};
        	\draw (e2) -- (e3) node[midway, above] {$add(s(X), Y, s(Z)) \eqq add(\enat{2}, \enat{2}, R)$};
        \end{tikzpicture}
        \subcaption{0 recursive call.}
    \end{subfigure}

    \begin{subfigure}{0.9\textwidth}
        \centering
        \begin{tikzpicture}[every node/.style={scale=0.8}]
        	\node[circle,draw] (e1) at (0,0) {$\phi_1$};
        	\node[circle,draw] (e2) at (6,0) {$\phi_2$};
        	\node[circle,draw] (e2p) at (6,-1) {$\phi_2$};
        	\node[circle,draw] (e3) at (12,-1) {$\phi_3$};
            \node (r) at (12.5,-1) {$R$};

        	\draw (e1) -- (e2) node[midway, above] {$add(\enat{0}, Y, Y) \eqq add(X, Y, Z)$};;
        	\draw (e2) -- (e2p) node [left=3mm,midway] {$add(s(X), Y, s(Z)) \eqq add(X, Y, Z)$};
        	\draw (e2p) -- (e3) node[midway, above] {$add(s(X), Y, s(Z)) \eqq add(\enat{2}, \enat{2}, R)$};
        \end{tikzpicture}
        \subcaption{1 recursive call.}
        \label{subfig:fulldiagram}
    \end{subfigure}

    \begin{subfigure}{0.9\textwidth}
        \centering
        \begin{tikzpicture}[every node/.style={scale=0.8}]
        	\node[circle,draw] (e1) at (0,0) {$\phi_1$};
        	\node[circle,draw] (e2) at (6,0) {$\phi_2$};
        	\node[circle,draw] (e2p) at (6,-1) {$\phi_2$};
            \node[circle,draw] (e2pp) at (6,-2) {$\phi_2$};
        	\node[circle,draw] (e3) at (12,-2) {$\phi_3$};
            \node (r) at (12.5,-2) {$R$};

        	\draw (e1) -- (e2) node[midway, above] {$add(\enat{0}, Y, Y) \eqq add(X, Y, Z)$};;
        	\draw (e2) -- (e2p) node [left=3mm,midway] {$add(s(X), Y, s(Z)) \eqq add(X, Y, Z)$};
            \draw (e2p) -- (e2pp) node [left=3mm,midway] {$add(s(X), Y, s(Z)) \eqq add(X, Y, Z)$};
        	\draw (e2pp) -- (e3) node[midway, above] {$add(s(X), Y, s(Z)) \eqq add(\enat{2}, \enat{2}, R)$};
        \end{tikzpicture}
        \subcaption{2 recursive calls.}
    \end{subfigure}

    \caption{Examples of diagrams for the constellation $\Phi_\nat^{2+2}$. The number of occurrences of $\phi_2$ corresponds to the number of recursive calls. They correspond to unfolding of the loop of \Cref{fig:dgraph} corresponding to the possibility of recursive call. }
    \label{fig:diagrams}
\end{figure}

\begin{defi}[Dependency graph]
\label{def:dgraph}
The \emph{dependency graph} of a constellation $\Phi$ \wrt a set of colours $A \subseteq C$ is the undirected labelled multigraph\footnote{A multigraph is a graph with possibly several edges between two same vertices.} $\dgraph[A]{\Phi} := (V, E, \ell)$ where $V := I_\Phi$ and for each $(i, j), (i', j') \in \idrayspol{\Phi}$ such that $\Phi[i][j] \match_A \Phi'[i'][j']$, we have $\{i, i'\} \in E$ and the edge labelling $\ell(i, i') = (j, j')$. We simply write $\dgraph{\Phi}$ when links for all colours appearing in $\Phi$ are allowed.
\end{defi}

\begin{exa}
Two examples of dependency graphs for the two constellations $\Phi^{n+m}_\nat$ and $\Phi_\nat$ of Example~\ref{ex:constellation} are presented in Figure~\ref{fig:dgraph}.
\end{exa}

\begin{figure}
    \centering
    \begin{tikzpicture}[every node/.style={scale=0.8}]
        \node (r1) at (0,0) {$add(\enat{0}, Y, Y)$};
        \node (r2) at (4,0) {$add(X, Y, Z)$};
        \node (r3) at (8,0) {$add(s(X), Y, s(Z))$};
        \node (r4) at (12,0) {$add(\enat{2}, \enat{2}, R)$};

        \draw (r1) -- (r2);
        \draw (r3) -- (r4);
    \end{tikzpicture}
    \caption{Ray linking graph for the first diagram of \Cref{fig:diagrams} corresponding to the case with no recursive call.}
    \label{fig:rlg}
\end{figure}
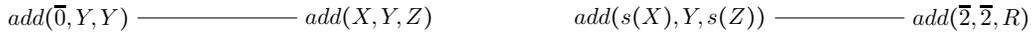

\begin{defi}[Diagram]
\label{def:diagram}
An $A$-\emph{diagram} (or simply \emph{diagram} when working with all colours) $\delta$ on a set of colours $A \subseteq C$ over a constellation $\Phi$ is a graph homomorphism $\delta : G_\delta \rightarrow \dgraph[A]{\Phi}$ from a non-empty finite connected multigraph $G_\delta$.

We define the set of rays of $\delta$ by $\idrays{\delta} := \{(\delta(x), j) \mid x \in V^{G_\delta}, j \in I_{\delta(x)}\}$ (we keep track of the stars from which the rays come) and extend the definition to the set $\idrayspol{\delta}$ of its coloured rays.

The labelling function $\ell$ of $\dgraph[A]{\Phi}$ is extended to a diagram $\delta$ by $\ell(e) := \ell(\delta(e))$ for $e \in E^{\dgraph[A]{\Phi}}$.

The \emph{ray linking graph} $\mathtt{RLG}(\delta)$ of $\delta$ is a graph showing how rays are linked (instead of stars). It is a graph $(V, E)$ defined with $V := \idrays{\delta}$ such that $(j, j') \in E$ when $j$ and $j'$ are linked in $G_\delta$, \ie when there is some $e \in E^{G_\delta}$ of label $(j, j')$.

Finally, we require that a diagram has a ray linking graph which is a simple graph (A graph without loop on vertices and without multiple edges between two vertices).  

The graph $G_\delta$ is considered up to renaming of the vertices and edges and for convenience, we will often have $V \subseteq \nat$ in practice.
\end{defi}

\begin{exa}
An example of three diagrams for the constellation $\Phi_\nat^{2+2}$ (which is an instance of the constellation $\Phi_\nat^{n+m}$ of Example~\ref{ex:constellation}) is given in Figure~\ref{fig:diagrams}. An example of ray linking graph for the first diagram is given in \Cref{fig:rlg}.
\end{exa}

\begin{nota}[Free rays and closed diagrams]
\label{not:diagram}
Given an $A$-diagram $\delta$ of a constellation $\Phi$, we define its set of \emph{free} (unconnected) rays $\freerays{\delta} \subseteq V^{\mathtt{RLG}(\delta)}$ by the set of isolated rays index in $\mathtt{RLG}(\delta)$. If $\freerays{\delta} = \emptyset$, we say that $\delta$ is \emph{closed}.
\end{nota}

We usually would like diagrams to be impossible to extend by connecting more stars, which corresponds to a notion of \emph{saturation}. In terms of tiling it is understood as the construction of the largest constructible tiling with occurrences of tiles from a given tile set and in terms of programming, it corresponds to a complete computation to be done.

\begin{figure}
	\begin{subfigure}{0.4\textwidth}
		\begin{tikzpicture}
			\node at (0, 0) (a) {$\dgraph[A]{\Phi}$};
			\node at (-4, 0.75) (b) {$H \subseteq G_{\delta'}$};
			\node at (-4, -0.75) (c) {$G_\delta$};
			\draw[-latex] (b) -- (a) node[midway, above] {$\delta'$};
			\draw[-latex] (c) -- (a) node[midway, below] {$\delta$};
			\draw (b) -- (c) node[midway, left] {$\overset{\varphi}{\simeq}$};
		\end{tikzpicture}
		\caption{Saturation expressed with graph homomorphisms.}
	\end{subfigure}
	\hspace{1cm}
	\begin{subfigure}{0.4\textwidth}
		\begin{tikzpicture}
			\node[circle,fill,inner sep=1pt] at (0, 0) (a) {};
			\node[circle,fill,inner sep=1pt] at (-0.5, -0.5) (b) {};
			\node[circle,fill,inner sep=1pt] at (0.5, -0.5) (c) {};
			\node[circle,fill,inner sep=1pt] at (-0.25, -1) (d) {};
			\node[circle,fill,inner sep=1pt] at (0.25, -1) (e) {};
			\draw (a) -- (c);
			\draw (b) -- (c);
			\draw (b) -- (d);
			\draw (c) -- (d);
			\draw (c) -- (e);
			
			\node at (1, -0.5) {$\sqsubseteq$};
			
			\node[circle,fill,inner sep=1pt] at (2, 0) (a) {};
			\node[circle,fill,inner sep=1pt] at (1.5, -0.5) (b) {};
			\node[circle,fill,inner sep=1pt] at (2.5, -0.5) (c) {};
			\node[circle,fill,inner sep=1pt] at (1.75, -1) (d) {};
			\node[circle,fill,inner sep=1pt] at (2.25, -1) (e) {};
			\draw (a) -- (c);
			\draw (b) -- (c);
			\draw (b) -- (d);
			\draw (c) -- (d);
			\draw (c) -- (e);
			\node[circle,fill=Emerald,inner sep=1pt] at (3, -0.5) (g) {};
			\node[circle,fill=Emerald,inner sep=1pt] at (1.5, 0) (h) {};
			\node[circle,fill=Emerald,inner sep=1pt] at (2.5, -1.25) (i) {};
			\draw[Emerald] (g) -- (a);
			\draw[Emerald] (g) -- (e);
			\draw[Emerald] (g) -- (c);
			\draw[Emerald] (a) -- (h);
			\draw[Emerald] (h) -- (b);
			\draw[Emerald] (h) -- (c);
			\draw[Emerald] (i) -- (c);
			\draw[Emerald] (i) -- (d);
		\end{tikzpicture}
		\caption{We have $\delta \sqsubseteq \delta'$ when $\delta'$ is an extension of $\delta$ with further links and possibly repetitions of stars. Both follow the connexions of a dependency graph $\dgraph[A]{\Phi}$.}
	\end{subfigure}
	\caption{Order $\sqsubseteq$ on diagrams representing an idea of saturation.}
	\label{fig:saturationorder}
\end{figure}
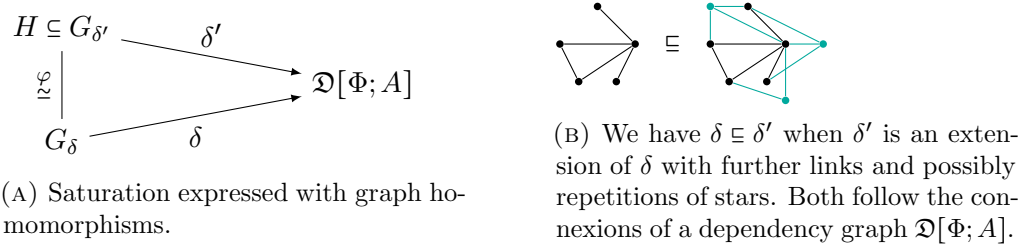

\begin{defi}[Saturated diagram]
\label{def:saturated}
We define a binary relation $\sqsubseteq$ (illustrated in \Cref{fig:saturationorder}) on $A$-diagrams over a constellation $\Phi$ by: $\delta \sqsubseteq \delta'$ if there exists an isomorphism $\varphi$ from a graph $H \subseteq G_{\delta'}$ to $G_\delta$ such that $\delta = \delta'\circ\varphi$. A maximal $A$-diagram \wrt $\sqsubseteq$ is called \emph{saturated}.
\end{defi}

\begin{prop}
The relation $\sqsubseteq$ is a preorder.
\end{prop}
\begin{proof}
We have $\delta \sqsubseteq \delta$ by taking the subgraph $G_\delta \subseteq G_\delta$. The isomorphism $\varphi$ is the identity function so we trivially have $\delta = \delta \circ \varphi$.

Assume we have $\delta_1 \sqsubseteq \delta_2$ and $\delta_2 \sqsubseteq \delta_3$. Hence, we have isomorphisms $\varphi_{1,2} : (H_2 \subseteq G_{\delta_2}) \simeq G_{\delta_1}$ and $\varphi_{2,3} : (H_3 \subseteq G_{\delta_3}) \simeq G_{\delta_2}$ such that $\delta_1 = \delta_2\circ\varphi_{1,2}$ and $\delta_2 = \delta_3\circ\varphi_{2,3}$. We can construct an isomorphism $\varphi_{1,2} \circ \varphi_{2,3}$ such that $\delta_1 = \delta_3 \circ \varphi$ and $G_{\delta_3}$ is indeed an extension of $G_{\delta_1}$ following the connexions of the same dependency graph. 
\end{proof}

Links in a diagram have an underlying equation. It follows that a whole diagram is associated to a unification problem (\cf\Cref{sec:unification}). A minor but important technical problem is that variables appearing in a constellation $\Phi$ are meant to be bound to their star. Hence, before evaluating, we must rename variables so to mark their membership to a star of $\Phi$. Fortunately, it is possible to define a canonical renaming by using the star indexes $I_\Omega$.

\begin{defi}[Underlying equation and problem]
\label{def:underproblem}
Let $\delta$ be an $A$-diagram of a constellation $\Phi$. We define a canonical family of renamings for variables defined by $\alpha_v(x) = x_v$ for $v \in V^{G_\delta}$ and any variable $x$.

The \emph{underlying equation} of a link $e = (v, v')$ of label $(j, j')$ in $E^{G_\delta}$ is defined by $\undereq{e} := \alpha_v\floor{\Phi[v][j]} \eqq \alpha_{v'}\floor{\Phi'[v'][j']}$ and the \emph{underlying problem} of $\delta$ is defined by
$\mathcal{P}(\delta) = \{ \undereq{e} \mid e \in E^{G_\delta} \}$.
\end{defi}

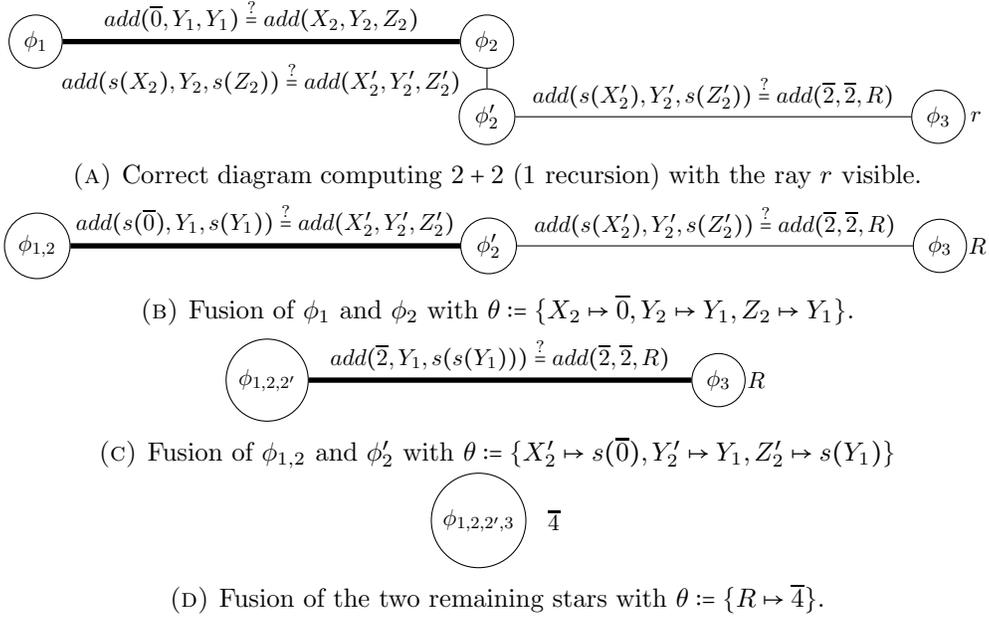
\begin{figure}
    \begin{subfigure}{0.9\textwidth}
        \centering
        \begin{tikzpicture}[every node/.style={scale=0.8}]
        	\node[circle,draw] (e1) at (0,0) {$\phi_1$};
        	\node[circle,draw] (e2) at (6,0) {$\phi_2$};
        	\node[circle,draw] (e2p) at (6,-1) {$\phi_2'$};
        	\node[circle,draw] (e3) at (12,-1) {$\phi_3$};
            \node (r) at (12.5,-1) {$r$};

        	\draw[line width=2pt] (e1) -- (e2) node[midway, above] {$add(\enat{0}, Y_1, Y_1) \eqq add(X_2, Y_2, Z_2)$};;
        	\draw (e2) -- (e2p) node[left=3mm,midway] {$add(s(X_2), Y_2, s(Z_2)) \eqq add(X_2', Y_2', Z_2')$};
        	\draw (e2p) -- (e3) node[midway, above] {$add(s(X_2'), Y_2', s(Z_2')) \eqq add(\enat{2}, \enat{2}, R)$};
        \end{tikzpicture}
        \subcaption{Correct diagram computing $2+2$ (1 recursion) with the ray $r$ visible.}
        \label{subfig:fusion0}
    \end{subfigure}

     \begin{subfigure}{0.9\textwidth}
        \centering
        \begin{tikzpicture}[every node/.style={scale=0.8}]
        	\node[circle,draw] (e1) at (0,0) {$\phi_{1,2}$};
        	\node[circle,draw] (e2p) at (6,0) {$\phi_2'$};
        	\node[circle,draw] (e3) at (12,0) {$\phi_3$};
            \node (r) at (12.5,0) {$R$};

        	\draw[line width=2pt] (e1) -- (e2p) node[midway, above] {$add(s(\enat{0}), Y_1, s(Y_1)) \eqq add(X_2', Y_2', Z_2')$};
        	\draw (e2p) -- (e3) node[midway, above] {$add(s(X_2'), Y_2', s(Z_2')) \eqq add(\enat{2}, \enat{2}, R)$};
        \end{tikzpicture}
        \subcaption{Fusion of $\phi_1$ and $\phi_2$ with $\theta := \substseq{X_2 \mapsto \enat{0}, Y_2 \mapsto Y_1, Z_2 \mapsto Y_1}$.}
        \label{subfig:fusion1}
    \end{subfigure}

     \begin{subfigure}{0.9\textwidth}
        \centering
        \begin{tikzpicture}[every node/.style={scale=0.8}]
        	\node[circle,draw] (e1) at (0,0) {$\phi_{1,2,2'}$};
        	\node[circle,draw] (e3) at (6,0) {$\phi_3$};
            \node (r) at (6.5,0) {$R$};

        	\draw[line width=2pt] (e1) -- (e3) node[midway, above] {$add(\enat{2}, Y_1, s(s(Y_1))) \eqq add(\enat{2}, \enat{2}, R)$};
        \end{tikzpicture}
        \subcaption{Fusion of $\phi_{1,2}$ and $\phi_2'$ with $\theta := \substseq{X_2' \mapsto s(\enat{0}), Y_2' \mapsto Y_1, Z_2' \mapsto s(Y_1)}$}
        \label{subfig:fusion2}
    \end{subfigure}

    \begin{subfigure}{0.9\textwidth}
        \centering
        \begin{tikzpicture}[every node/.style={scale=0.8}]
        	\node[circle,draw] (e1) at (0,0) {$\phi_{1,2,2',3}$};
            \node (r) at (1,0) {$\enat{4}$};
        \end{tikzpicture}
        \subcaption{Fusion of the two remaining stars with $\theta := \substseq{R \mapsto \enat{4}}$.}
        \label{subfig:fusion3}
    \end{subfigure}
    \caption{Fusion of the diagram from \Cref{subfig:fulldiagram}.}
    \label{fig:fusion}
\end{figure}

In Girard's original paper \cite[Section 2.3]{transyn1}, the evaluation of diagrams is defined as an edge contraction we call \emph{fusion}. An edge $e$ between two stars $\phi$ and $\phi'$ contain equations which are resolved and then the associated solution is propagated to both $\phi$ and $\phi'$. The two connected rays associated to $e$ are finally destructed in the process. It reminds of chemical interactions but also of how information is propagated and organised in a network. This process can fail in presence of errors during the execution of a unification algorithm.

We define this step-by-step procedure of fusion but also an alternative and equivalent notion of evaluation we call \emph{actualisation} which evaluates a diagram by solving the whole unification problem associated. It is similar to how small step evaluation differ to big step evaluation in the theory of programming languages \cite[Section 1.1]{amadio2016operational}.

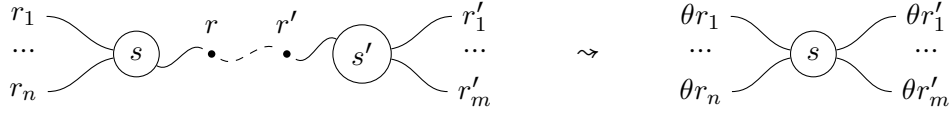
\begin{figure}
    \begin{tikzpicture}
        \node[circle,draw=black] at (0, 0) (c1) {$s$};
        \node[dot, label={90:$r$}] at (1, 0) (b2) {};
        \draw (c1) to[out=-30, in=150] (b2);

        \node[circle,draw=black] at (3, 0) (c2) {$s'$};
        \node[dot, label={90:$r'$}] at (2, 0) (a1) {};
        \draw (c2) to[in=-30, out=150] (a1);
        \draw[dashed] (b2) to[out=-30, in=150] (a1);
        
        \node at (-1.5, 0.5) (i1) {$r_1$};
        \node at (-1.5, 0) (i2) {$\hdots$};
        \node at (-1.5, -0.5) (i3) {$r_n$};
        \draw[-] (i1) to[in=170, out=0] (c1);
        \draw[-] (i3) to[in=190, out=0] (c1);
        \node at (4.5, 0.5) (o1) {$r_1'$};
        \node at (4.5, 0) (o2) {$\hdots$};
        \node at (4.5, -0.5) (o3) {$r_m'$};
        \draw[-] (c2) to[out=10, in=180] (o1);
        \draw[-] (c2) to[out=-10, in=180] (o3);
        
        \node at (6, 0) (arrow) {$\leadsto$};
        
        \node[circle,draw=black] at (9, 0) (c) {$s$};
        
        \node at (7.5, 0.5) (i1) {$\theta r_1$};
        \node at (7.5, 0) (i2) {$\hdots$};
        \node at (7.5, -0.5) (i3) {$\theta r_n$};
        \draw[-] (i1) to[in=170, out=0] (c);
        \draw[-] (i3) to[in=190, out=0] (c);
        \node at (10.5, 0.5) (o1) {$\theta r_1'$};
        \node at (10.5, 0) (o2) {$\hdots$};
        \node at (10.5, -0.5) (o3) {$\theta r_m'$};
        \draw[-] (c) to[out=10, in=180] (o1);
        \draw[-] (c) to[out=-10, in=180] (o3);
    \end{tikzpicture}
    \caption{Illustration of a step of fusion where $\theta$ is the principal unifier of the underlying unification problem of the pair of rays $(r, r')$. The fusion of the two stars $s$ and $s'$ along the rays $r$ and $r'$ produces a new star $s$.}
    \label{fig:fusionstep}
\end{figure}

\begin{defi}[Fusion, Figure~\ref{fig:fusionstep}]
Let $\delta : G_\delta \rightarrow \dgraph[A]{\Phi}$ be an $A$-diagram of a constellation $\Phi$. We define the \emph{fusion} of $\delta$ along a link $e = (v, v')$ in $E^{G_\delta}$ of label $(j, j')$ as a new diagram $\delta' : G_{\delta'} \rightarrow \dgraph[A]{\Phi + \phi}$ where $G_{\delta'}$ is $G_\delta$ such that:
\begin{enumerate}
	\item we compute $\theta := \solution{\{\undereq{e}\}}$;
	\item we define $\phi_1 := \Phi[\delta(v)]$ and $\phi_2 := \Phi[\delta(v')]$;
	\item we define $\phi_1'$ by $I_{\phi_1'} := I_{\phi_1'}\backslash\{j\}$ and $\phi_2'$ by $I_{\phi_2'} := I_{\phi_2'}\backslash\{j'\}$ and $\phi[i]$ behaves like $\phi'[i']$;
	\item we define a new star $\phi := \theta\phi_1' \uplus \theta\phi_2'$;
	\item $v$ and $v'$ merge and are replaced by $\phi$, \ie $\delta(v') = \delta(v) = \phi$ and some $x \in V^{G_\delta}$ is linked with $v$ if and only if it is connected to $v'$.
\end{enumerate}
We use the notation $G_\delta \leadsto_e G_{\delta'}$ for a step of this procedure resulting in $G_{\delta'}$, $G_\delta \leadsto^*_e G_{\delta'}$ for its reflexive transitive closure and $G_\delta \leadsto^n_e G_{\delta'}$ for the reachability of $G_{\delta'}$ from $G_\delta$ in $n \in \nat$ steps. We leave the reduced edge $e$ implicit when obvious or not important.
\end{defi}

\begin{defi}[Correct diagrams and their actualisation]
\label{def:actualisation}
An $A$-diagram $\delta$ of a constellation $\Phi$ is \emph{correct} if $\mathcal{P}(\delta)$ has a solution.

The \emph{actualisation} of a correct diagram $\delta$ is the star $\actu\delta$ defined by $I_{\actu\delta} := \freerays{\delta}$ such that $(\actu\delta)[(i,j)] = (\psi\circ\theta)(\Phi[i][j])$, where $\psi = \solution{\mathcal{P}(\delta)}$ and $\theta := \alpha_{v_1} \circ ... \circ \alpha_{v_n}$ with $V^{G_\delta} = \{v_1, ..., v_n\}$ is the composition of renamings of Definition~\ref{def:underproblem}.
\end{defi}

There are several ways to compute the solution of a unification problem. In this paper we use the Martelli-Montanari algorithm \cite{unifalgo}. We call \emph{partial execution} of a problem $P$ an arbitrary sequence of steps of the algorithm applied on $P$. Further details are given in \Cref{sec:unification}.

In the following proof, we treat $\freerays{\delta}$ as a star made of the free rays of $\delta$ for readability.

\begin{lem}[Equivalence of diagram reduction]
\label{lem:simfusion}
For all diagram $\delta$, there exists $\delta'$ such that $G_\delta \leadsto_e G_{\delta'}$ if and only if
there exists a partial execution from $\mathcal{P}(\delta)$ to $\mathcal{P}(\delta') \cup \{X_1 \eqq t_1, ..., X_k \eqq t_k\}$ with $\{X_1, ..., X_k\} \cap \bigcup_{j=1}^{k} \freevars{t_j} = \emptyset$ and $\freerays{\delta'} = \substseq{X_1 \mapsto t_1, ..., X_k \mapsto t_k}\freerays{\delta}$.
It means that a step of fusion corresponds to some steps of the unification algorithm (\cf\Cref{sec:unification}).
\end{lem}
\begin{proof}
We show the two implications.
\begin{itemize}
    \item $(\Rightarrow)$ Assume the fusion succeeds and produces a graph $G_{\delta'}$ by using the substitution $\theta := \substseq{X_1 \mapsto t_1, ..., X_k \mapsto t_k}$ corresponding to $\solution{\{\undereq{e}\}}$. We now consider the actualisation of $\delta$ which corresponds to solving the equation associated to $\delta$.
    By confluence of the unification algorithm (\cf \Cref{sec:unification}), we can focus on $e$ and isolate the result in order to obtain $\mathcal{P}(\delta') \cup \{X_1 \eqq t_1, ..., X_k \eqq t_k\}$ with $\{X_1 \eqq t_1, ..., X_k \eqq t_k\}$ in solved form (notice that $\mathcal{P}(\delta') = \mathcal{P}(\delta)\backslash\{e\}$).
    Then we have $\substseq{X_1 \mapsto t_1, ..., X_k \mapsto t_k} \mathcal{P}(\delta') \cup \{X_1 \eqq t_1, ..., X_k \eqq t_k\}$ by application of the unification algorithm (``replace" rule) on the equations $X_i \eqq t_i$ we previously isolated and ``stored".
    We obtain the equations corresponding to a new diagram $\delta''$ such that $\mathcal{P}(\delta'') = \substseq{X_1 \mapsto t_1, ..., X_k \mapsto t_k}\mathcal{P}(\delta')$.
    After application of $\solution{\{\undereq{e}\}}$ on $\mathcal{P}(\delta')$, the variables $X_1, ..., X_k$ are ``fixed", \ie they appear nowhere else, which prevents them to be altered during the execution of the algorithm and hence the substitutions of $\substseq{X_1 \mapsto t_1, ..., X_k \mapsto t_k}$ will appear in the last substitution applied on the free rays. It shows that we will have $\freerays{\delta''} = \substseq{X_1 \mapsto t_1, ..., X_k \mapsto t_k}\freerays{\delta}$ if we consider a notion of partial actualisation. This corresponds to $\delta'$, hence $\delta' = \delta''$.

    \item $(\Leftarrow)$ The previous point defines a correspondence between a step of fusion and some steps of the unification algorithm. By confluence of the unification algorithm, it is always possible to reorganise the order of resolution of equation so that the first step will correspond to a step of fusion, without any effect on the result.
\end{itemize}
\end{proof}

\begin{thm}[Equivalence between fusion and actualisation]
\label{thm:eqfusionactu}
For all diagram $\delta$, we have $G_\delta \leadsto^n G_{\actu \delta}$ for $n = |\mathcal{P}(\delta)|$.
\end{thm}
\begin{proof}
By induction on $n$. Assume we have $0$ links, hence $G_\delta$ does not reduce and has no edges. The only connected graph with no edge is a single vertex. This indeed corresponds to $G_{\actu \delta}$, as expected.
For the inductive case, we show that there exists a diagram $\delta'$ such that $G_\delta \leadsto G_{\delta'} \leadsto^n G_{\actu\delta}$ knowing $G_{\delta'} \leadsto^n G_{\actu\delta}$ (by induction hypothesis). The simulation of fusion (\cf Lemma~\ref{lem:simfusion}) tells us that a step of fusion exactly corresponds to some steps of the unification algorithm.
Consider a full application of the unification algorithm on $\mathcal{P}(\delta)$. By the confluence of the algorithm (\cf \Cref{sec:unification}), we can reorganise the computation of $\actu\delta$ so that some steps correspond to $G_\delta \leadsto G_{\delta'}$ and the remaining ones to $G_{\delta'} \leadsto^n G_{\actu\delta}$. Hence, we necessarily have a step of fusion $G_\delta \leadsto G_{\delta'}$.
\end{proof}

\begin{figure}
    \begin{tikzpicture}
        \node at (0, 1.75) (c) {$\phi_1+\hdots+\phi_n$};
        \node at (0, -0.5) (cl) {Constellation $\Phi$};
        
        \node at (3.5, 2.25) (s1) {$\phi_1$};
        \node at (3, 1.75) (s2) {$\phi_2$};
        \node at (4, 1.75) (s3) {$\phi_3$};
        \node at (3.5, 1.75) (dots) {$\hdots$};
        \node at (3.5, 1.25) (sn) {$\phi_n$};
        \node at (3.5, -0.5) (sl) {$\dgraph{\Phi}$};
        \path (s1) edge[loop above] (s1);
        \draw (s1) to (s2);
        \draw (s1) to (s3);
        \draw (s3) to (sn);

        \node at (6.5, 3.5) (s1) {$\phi_1$};
        \node at (7, 3) (s3) {$\phi_3$};
        \node at (6.5, 3) (dots) {$\hdots$};
        \node at (6.5, 2.5) (sn) {$\phi_n$};
        \draw (s1) to (s3);
        \draw (s3) to (sn);

        \node at (6.5, -0.5) (diags) {$\csatdiags(\Phi)$};
        \node at (6.5, 1.75) (dots3) {$\vdots$};

        \node at (6.5, 1) (t1) {$\phi_1$};
        \node at (6, 0.5) (t1pp) {$\phi_1$};
        \node at (7, 0.5) (t1p) {$\phi_1$};
        \node at (6.5, 0.5) (dots2) {$\hdots$};
        \node at (6.5, 0) (t2) {$\phi_2$};
        \draw (t1) to (t1p);
        \draw (t1) to (t1pp);
        \draw (t1pp) to (t2);

        \node at (10, 3) (d1) {$\psi_1$};
        \node at (10, 1.75) (dots4) {$\vdots$};
        \node at (10, 0.5) (d2) {$\psi_m$};
        \node at (10, -0.5) (ex) {$\exec(\Phi)$};

        \node[label={[align=center]\footnotesize list\\\footnotesize dependencies}] at ($(cl)!0.5!(sl)+(0,2.5)$) (l1) {};
        \node at ($(cl)!0.5!(sl)+(0,2.25)$) (arr1) {$\leadsto$};
        \node[label={[align=center]\footnotesize list\\\footnotesize diagrams}] at ($(sl)!0.5!(diags)+(0,2.5)$) (l1) {};
        \node at ($(sl)!0.5!(diags)+(0,2.25)$) (arr2) {$\leadsto$};
        \node[label={[align=center]\footnotesize evaluate\\\footnotesize diagrams}] at ($(diags)!0.5!(ex)+(0,2.5)$) (l1) {};
        \node at ($(diags)!0.5!(ex)+(0,2.25)$) (arr3) {$\leadsto$};
    \end{tikzpicture}
    \caption{Illustration of the execution of a finite and strongly normalising constellation. Notice that diagrams can be thought of as sort of runs in the dependency graph seen as a generalised automaton.}
    \label{fig:exec}
\end{figure}
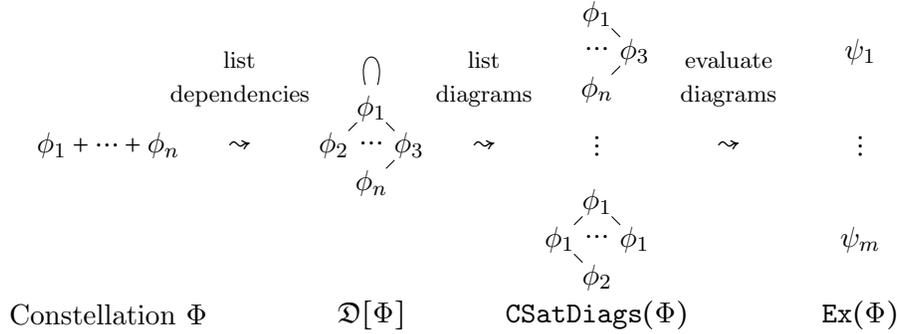

The \emph{execution} of a constellation $\Phi$ (\Cref{fig:exec}) consists in computing all the correct saturated diagrams of $\Phi$ and actualising them. In appearance, it is very similar to the resolution operator \cite[Chapter 3]{leitsch2012resolution} which is analogous to the consequence operator \cite[Section 2.2]{dantsin2001complexity} of logic programming. The difference is that we allow cyclic diagrams which makes our model closer to the construction of tilings in tile systems.

As we later show in \Cref{subsec:logicprograms}, allowing cyclic diagrams still preserves the interpretation of logic programs since cyclic diagrams are often wrong for logic programs: for the constellation $\Phi_\nat^{n+m}$ of Example~\ref{ex:constellation}, a loop can be constructed with
\[[-add(X, Y, Z), +add(s(X), Y, s(Z))],\] leading to the equation $X \eqq s(X)$ which has no solution.

\begin{nota}[Set of correct saturated diagrams]
We write $\satdiags_A(\Phi)$ for the \emph{set of all saturated diagrams} obtained from $\dgraph[A]{\Phi}$ for a constellation $\Phi$ and a set of colours $A \subseteq C$. We omit the set of colours and simply write $\satdiags(\Phi)$ when $A = C$.

We write $\csatdiags_A(\Phi)$ for the set of all diagrams in $\satdiags_A(\Phi)$ which are correct.
\end{nota}

\begin{defi}[Execution and normal form]
The \emph{execution} of a constellation $\Phi$ \wrt a set of colours $A \subseteq C$ is defined by $\exec_A(\Phi) := \actu\csatdiags_A(\Phi)$, where $\actu \csatdiags_A(\Phi) := \{\actu\delta \mid \delta \in \csatdiags_A(\Phi)\}$. We write $\exec(\Phi)$ when all colours in $\Phi$ participate in the execution.
\end{defi}

We discuss some design choices. Notice that the definition of diagram allows duplication of a same star, hence this apparently makes no difference whether or not a constellation is defined as a set or multiset. The purpose of defining constellations as multiset (actually indexed families) is to allow a quantitative analysis in the normal form. For instance, it would be possible to count how many times a given star appeared in the normal form.

Saturated diagrams of a constellation $\Phi$, although impossible to extend, may have free coloured rays which can be connected to the rays of another constellation $\Phi'$ when computing the interaction $\exec(\Phi \uplus \Phi')$. This is necessary in order to consider composition in logic and prove the associativity of execution (\cf Theorem~\ref{thm:assocexec}). However, this definition is different from Girard's original definition \cite[Section 2.3]{transyn1} which erases stars containing free coloured rays for technical reasons explicited in our interpretation of MLL (\cf \Cref{sec:proofs}). We instead split Girard's execution by defining more primitive operations which can be combined to our execution.

We define an operation of \emph{concealing} which violently mutes the constellation by removing stars containing polarised rays, thus forbidding any communication with other stars.

\begin{defi}[Concealing]
\label{def:concealing}
Let $\Phi$ be a constellation. The \emph{concealing} of $\Phi$ is the constellation $\lightning\Phi$ defined by $I_{\lightning\Phi} := \{i \in I_\Phi \mid \phi := \Phi[i], \forall j \in I_\phi, \phi[j] \text{ is uncoloured}\}$.
\end{defi}

We define an operation of \emph{noise filtering} of a constellation which removes the empty stars which are irrelevant since they cannot be connected. However, they still are valuable for quantitative analyses as we will set in the interpretation of MLL (\cf \Cref{sec:proofs}).

\begin{defi}[Noise filtering]
\label{def:filtering}
Let $\Phi$ be a constellation. The \emph{noise filtering} of $\Phi$ is the constellation $\flat\Phi := \{i \in I_\Phi \mid \Phi[i] \neq \emptystar\}$.
\end{defi}

\section{Computational illustrations and properties of execution}\label{sec:illustrations}

We illustrate how several common kinds of computation can be implemented in our model as certain classes of constellations. It shows that the stellar resolution is a very general and modular model of computation which expresses computation by transmission of data within a hypergraph structure. In particular, this generalises various classes of automata and Horn clauses used in logic programming and various tile systems. Although not explicitly shown in this paper, the stellar resolution should also represent a computational version of labelled transition systems which are commonly used in model checking \cite[Chapter 2]{baier2008principles}.

\subsection{Logic programs}\label{subsec:logicprograms}

A natural illustration of the computational power of the stellar resolution is the encoding of logic programs since the stellar resolution directly generalises Robinson's first-order resolution which corresponds to the core of logic programming.

First, it is possible to do programming with predicate calculus \cite{kowalski1974predicate}. It is then known that formulas of predicate calculus can be normalised so that formulas are represented only by conjunction of disjunctions (called clauses) with only universal quantifiers appearing as prefix \cite[Section 3.2]{hedman2004first}. Formulas are then of the shape $\forall x_1, ..., x_n. (A_1^1 \lor ... \lor A^1_n) \land ... \land (A^m_1 \lor ... \lor A^m_k)$ where every $A^x_y$ is an atomic formula. We use those normalised formulas of predicate calculus with at most one positive (without negation) atom in each clause. Such normalised formulas called \emph{Horn clauses} represent sequents $\Gamma \vdash A$ for a set of hypotheses $\Gamma$.

A \emph{fact} is a closed (variable-free) first-order formula. Several facts form a \emph{knowledge base}. We have \emph{rules} which can be used to infer new facts from the available ones and thus expend the knowledge base. Rules are often represented as implications $A_1, ..., A_n \vdash B$ called \emph{Horn clauses} \cite{horn1951sentences, tarnlund1977horn}. A \emph{query} asks if it is possible to infer a given fact from the knowledge base and is itself represented as a fact symbolising a goal. A \emph{logic program} is a multiset of rules and facts.

The translation is direct. We use the use the polarities to distinguish between hypothesis and conclusion (or input and output). The translation of a fact is defined by
\[P(t_1, ..., t_n)^\bigstar := [+P(t_1, ..., t_n)].\]
For a rule, the translation is defined by:
\[\big( \wedge_{i=1}^m P_1(t_1^i, ..., t_n^i) \vdash Q(u_1, ..., u_k)\big)^\bigstar := (\bigcup_{i=1}^m \{-P_1(t_1^i, ..., t_n^i)\}) \cup \{+Q(u_1, ..., u_k)]\}.\]
Finally, for a query, we have:
\[(?P(t_1, ..., t_n))^\bigstar := [-P(t_1, ..., t_n), r_1, ..., r_m]\]
with $\{r_1, ..., r_m\} = \bigcup_{i=1}^n \vars{t_i}$ which represent the information we would like to make visible in the output (see Example~\ref{ex:constellation} where $[-add(\enat{n}, \enat{m}, r), r]$ is the query).
A logic program $P := \bigcup_{i=1}^n \{C_i\}$ becomes $P^\bigstar := \bigcup_{i=1}^n \{C_i^\bigstar\}$.

The set of answers for a query $q$ on a program $P$ is defined by a set of substitutions $A_P^q = \{\theta_1, ..., \theta_k\}$ such that for all $\theta \in A_P^q$, we have $\theta q$ logically satisfied by $P$, written $P \vDash \theta q$.
The answers are usually computed by iteratively applying the resolution rule (\cf Definition~\ref{def:resolution} in \Cref{sec:unification}) between $q$ and all possible $C \in P$ until either no variables remain in $q$ or the resolution rule is no more applicable. We refer to definitions of the SLD-resolution itself derived from Kowalski's SL-resolution \cite{kowalski1974predicate, kowalski1971linear} for more details about the computation of answers.

\begin{thm}[Simulation of logic programs]
Let $P$ be a logic program with query $q$ and $P^\bigstar$ and $q^\bigstar$ be their translation. We have $\flat\exec(P^\bigstar + q^\bigstar) = \{\theta_1 \phi_1, ..., \theta_k \phi_k\}$ if and only if for all $\theta_i$, $P \vDash \theta_i q$.
\end{thm}
\begin{proof}
The proof relies on the fact that the execution reproduces the SLD-resolution \cite{kowalski1974predicate, kowalski1971linear}. The satisfiability of a query is linked to the idea of ``proof-search": it is satisfiable when proved by facts, themselves proved by other facts and so on until nothing is left unproven. SLD-resolution tries to satisfy a query by matching it with the available facts and rules. Stars can actually be seen as first-order disjunctive clauses. Looking for justifications of facts corresponds to the construction of diagrams and the fact of leaving nothing unproven corresponds to the saturation of diagrams. In particular, in presence of $k$ possible choices of rules, $k$ answers are computed independently in parallel and we obtain saturated diagrams $\delta_1, ..., \delta_k$. 

The fact of matching a query against available facts and rules can be seen as constructing a saturated and correct diagram. In the absence of error, we indeed obtain an instantiation of the variables $\vars{q}$ through a substitution $\theta_i$. This exactly coincides with the actualisation of correct diagrams. By Lemma~\ref{thm:eqfusionactu} this is equivalent to a full fusion. Having number of correct diagrams corresponds to the number of answers.

Remark that in the case of diagrams, only the free rays survive in the output, hence we have to add uncoloured rays corresponding to $\vars{q}$ to correctly simulate logic programs. We could also add rays $x \cdot X$ where $x$ is a constant representing $X$ in order to keep the name of variable in the output. We would finally obtain a normal form made of stars $\phi_i = \bigcup_{i=1}^k \{x_i \cdot r_i\}$ such that $\phi_i$ corresponds to some $\theta \in A_P^q$.

Additionally, we have to ensure that our relaxation to cyclic diagrams do not cause problems. In logic programming, we usually require that the rays of a star have exactly the same variables (all variables are bound). Because of this restriction, cycles in dependencies graphs of logic programs, when reduced to a loop on a single star, either involve equations of the shape $t \eqq t$ or equations of the type $X \eqq f(X)$. In the former case, if the associated rule is binary, \ie of the shape $A \vdash B$, we obtain the empty star $\emptystar$ which is irrelevant in the computation and removed by the operator $\flat$. If the rule is not binary, \eg of the shape $A_1, A_2, ..., A_n \vdash B$, the equation $t \eqq t$ associated to the loop is erased because it has no effect on the computation. In the latter case of the ill-behaving equation $X \eqq f(X)$, the whole diagram is incorrect and ignored in the output.
\end{proof}

An example of logic program computing unary addition and its evaluation for the case of $2+2$ is illustrated in \Cref{fig:diagrams} and \ref{fig:fusion}.

\subsection{Non-deterministic Turing machines}\label{subsec:turing}

Definitions of Turing machines are taken for Sipser's introduction to the theory of computation \cite[Section 3.1 in Second Edition]{sipser1996introduction}.

Links between tile systems and automata have already been studied \cite{mazurkiewicz1988basic, thomas1991logics} by considering recognisability on graphs, inducing sort of \emph{asynchronous automata}. The encoding of automata in the stellar resolution follows a similar idea: the construction of diagrams simulates a run on a given word. It is possible to encode directed graphs by translating edges $(e, e')$ with binary stars $[-g(e), +g(e')]$. It is then possible to encode an automata transitions by first encoding their state graph then extending the rays so that the fusion triggers a transmission of information (the remaining characters to be read). The final state will contain a dummy unpolarised ray $\lacc$ so that the existence of a visible output in the normal form will correspond to the acceptation a word.

In this section, we suggest an encoding of non-deterministic Turing machines. We use the fact that Turing machines can be represented with two stacks in order to represent the left and right part of a tape. A move of the head will be represented as a manipulation of stack (push or pop of a symbol).

\begin{defi}[Encoding of words]
If $w = c_1...c_n$ then $w^\bigstar = [+i(c_1 \cdot ... \cdot c_n \cdot \text{\textvisiblespace})]$ with the binary function symbol $\cdot$ which is considered right-associative, \ie $a \cdot b \cdot c = a \cdot (b \cdot c)$ and a constant $\text{\textvisiblespace}$ for the empty character.
\end{defi}

Different encodings of words are possible. For instance, in Aubert and Bagnol's works \cite[Definition 24]{aubert2014unification}\cite[Definition 10]{aubert2016unary}, the characters are encoded with objects called \emph{flows} (which can be seen as binary stars) forming a cyclic chain of $\alpha$-unifiable terms which interact with the encoding of an automaton. This defines a characterisation of logspace computation where the input is explored with pointers.

A non-deterministic Turing machine is a tuple $M = (Q, \Gamma, \delta, q_0, q_a, q_r)$ where $Q$ is the set of states, $\Gamma$ is the alphabet of the tape, $\delta : Q \times \Gamma_\text{\textvisiblespace} \rightarrow \mathcal{P}(Q \times \Gamma_\text{\textvisiblespace} \times \{\lc, \rc, \stopc\})$ is the transition function, $q_0$ is the initial state and finally, $q_a$ and $q_r \neq q_a$ are respectively the state of acceptation and rejection.
We write $\Gamma_\text{\textvisiblespace}$ for $\Gamma \cup \{\text{\textvisiblespace}\}$. 

A configuration is a triple $(l, q, r)$ where $q \in Q$ and $l, r$ are tapes. It represents the position of the head on the tape and the associated state. We say that a configuration $C$ leads to $C'$ when moving the head in $C$ accordingly to $\delta$ leads to $C'$. A word $w = c_1 ... c_n$ is accepted by $M$, written $M(w) = 1$, when there is a sequence of configurations $C_1, ..., C_n$ such that:
\begin{enumerate}
    \item $C_1 = (\text{\textvisiblespace}, q_0, w)$;
    \item $C_i$ leads to $C_{i+1}$;
    \item $C_n = (l, q_a, r)$ for some $l$ and $r$.
\end{enumerate}

If the last configuration has a state $q_r$ instead, we say that $M$ rejects $w$, which is written $M(w) = 0$. When $M$ loops infinitely on $w$, we write $M(w) = \infty$. We require that an NTM necessarily ends on $q_a$ or $q_r$ when it stops.

For the encoding, we use the facts that Turing machines can be represented with two stacks in order to represent the left and right part of a tape. A move of the head will be represented as a manipulation of stack.

We use terms $m(L, Q, X, R)$ where $L$ and $R$ are the left and right part of the tape relatively to the current position of the head. The variables $Q$ and $X$ respectively represent the current state and symbol read by the head. We implicitly consider the symbol $\bullet$ as left-associative (hence $a \bullet b \bullet c = (a \bullet b) \bullet c$) and $\circ$ right-associative (hence $a \circ b \circ c = a \circ (b \circ c)$) so that it looks like we are traversing a tape.

\begin{figure}
    \centering
    \begin{tikzpicture}
        \node[state] (qa) {$\mathtt{q_a}$};
        \node[state, initial above, right of=qa, xshift=20mm] (q0) {$\mathtt{q_0}$};
        \node[state, right of=q0, xshift=40mm] (q1) {$\mathtt{q_1}$};
        \node[state, right of=q1, xshift=30mm] (qr) {$\mathtt{q_r}$};
        \node[state, right of=q0, xshift=15mm, yshift=20mm] (q2) {$\mathtt{q_2}$};
        \node[state, right of=q0, xshift=15mm, yshift=-20mm] (q3) {$\mathtt{q_3}$};
        \draw[-latex]
        (q0) edge[loop below] node{$\$ \rightarrow \$, \rc$} (q0)
        (q2) edge[loop above] node[text width=2cm]{$a \rightarrow a, \rc$\\$\$ \rightarrow \$, \rc$} (q2)
        (q3) edge[loop below] node[text width=2cm]{$b \rightarrow b, \rc$\\$\$ \rightarrow \$, \rc$} (q3)
        (q1) edge[loop right] node[text width=2cm]{$a \rightarrow a, \lc$\\$b \rightarrow b, \lc$\\$\$ \rightarrow \$, \lc$} (q1)
        (q0) edge[above, midway, sloped] node{$a \rightarrow \$, \rc$} (q2)
        (q0) edge[below, midway, sloped] node{$b \rightarrow \$, \rc$} (q3)
        (q2) edge[above, midway, sloped] node{$b \rightarrow \$, \lc$} (q1)
        (q3) edge[below, midway, sloped] node{$a \rightarrow \$, \lc$} (q1)
        (q2) edge[above, bend left, midway, sloped] node{$\text{\textvisiblespace} \rightarrow \text{\textvisiblespace}, \stopc$} (qr)
        (q3) edge[below, bend right, midway, sloped] node{$\text{\textvisiblespace} \rightarrow \text{\textvisiblespace}, \stopc$} (qr)
        (q0) edge[above] node{$\text{\textvisiblespace} \rightarrow \text{\textvisiblespace}, \stopc$} (qa)
        (q1) edge[above] node{$\text{\textvisiblespace} \rightarrow \text{\textvisiblespace}, \rc$} (q0);
    \end{tikzpicture}
    $M^\bigstar = [-i(C \cdot W), +m(\text{\textvisiblespace}, q_0, C, W)]+ [-i(\text{\textvisiblespace}), +m(\text{\textvisiblespace}, q_0, \text{\textvisiblespace}, \text{\textvisiblespace})]+$ \\

    $[-m(L, q_0, \text{\textvisiblespace}, R), +m(L, q_a, \text{\textvisiblespace}, R)]+[-m(L, q_2, \text{\textvisiblespace}, R), +m(L, q_r, \text{\textvisiblespace}, R)]+$ \\

    $[-m(L, q_0, \$, C \circ R), +m(L \bullet \$, q_0, C, R)]+
    [-m(L, q_2, \$, C \circ R), +m(L \bullet \$, q_2, C, R)]+$ \\

    $[-m(L, q_0, a, C \circ R), +m(L \bullet \$, q_2, C, R)]+
    [-m(L, q_2, a, C \circ R), +m(L \bullet a, q_2, C, R)]+$ \\

    $[-m(L, q_0, b, C \circ R), +m(L \bullet \$, q_3, C, R)]+
    [-m(L \bullet C, q_2, b, R), +m(L, q_1, C, \$ \circ R)]+$ \\

    $[-m(L, q_1, \text{\textvisiblespace}, C \circ R), +m(L \bullet \text{\textvisiblespace}, q_0, C, R)]+
    [-m(L, q_3, \text{\textvisiblespace}, R), +m(L, q_r, \text{\textvisiblespace}, R)]+$ \\

    $[-m(L \bullet C, q_1, \$, R), +m(L, q_1, C, \$ \circ R)]
    +[-m(L, q_3, \$, C \circ R), +m(L \bullet \$, q_3, C, R)]+$ \\

    $[-m(L \bullet C, q_1, a, R), +m(L, q_1, C, a \circ R)]+
    [-m(L \bullet C, q_3, a, R), +m(L, q_1, C, \$ \circ R)]+$ \\

    $[-m(L \bullet C, q_1, b, R), +m(L, q_1, C, b \circ R)]+
    [-m(L, q_3, b, C \circ R), +m(L \bullet b, q_3, C, R)]+$ \\

    $[-m(L, q_a, X, R), \lacc]+[-m(L, q_r, X, R), \lrej]+$ \\
    $[-m(\text{\textvisiblespace}, Q, C, R), +m(\text{\textvisiblespace} \bullet \text{\textvisiblespace}, Q, C, R)] + [-m(L, Q, C, \text{\textvisiblespace}), +m(L, Q, C, \text{\textvisiblespace} \circ \text{\textvisiblespace})]$
    \caption{A Turing machine accepting words containing as many symbols $\mathtt{a}$ as symbols $\mathtt{b}$ where $a \rightarrow b, d$ from a state $q$ to $q'$ corresponds to a transition $\delta(q, a) = (q', b, d)$. When computing $\exec(M^\bigstar + a^\bigstar)$, we plug the input with the correct initial star and obtain $[+m(\text{\textvisiblespace}, \mathtt{q_0}, a, \text{\textvisiblespace})]$. No star can be connected, hence we have to connect to the right allocation star and obtain $[+m(\text{\textvisiblespace}, \mathtt{q_0}, a, \text{\textvisiblespace} \circ \text{\textvisiblespace})]$. We can use the star corresponding to $a \rightarrow \$, \rc$ and obtain $[+m(\text{\textvisiblespace} \bullet \$, \mathtt{q_2}, \text{\textvisiblespace}, \text{\textvisiblespace})]$. Since we read $\text{\textvisiblespace}$, we the use star corresponding to the transition $\text{\textvisiblespace} \rightarrow \text{\textvisiblespace}, \stopc$ and obtain $[+m(\text{\textvisiblespace} \bullet \$, \mathtt{q_r}, \text{\textvisiblespace}, \text{\textvisiblespace})]$. We can only use the star corresponding to $q_r$ and obtain $[\lrej]$. If we had a character $b$ next to $a$, we would reach $[\lacc]$. }
    \label{fig:automata}
\end{figure}

\begin{defi}[Encoding of non-deterministic Turing machines]
\label{def:encntm}
The encoding of an NTM $M = (Q, \Gamma, \delta, q_0, q_a, q_r)$ is defined by a constellation $M^\bigstar$ such that:
\begin{itemize}
    \item $q_0$ is translated into $[-i(C \cdot W), +m(\text{\textvisiblespace}, q_0, C, W)]+[-i(\text{\textvisiblespace}), +m(\text{\textvisiblespace}, q_0, \text{\textvisiblespace}, \text{\textvisiblespace})]$;
    \item $q_a$ is translated into $[-m(L, q_a, X, R), \lacc]$;
    \item $q_r$ is translated into $[-m(L, q_r, X, R), \lrej]$;
    \item for each $q \in Q$ and $c \in \Gamma_\text{\textvisiblespace}$ such that $(q', c', d) \in \delta(q, c)$:
    \begin{itemize}
        \item if $d = \lc$ (going left) then we have $[-m(L \bullet X, q, c, R), +m(L, q', X, c' \circ R)]$;
        \item if $d = \rc$ (going right) then we have $[-m(L, q, c, X \circ R), +m(L \bullet c', q', X, R)]$;
        \item if $d = \stopc$ (staying still) then we have $[-m(L, q, c, R), +m(L, q', c', R)]$;
    \end{itemize}
    \item we add two additional ``memory allocation stars'':
    \[[-m(\text{\textvisiblespace}, Q, C, R), +m(\text{\textvisiblespace} \bullet \text{\textvisiblespace}, Q, C, R)] + [-m(L, Q, C, \text{\textvisiblespace}), +m(L, Q, C, \text{\textvisiblespace} \circ \text{\textvisiblespace})].\]
\end{itemize}
\end{defi}

The two last stars are used to dynamically allocate space on the tape when necessary (similarly to $\mathtt{malloc()}$ in the C language). Instead of considering Turing machines as word acceptors, it is also possible to output the content of the tape and hence compute functions by translating $q_a$ into $[-m(L,q_a, X, R), \lacc(L, X, R)]$.

\begin{thm}[Simulation of non-deterministic Turing machines]
\label{thm:simturing}
Let $M$ be an NTM such that $q_a$ and $q_r$ have no outgoing transitions and $w$ a word. We have:
\begin{enumerate}
    \item $M(w) = 1$ if and only if $[\lacc] \in \flat\lightning\exec(M^\bigstar + w^\bigstar)$;
    \item $M(w) = 0$ if and only if $\big([\lacc] \not\in \flat\lightning\exec(M^\bigstar + w^\bigstar)$ and
    $\flat\lightning\exec(M^\bigstar + w^\bigstar) \neq \emptyset\big)$.
\end{enumerate}
\end{thm}
\begin{proof}
By design, $\dgraph{M^\bigstar + w^\bigstar}$ is isomorphic to the state graph of $M$ (that is, there is a link between two rays if and only if the two corresponding states are adjacent) and each run is isomorphic to some linear correct saturated diagram (because transitions correspond to binary stars).

A major difference with finite automata is the possibility of infinite computation. Such infinite computation made by constantly going from one state to the another (which can be the same one) corresponds to the existence of a diagram which cannot be saturated and hence does not appear in the normal form. Since $q_a$ and $q_r$ have no outgoing edges, it is impossible to have the non-deterministic choice of either stopping on a terminal state or continuing. If it was possible, we would obtain a misleading $[\lacc]$ or $[\lrej]$ in the output.   

We give arguments showing that the dynamics of Turing machine (the transition function) is correctly simulated. We show that the fusion of stars correctly simulates the composition of transitions. Assume we have a star $[+m(l, q, c, r)]$ representing a configuration of the Turing machine. We have three cases depending on the direction: 
\begin{itemize}
	\item if we are going left, we have a transition $[-m(L \bullet X, q, c, R), +m(L, q', X, c' \circ R)]$. The fusion is successful only when $l$ is of the shape $l' \bullet k$. In this case, by unification we have $l' = L$ and $X = c$ which identifies a next symbol on the left. The evaluation produces the star $[+m(L, q', k, c' \circ r)]$ which corresponds to writing $c'$ after reading $c$ and placing it on the right part of the tape to read the symbol $k$ on the left. This indeed corresponds to ``going on the left" in the tape;
	\item if we are going right, the reasoning is similar;
	\item if we stop the head, we have a transition $[-m(L, q, c, R), +m(L, q', c', R)]$ and the fusion only moves from a state $q$ (when reading a symbol $c$) to another state $q'$ (after writing the symbol $c'$).
\end{itemize}

We now check the limit cases when the machine is out of memory (not enough space on the tape to apply a transition). These cases happen because Turing machines have potentially infinite tapes but we only manipulate finite extensible tapes. When going on the left, it happens that the left part of the tape $l$ is not of the shape $l' \bullet k$. The typical case is when we have $l = \text{\textvisiblespace}$. In this case, we can arbitrarily use the left allocation star and freely obtain $\text{\textvisiblespace} \bullet \text{\textvisiblespace}$. Remark that it is impossible to allocate too much space because the allocation stars require that we have a tape equal to $\text{\textvisiblespace}$. Otherwise, we would have infinitely many diagrams for all the possible amount of space allocation and no encoding of Turing machine would be strongly normalising.

Using all the previous arguments, we check the statements (1) and (2) of the simulation theorem for Turing machines.
\begin{itemize}
    \item (1) Assume that $w \in \mathcal{L}(M)$ and there is a non-deterministic run reaching either $q_a$. By correspondence between runs and diagram for Turing machines, we must have a saturated and correct linear diagram reaching the ray $\lacc$. Since stars are binary and $q_a$ is terminal, this ray can only be reached once in a diagram. Hence, such a diagram actualises into $[\lacc]$. Therefore, we have $[\lacc] \in \flat\lightning\exec(M^\bigstar + w^\bigstar)$.
    The converse implication uses the same argument with the remark that uncomplete runs correspond to diagrams which are erased by the operator $\lightning$. Notice that we can reject sometimes but what matters is the existence of at least one non-deterministic run which reach $q_a$.
    \item (2) We show the two implications for the second statement.
    \begin{itemize}
        \item $(\Rightarrow)$ Assume that $w \not\in \mathcal{L}(M)$ and that $M$ terminates, meaning that $M$ rejects $w$. Similarly to the proof of statement (1), we reach $\lrej$ and never $\lacc$. We indeed have $[\lacc] \not\in \flat\lightning\exec(M^\bigstar + w^\bigstar)$ and
        $\flat\lightning\exec(M^\bigstar + w^\bigstar) \neq \emptyset$.
        \item $(\Leftarrow)$ Assume that $[\lacc] \not\in \flat\lightning\exec(M^\bigstar + w^\bigstar)$ and
        $\flat\lightning\exec(M^\bigstar + w^\bigstar) \neq \emptyset$. The only possibility is to have $n$ occurrences of $[\lrej]$ in $\flat\lightning\exec(M^\bigstar + w^\bigstar)$. By correspondence between runs and diagram, we have a run reaching a rejecting state, hence $w \not\in \mathcal{L}(M)$ and $M$ terminates on $w$.
    \end{itemize}
\end{itemize}

We have to check that cyclic diagrams cause no problems. Since Turing machines only use binary polarised stars, all cyclic diagrams must be closed. Consider such a cyclic diagram. If it is incorrect, then it is erased in the normal form. In case it is correct, it represents a trivially infinite loop during the execution of the machine such as $[-m(L, q, c, R), +m(L, q, c, R)]$. This diagram will reduce into the empty star $\emptystar$ which is erased by the operator $\flat$. Hence, it has no effect on the statement.
\end{proof}

Although we have shown an example of \emph{deterministic} Turing machine in \Cref{fig:automata}, it is easy to see how the non-deterministic case works. We can have several choices so that a same transition can match with several transitions. This will necessarily yield several diagrams corresponding to different runs. The whole machine accepts the input when at least one run accepts the word.

\begin{cor}[Halting problem]
The problem of determining if $\flat\lightning\exec(M^\bigstar + w^\bigstar) \neq \emptyset$ is undecidable.
\end{cor}
\begin{proof}
By the simulation of non-deterministic Turing machines, a Turing machine $M$ terminates if and only if $\flat\lightning\exec(M^\bigstar + w^\bigstar) \neq \emptyset$ since an infinite computation is either represented as the emptiness of the output (in case we only have diagrams impossible to saturate) or as the production of the empty star (\eg $[-m(L, q, c, R), +m(L, q, x, R)]$) which is considered as a trivial loop. This is a known undecidable problem \cite{davis1958computability}.
\end{proof}

\subsection{Abstract tile assembly model}\label{subsec:atam}

The abstract tile assembly \cite{winfree1998algorithmic, patitz2014introduction} (aTAM) is a tile system used in DNA computing \cite{seeman} which extends Wang tiles (\cf \Cref{subsec:tiles2prolog}). We present the idea without too much formality and refer to Lathrop et al. \cite{lathrop2009strict} for more details\footnote{However, we use a variant without seed assembly $\sigma$ because it is more natural in our case.}.

\begin{figure}
    \centering
    \scalebox{0.8}{
    \begin{tikzpicture}[every node/.style={scale=0.8}]
        \draw (0,0) -- (0,1.5) -- (1.5,1.5) -- (1.5,0) -- (0,0);
        \node[below] at (0.75,1.5) {$N^2$};
        \node[right] at (0,0.75) {$W^2$};

        \draw (-1.5,0) -- (-1.5,1.5) -- (0,1.5) -- (0,0) -- (-1.5,0);
        \node[below] at (-0.75,1.5) (c) {$O^1$};
        \node[left] at (0,0.75) {$W^2$};
        \node[right] at (-1.5,0.75) {$W^2$};

        \draw (-3,0) -- (-3,1.5) -- (0,1.5) -- (0,0) -- (-3,0);
        \node[below] at (-2.25,1.5) {$O^1$};
        \node[left] at (-1.5,0.75) {$W^2$};
        \node[right] at (-3,0.75) {$W^2$};

        \draw (0,1.5) -- (0,3) -- (1.5,3) -- (1.5,1.5) -- (0,1.5);
        \node[below] at (0.75,3) {$N^2$};
        \node[right] at (0,2.25) (d) {$O^1$};
        \node[above] at (0.75,1.5) {$N^2$};

        \node at (-3.5,0.75) {$\hdots$};
        \node at (0.75,3.5) {$\vdots$};

        \draw (-2,2) -- (-2,3.5) -- (-0.5,3.5) -- (-0.5,2) -- (-2,2);
        \node[below] at (-1.25,3.5) {$Z^1$};
        \node[left] at (-0.5,2.75) (a) {$O^1$};
        \node[right] at (-2,2.75) {$Z^1$};
        \node[above] at (-1.25,2) (b) {$O^1$};

        \draw[-latex] (a) -- (d);
        \draw[-latex] (b) -- (c);

        \node[scale=4, inner sep=2pt] at (-7,1) (o1) {$O^1$};
        \node[scale=1.5] at (-9,2.5) (g) {glue type};
        \node[scale=1.5] at (-5,2.5) (s) {strength};
        \draw[-latex] (g) -- (o1);
        \draw[-latex] (s) -- (o1);
    \end{tikzpicture}}
    \caption{Illustration of an assembly in an aTAM. Assume we are at temperature $\tau = 2$. We can connect a new tile to an assembly because the glue types match and the sum of strengths involved is $1+1 \geq \tau$.}
    \label{fig:atam}
\end{figure}
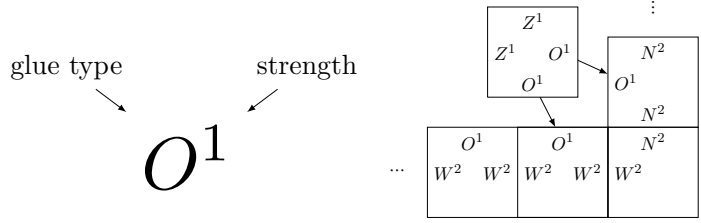

We define a \emph{tile type} by $t_i = (g^i_\mathtt{w}, g^i_\mathtt{e}, g^i_\mathtt{s}, g^i_\mathtt{n})$ for some $i$ in a finite set of indexes $I$ as objects intuitively corresponding to squares with each sides associated to a \emph{glue type} $\mathtt{gl}(g^i_d)$ for a direction $d \in \{\mathtt{w}, \mathtt{e}, \mathtt{s}, \mathtt{n}\}$ (for west, east, south and north) and a natural number $\str{g^i_d}$ called its \emph{strength}. The idea is that we have a global variable $\tau \in \nat$ called the temperature and that a tile can be connected to other ones if the sum of strength involved in the connexion is at least $\tau$. This phenomenon is known as \emph{cooperation}.

A tile assembly system (TAS) is a pair $\mathcal{T} = (T, \tau)$ where $T$ is a set of tile types and $\tau \in \nat$ is the temperature of $\mathcal{T}$.

Given a set of tile types $T$, a $T$-configuration is a partial function $\alpha : \rel^2 \rightarrow T$ pasting tiles to the plane. It is associated to a (connected) grid graph $G_\alpha$ with vertices $V^{G_\alpha} := \mathtt{dom}(\alpha)$ and there is an edge between two vertices representing tiles $t_i, t_j$ with $i \neq j$ when $\mathtt{gl}(g^i_d) = \mathtt{gl}(g^j_{d'})$ for $d = \op(d')$ where $\op$ is the involution defined by $\op(\mathtt{e}) = \mathtt{w}$ and $\op(\mathtt{n}) = \mathtt{s}$.

We say that $\alpha$ is $\tau$-stable if it is impossible to cut $E^{G_\alpha}$ into two parts such that it breaks bonds of total strength at least $\tau$. In other words, it means that a new tile can be added to a $T$-configuration only if the total strength value of its bonding is at least $\tau$.

A $T$-assembly for $\tau$ is a $T$-configuration which is $\tau$-stable. Given a TAS $\mathcal{T} = (T, \tau)$, we write $\mathcal{A}_\square[\mathcal{T}]$ for the set of all $T$-assemblies for $\tau$ which are connected and maximal (impossible to extend with more tiles from a given set of tile types). An example is given in \Cref{fig:atam}.

We suggest an encoding of the aTAM in $\nat^2$ instead of $\rel^2$ which is more natural but not less powerful since it is known that $\nat^2 \simeq \rel^2$ and also because we are able to compute any computable function. Tile types $t_i = (g^i_\mathtt{w}, g^i_\mathtt{e}, g^i_\mathtt{s}, g^i_\mathtt{n})$ are encoded by a star $t_i^\bigstar$:
\[[-\overset{\bullet}{h}(\mathtt{gl}(g^i_\mathtt{w})(X), X, Y),-\overset{\bullet}{v}(\mathtt{gl}(g^i_\mathtt{s})(Y), X, Y),\]
\[+\overset{\circ}{h}(\mathtt{gl}(g^i_\mathtt{e})(s(X)), s(X), Y),+\overset{\circ}{v}(\mathtt{gl}(g^i_\mathtt{n})(s(U)), X, s(Y))]\]
where $gl(g)(X) := g(X) \cdot \enat{str(g)}$ for $str(g) \in \nat$. The symbols $h$ (horizontal) and $v$ (vertical) represent axis of connexion. The key point of the encoding is that because of the dots $\bullet$ and $\circ$, the tiles cannot connect directly but has to use an intermediary star checking that the connexion is possible that we need to define.

The \emph{environment constellation} for a temperature $\tau \in \nat\backslash\{0\}$ is defined by $\Phi_{env}^\tau :=$
\[[+temp(\enat{\tau})]+\left[
\begin{array}{ll}
    +\overset{\bullet}{v}(g_1(X_1) \cdot N_1, X_1, Y_1), &
    -\overset{\circ}{v}(g_2(X_3) \cdot N_2, X_3, Y_3),
    \\
    +\overset{\bullet}{h}(g_3(X_5) \cdot N_3, X_5, Y_5), &
    -\overset{\circ}{h}(g_4(X_7) \cdot N_4, X_7, Y_7),
    \\
    -\overset{\circ}{v}(g_1(X_2) \cdot N_1, X_2, Y_2), &
    +\overset{\bullet}{v}(g_2(X_4) \cdot N_2, X_4, Y_4),
    \\
    -\overset{\circ}{h}(g_3(X_6) \cdot N_3, X_6, Y_6), &
    +\overset{\bullet}{h}(g_4(X_8) \cdot N_4, X_8, Y_8),
    \\
    -add(N_1, N_2, R_1), & -add(N_3, N_4, R_2),
    \\
    \multicolumn{2}{c}{-add(R_1, R_2, R), -geq(R, T, \enat{1}), -temp(T)}
\end{array}
\right]
\]
\[
+
[-\overset{\bullet}{v}(g(X) \cdot \enat{0}, X, Y)]
+
[+\overset{\circ}{v}(g(X) \cdot \enat{0}, X, Y)]
+
[-\overset{\bullet}{h}(g(X) \cdot \enat{0}, X, Y)]
+
[+\overset{\circ}{h}(g(X) \cdot \enat{0}, X, Y)]
\]
\[
+
[+\overset{\circ}{v}(g(X) \cdot \enat{0}, X, Y)]
+
[-\overset{\bullet}{v}(g(X) \cdot \enat{0}, X, Y)]
+
[+\overset{\circ}{h}(g(X) \cdot \enat{0}, X, Y)]
+
[-\overset{\bullet}{h}(g(X) \cdot \enat{0}, X, Y)]
\]
\[+[+geq(\enat{0},\enat{0},\enat{1})]+[+geq(s(X),s(Y),R), -geq(X, Y, R)]+[+geq(s(X), \enat{0}, \enat{0})]+\]
\[[+geq(\enat{0}, s(Y), \enat{0})]+[+add(\enat{0}, Y, Y)]+[-add(X, Y, Z), +add(s(X), Y, s(Z))]\]

We define the translation of a set of tile types $T$ as the constellation $T^\bigstar := \sum_{t_i \in T} t_i^\bigstar$.

\begin{thm}[Simulation of the aTAM]
Let $\mathcal{T} = (T, \tau)$ be a TAS. We have\[\csatdiags(T^\bigstar + \Phi^\tau_{env}) \simeq \mathcal{A}_\square[\mathcal{T}].\]
\end{thm}
\begin{proof}
It is sufficient to show that the computation of $\csatdiags(T^\bigstar + \Phi^\tau_{env})$ behaves like the construction of tilings in aTAM.

Direct connexions between tiles without using $\Phi_{env}^\tau$ is forbidden because of the symbols $\circ$ and $\bullet$. Notice that the colours $v$ and $h$ force the connexions to be on the same axis in order to follow the geometric restriction of tiling in a plane. The tiles are designed so that a plugging increment a coordinate $x$ or $y$ depending on the position/axis of the side. The purpose of this feature is to simulate a shifting of tile on a plane so that two tiles cannot connect on two sides at the same time.

Because of the symbols $\circ$ and $\bullet$, we have to use the constellation $\Phi_{env}^\tau$ as an intermediate for the connexion of two tile sides. We consider a tile $t_i \in \mathtt{dom}(\alpha)$. We start with $t_i^\bigstar$. Assume $t_i$ can be connected to $k$ other tiles in $\mathtt{dom}(\alpha)$. They can only be connected through $\Phi_{env}^\tau$ by their connectable sides. Their glue type and strength for the connected sides have to match because of the shared variables for opposite sides in $\Phi_{env}^\tau$.
All other unused sides of the connector star will be plugged by the unary stars used as fillers. By using principles of logic programming, the diagram can only be correct and saturated if the sum of connected sides of $t_i$ is greater or equal to $\tau$ (note that the filled unused sides add $0$ to the sum). The stars sing symbols $add$ and $geq$ are common logic programs, hence their correctness is assumed.

Since all $t_i \in \mathtt{dom}(\alpha)$ satisfy the above property, the two operations have the same dynamics. Moreover, each tile corresponds exactly to a star and each of its sides corresponds to a ray and we have a structural isomorphism between tiles and their translation. It follows that we have a bijection between the set of non-empty finite assemblies constructible from $T$ at temperature $\tau$ and $\csatdiags(T^\bigstar + \Phi^\tau_{env})$.
\end{proof}

\subsection{Properties of constellations and their execution}\label{subsec:propexec}

In this section, few results of the execution are detailed. Firstly, our model is Turing-complete, which is not too surprising since it is very close to logic programming which is itself known to be Turing-complete (especially through Horn clauses \cite{horn1951sentences, tarnlund1977horn}) but also able to simulate the aTAM which is also Turing-complete \cite[Section 3.2.5]{winfree1998algorithmic}\cite[Section 2]{woods2015intrinsic}.

Borrowing terminology from rewriting and the $\lambda$-calculus, we define the strong normalisation which corresponds to termination of the execution and the confluence asserting that it is possible to focus on a specific set of colours during the execution with no impact on the result, \ie order is irrelevant.

\begin{prop}[Turing-completeness]
\label{prop:turing}
The stellar resolution is Turing-complete.
\end{prop}
\begin{proof}
Consequence of Theorem~\ref{thm:simturing}. Although we can encode Turing machines, the stellar resolution is actually ``stronger'' but for wrong reasons: the ability to compute infinite normal forms. In particular, it is possible to construct infinite non-uniform families of boolean circuits which are known to be theoretically able to decide any language but without concrete implementation of how such families work (for that reason, we usually require families to be \emph{uniform}, \ie that they can be generated by a Turing machine). This is not a problem since we are usually interested in finite constellations and finite normal forms. 
\end{proof}

\begin{defi}[Strong normalisation]
\label{def:strongnormal}
A constellation $\Phi$ is \emph{strongly normalising} \wrt a set of colours $A \subseteq C$ if and only if $\exec_A(\Phi)$ is a finite constellation (or equivalently that $\csatdiags_A(\Phi)$ is finite). We write $|\exec_A(\Phi)| < \infty$ (or $|\csatdiags_A(\Phi)| < \infty$) in this case. When $A = C$, we simply say that $\Phi$ is strongly normalising and omit to write $A$.
\end{defi}

The shape of $\dgraph[A]{\Phi}$ for a constellation $\Phi$ contains a lot of information about $\exec_A(\Phi)$. By observing the shape of constellations in \Cref{sec:illustrations}, we observe that only cycles make iteration possible and that several rays $\alpha$-unifiable with the same single ray are linked to a non-determinism creating diagrams in parallel. However, duplication of stars can still occur without cycle nor non-determinism. Since the relationship between the structure of $\dgraph[A]{\Phi}$ and the computational behaviour of $\exec_A(\Phi)$ is a bit complex, we suggest few structural classes of constellations and establish theorems which will be useful to reason with constellations.

\begin{defi}[Properties of constellation]
\label{def:propconst}
A constellation $\Phi$ is:
\begin{itemize}
    \item \emph{exact} if all equations induced by the edges of $\dgraph{\Phi}$ are of the shape $t \eqq t$;
    \item \emph{acyclic} when $\dgraph{\Phi}$ is acyclic and otherwise it is \emph{cyclic};
    \item \emph{connected} when $\dgraph{\Phi}$ is connected;
    \item \emph{ambivalent} if the ray linking graph $\mathtt{RLG}{\Phi}$ (\cf Definition~\ref{def:diagram}) has \emph{ambivalent links} which are links between several vertices $v_1, ..., v_n$ (with $n>1$) and a same vertex $v$. Otherwise, it is \emph{monovalent}. In case it is ambivalent it is called:
    \begin{itemize}
        \item \emph{replicating} if we can only construct saturated diagram containing $v_1, ..., v_n$, meaning that any duplications of stars remains in the same diagram;
        \item \emph{branching} or \emph{non-deterministic} otherwise, meaning that duplications occur in ``parallel universes'';
    \end{itemize}
    \item \emph{deterministic} when it is either monovalent or replicating.
\end{itemize}
All the definitions can be naturally parametrised with a set of colours $A \subseteq C$.
\end{defi}

\begin{exas}
We illustrate the properties defined above.
\begin{itemize}
    \item The constellation $\Phi_\nat^{n+m}$ of Example~\ref{ex:constellation} is connected, cyclic and non-deterministic. The middle star handles recursion but the construction of diagrams can either continue or exit the loop.
    \item $[+a(X), +a(X)]+[-a(X), -a(X), X]$ is exact, connected, cyclic and ambivalent.
    \item $[X, -c(X)]+[+c(f(Y))]+[+c(g(Y))]$ is acyclic, connected, and non-deterministic. The ray $-c(X)$ has two independent choices and leads to the formation of two diagrams.
    \item $[+a(\lc), +a(\rc)]+[+b(\lc), +b(\rc)]+[-a(X), -b(X)]$ is connected, cyclic and replicating. Two choices are possible for the negative rays but all the stars can appear in the same diagram by duplicating $[-a(X), -b(X)]$ and connecting the $\lc$ (\resp $\rc$) together.
    \item The disjoint union of two above constellations gives a disconnected constellation.
\end{itemize}
\end{exas}

\begin{lem}[Termination of acyclic constellations]
\label{lem:acyclicity}
If a constellation $\Phi$ is acyclic \wrt $A \subseteq C$ then $|\csatdiags_A(\Phi)| < \infty$.
\end{lem}
\begin{proof}
Assume $\dgraph[A]{\Phi}$ is acyclic and consider a diagram $\delta : D_\delta \rightarrow \dgraph[A]{\Phi}$. It must be injective on the vertices, \ie for $v, v' \in D_\delta$ if $v \neq v'$ then $\delta(v) \neq \delta(v')$, meaning that $v$ and $v'$ do not correspond to dupliations of some star in $\dgraph[A]{\Phi}$.
Hence, the vertices of $V^{D_\delta}$ are uniquely taken from $V^{\dgraph[A]{\Phi}}$ and since stars have finitely many rays which must be uniquely connected, there are finitely many edges. There are only finitely many graphs we can construct with finitely many vertices and edges and in particular $\csatdiags_A(\Phi)$ is finite.
\end{proof}

\begin{lem}[Uniqueness]
\label{lem:uniqueness}
Let $\Phi$ be a constellation and $A \subseteq C$ a set of colours. If $\Phi$ is acyclic, connected and deterministic constellations \wrt $A$ then $|\satdiags_A(\Phi)| = 1$ (and $|\exec_A(\Phi)| \leq 1$).
\end{lem}
\begin{proof}
Since $\Phi$ is acyclic, by Lemma~\ref{lem:acyclicity}, we have $|\csatdiags_A(\Phi)| < \infty$ and there is no loop in $\dgraph[A]{\Phi}$ and since it is connected, it has the shape of a tree.
Because $\dgraph[A]{\Phi}$ is a tree, a saturated diagram for a connected constellation must be maximal and include all vertices and rays of $\dgraph[A]{\Phi}$.
A deterministic constellation is either monovalent or replicating. If it is monovalent then there is at most one choice of connexion for a ray. By choosing all these unique connexions we obtain a unique diagram. If it replicating, there is few duplications of stars but the whole forms a unique diagram as well (it is only finite duplication and the only way to duplicate because there is no loop). Hence, we have $|\satdiags_A(\Phi)| = 1$. Depending on if this diagram obtained in both cases is correct or not there is at most one correct saturated diagram. Hence $|\csatdiags_A(\Phi)| = 1$ and $|\exec(\Phi)| \leq 1$.
\end{proof}

\begin{lem}[Exactness]
\label{lem:exactness}
Let $\Phi$ be a constellation which is exact \wrt a set of colours $A \subseteq C$. We have $\diags_A(\Phi) = \csatdiags_A(\Phi)$.
\end{lem}
\begin{proof}
Let $\Phi$ be an exact constellation. By the definition of diagram (\cf Definition~\ref{def:diagram}), all equations $t \eqq u$ induced by diagrams are renamings of equations induced by the edges $e \in V^{\dgraph[A]{\Phi}}$, \ie they are of the shape $\alpha_v t' \eqq \alpha_{v'} u'$ for $v, v' \in V^{G_\delta}$. Assume $t' = u'$.
For any renamings $\alpha_1, \alpha_2$ we have $\alpha_1 t$ which is $\alpha$-unifiable with $\alpha_2 t$. We remark that $\solution{\mathcal{P}(\delta)}$ must always be a renaming. This makes $\delta$ correct. Hence, we have $\diags_A(\Phi) \subseteq \csatdiags_A(\Phi)$. By definition, we also have $\csatdiags_A(\Phi) \subseteq \diags_A(\Phi)$.
\end{proof}

\begin{lem}[Independence of connected components]
\label{lem:independence}
Let $\Phi$ be a constellation. If $\dgraph[A]{\Phi}$ has $n$ connected component corresponding to the subconstellations $\Phi_1, ..., \Phi_n \subseteq \Phi$, then $\exec_A(\Phi) = \bigcup_{i=1}^n \exec_A(\Phi_i)$.
\end{lem}
\begin{proof}
Each connected component $G_i \subseteq \dgraph[A]{\Phi}$ constitutes a constellation $\Phi_i$. When we execute $\Phi$, we form diagrams following the connexions of $\dgraph[A]{\Phi}$.
Since a diagram has to be connected and that no edge link the $G_i$ in $\dgraph[A]{\Phi}$, we necessarily have $\csatdiags_A(\Phi) = \csatdiags_A(\Phi_1) \cup ... \cup \csatdiags_A(\Phi_n)$, hence $\exec_A(\Phi) = \bigcup_{i=1}^n \exec_A(\Phi_i)$.
\end{proof}

\definecolor{BlueConstellation}{RGB}{206,241,253}
\definecolor{RedConstellation}{RGB}{255,172,172}
\begin{figure}
    \begin{center}
    \begin{minipage}{0.3\textwidth}
        \begin{tikzpicture}
            \node[dot, inner sep=0.5em, fill=MidnightBlue] at (0, -0.25) (a) {};
            
            \node[dot, inner sep=0.5em, fill=MidnightBlue] at (2, 0.5) (f) {};
            
            \node[dot] at (-0.5, 1) (k) {};
            \node[dot] at (0.5, 0.5) (l) {};
            \node[dot] at (0.75, 1.25) (m) {};
            \node[dot] at (1.25, -0.75) (n) {};
            \node[dot] at (2.25, -1.25) (o) {};
            \node[dot] at (-1, -1.25) (p) {};
            \node[dot] at (2.75, -0.75) (q) {};
            \node[dot] at (-1, 0.15) (r) {};
            \node[dot] at (1, -1.25) (s) {};
            
            \draw[Bittersweet] (k) -- (l);
            \draw[Bittersweet] (m) -- (l);
            \draw[Bittersweet] (f) -- (l);
            \draw[Bittersweet] (p) -- (a);
            \draw[Bittersweet] (r) edge[bend left] (a);
            \draw[Bittersweet] (r) edge[bend right] (a);
            \draw[Bittersweet] (n) -- (a);
            \draw[Bittersweet] (n) -- (l);
            \draw[Bittersweet] (o) -- (q);
            \draw[Bittersweet] (q) -- (f);
            \draw[Bittersweet] (s) -- (a);
            \draw[Bittersweet] (s) -- (n);
            \draw[Bittersweet] (m) -- (f);
            \draw[Bittersweet] (n) -- (o);
        \end{tikzpicture}
    \end{minipage}
    ${\overset{\exec_B(\Phi)}{\Longleftarrow}}\qquad$
    \begin{minipage}{0.3\textwidth}
        \begin{tikzpicture}
            \node[dot] at (0, 0) (a) {};
            \node[dot] at (-0.5, -0.5) (b) {};
            \node[dot] at (0.5, -0.5) (c) {};
            \node[dot] at (-0.25, -1) (d) {};
            \node[dot] at (0.25, -1) (e) {};
            
            \node[dot] at (2, 1) (f) {};
            \node[dot] at (1.5, 0.5) (g) {};
            \node[dot] at (2.5, 0.5) (h) {};
            \node[dot] at (1.75, 0) (i) {};
            \node[dot] at (2.25, 0) (j) {};
            
            \node[dot] at (-0.5, 1) (k) {};
            \node[dot] at (0.5, 0.5) (l) {};
            \node[dot] at (0.75, 1.25) (m) {};
            \node[dot] at (1.25, -0.75) (n) {};
            \node[dot] at (2.25, -1.25) (o) {};
            \node[dot] at (-1, -1.25) (p) {};
            \node[dot] at (2.75, -0.75) (q) {};
            \node[dot] at (-1, 0.15) (r) {};
            \node[dot] at (1, -1.25) (s) {};
            
            \draw[MidnightBlue] (a) -- (b);
            \draw[MidnightBlue] (a) -- (c);
            \draw[MidnightBlue] (a) -- (d);
            \draw[MidnightBlue] (a) -- (e);
            
            \draw[MidnightBlue] (f) -- (g);
            \draw[MidnightBlue] (f) -- (h);
            \draw[MidnightBlue] (g) -- (i);
            \draw[MidnightBlue] (f) -- (j);
            
            \draw[Bittersweet] (k) -- (l);
            \draw[Bittersweet] (m) -- (l);
            \draw[Bittersweet] (g) -- (l);
            \draw[Bittersweet] (p) -- (b);
            \draw[Bittersweet] (r) -- (b);
            \draw[Bittersweet] (r) -- (a);
            \draw[Bittersweet] (n) -- (e);
            \draw[Bittersweet] (n) -- (l);
            \draw[Bittersweet] (o) -- (q);
            \draw[Bittersweet] (q) -- (j);
            \draw[Bittersweet] (s) -- (e);
            \draw[Bittersweet] (s) -- (n);
            \draw[Bittersweet] (m) -- (f);
            \draw[Bittersweet] (n) -- (o);
        \end{tikzpicture}
    \end{minipage}
    \end{center}
    \caption{Partial execution acts as a partial diagram contraction, which is only possible when $\Phi$ and $\Phi'$ do not act on a same variable. The ``blow-up" obtained by inverting the execution preserves the connexions between rays.}
    \label{fig:ppe}
\end{figure}
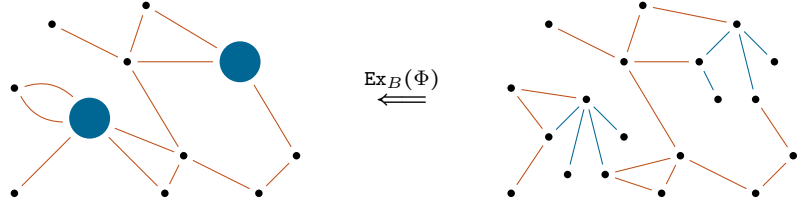

\begin{figure}
    \begin{tikzpicture}
        \node at (0, 0) (c1) {$\Phi = [X, +c(X)]+[-c(\lc \cdot X)]$};
        \node at (6, 0) (c3) {$[-c(\rc \cdot X)] = \Phi'$};
        \draw[dotted] (c1.120) edge[bend left] (c1.12);
        \draw[dotted] (c1.120) edge[bend left] (c3.110);
    \end{tikzpicture}
    \caption{Counter-example for partial pre-execution. We have $\exec_{\{c\}}(\Phi) = [\lc \cdot X]$ and $\exec_{\{c\}}(\exec_{\{c\}}(\Phi) \uplus \Phi') = [-c(\rc \cdot X)]+[\lc \cdot X]$, but $\exec_{\{c\}}(\Phi \uplus \Phi') = [\lc \cdot X]+[\rc \cdot X]$ which is different. Notice that both $[-c(\lc \cdot X)]$ and $[-c(\rc \cdot X)]$ needs $[X, +c(X)]$ but when executing $\Phi$, $\Phi'$ cannot be connected to it anymore.}
    \label{fig:ceppe}
\end{figure}

An important result is the possibility of executing only some colours on some constellations first then the others without any effect on the normal form, that is $\exec_{A \cup B}(\exec_B(\Phi) \uplus \Phi') = \exec_{A \cup B}(\Phi \uplus \Phi')$ for some set of colours $A$ and $B$. However, this is not valid in general as presented in Figure~\ref{fig:ceppe}. The problem is that stars from two disjoint constellations can alter a same variable and a partial execution will erase some potential connexions which were present. This problem is reminiscent of the idea of \emph{mutual exclusion} in concurrent programming \cite{dijkstra2001solution}: the constellations $\Phi$ and $\Phi'$ can modify a same variable $x$ but when executing $\Phi$ and accessing $x$, we may lose an access to $x$ which is still required by $\Phi'$.

This property of partial pre-execution is necessary in order to obtain a result of confluence and associativity of execution, hence a model of linear logic. We need to design a precondition for which these properties are valid and from which it is possible to express logic.

There are several possible choices. A simple choice is to reason on the accessibility of variables in a dependency graph. We do not want a variable to be accessible from two different constellations such that one is pre-executed before the other. For instance, in Figure~\ref{fig:ceppe}, the variable $X$ of $[X, +c(X)]$ is accessible both from $\Phi$ and $\Phi'$.

\begin{defi}[Shared variables]
\label{def:sharedvar}
Variables are written $X^i_j$ where $i$ is an index of star and $j$ an index of ray within that star. We write $\mathtt{acc}^A_{\Phi_\omega}(X^i_j, \Phi)$ for a constellation $\Phi \subseteq \Phi_\omega$ and a set of colours $A \subseteq C$ when there is an edge path $e_1, ..., e_n$ in $\dgraph[A]{\Phi_\omega}$ from some $\phi \in \Phi$ to $\Phi_\omega[i]$ such that $\ell(e_n) = r \match r'$ with $r' \in \Phi_\omega[i]$ and $X \in \vars{r'}$.

We define the \emph{set of variables shared by two constellations} $\Phi_1$ and $\Phi_2$ \wrt a set of colours $A \subseteq C$ as the set $\Phi_1 \ucap_A \Phi_2$ such that we have $X^i_j \in \Phi_1 \ucap_A \Phi_2$ when $\mathtt{acc}_{\Phi_1 \uplus \Phi_2}^A(X^i_j, \Phi_1)$ and $\mathtt{acc}_{\Phi_1 \uplus \Phi_2}^A(X^i_j, \Phi_2)$. We generalise the notation to the set of variables shared by $n$ constellations with the associative notation $\Phi_1 \ucap_A \hdots \ucap_A \Phi_n := \bigcap_{1 \leq i,j \leq n} \Phi_i \ucap_A \Phi_j$.
\end{defi}

\begin{prop}[Commutativity and associativity of shared variables]
\label{prop:sharedvar}
For any constellations $\Phi_1, \Phi_2$ and $\Phi_3$ and a set of colours $A \subseteq C$, we have $\Phi_1 \ucap_A \Phi_2 = \Phi_2 \ucap_A \Phi_1$ and if $\sigma$ is any permutation on $\{1, 2, 3\}$, we have $\Phi_1 \ucap_A \Phi_2 \ucap_A \Phi_3 = \Phi_{\sigma(1)} \ucap_A \Phi_{\sigma(2)} \ucap_A \Phi_{\sigma(3)}$.
\end{prop}
\begin{proof}
We obviously have $X \in \Phi_1 \ucap_A \Phi_2$ if and only if $X \in \Phi_2 \ucap_A \Phi_1$ because in both cases, $X$ is still accessible from both $\Phi_1$ and $\Phi_2$. The same reasoning holds for the associativity.
\end{proof}

\begin{lem}[Partial pre-execution]
\label{lem:ppe}
Let $\Phi$ and $\Phi'$ be constellations and $A, B \subseteq C$ be sets of colours such that $\Phi \ucap_{A \cup B} \Phi' = \emptyset$. We have $\exec_{A \cup B}(\exec_B(\Phi) \uplus \Phi') = \exec_{A \cup B}(\Phi \uplus \Phi')$.
\end{lem}
\begin{proof}
Assume we have a diagram $\delta^{A \cup B}_i \in \csatdiags_{A \cup B}(\exec_B(\Phi) \uplus \Phi')$. It is constructed by connecting the stars $\phi'_j$ of $\Phi'$ with stars $\phi^B_k$ of $\exec_B(\Phi)$.
These stars $\phi^B_k$ of $\exec_B(\Phi)$ come from diagrams $\delta^B_k \in \csatdiags_{B}(\Phi)$. We can perform a ``blow-up`` (\cf\Cref{fig:ppe}) on $\exec_B(\Phi)$ by replacing the stars $\phi^B_k$ by their corresponding diagram $\delta^B_k$. In some sense, we reversed the execution from $\exec_B(\Phi)$ to $\csatdiags_B(\Phi)$.
This is only possible because we have $\Phi \ucap_{A \cup B} \Phi' = \emptyset$, meaning that the stars of $\Phi$ and of $\Phi'$ cannot interfere by acting on a same variable in a same star. Hence, the execution of $\Phi$ makes diagrams in which variables are independent of the ones of $\Phi'$. Otherwise, some connexions could disappear (as in Figure~\ref{fig:ceppe}) and we would not preserve all connexions allowing this inversion of execution.
We obtain diagrams $\varphi(\delta^B_k)$ corresponding to diagrams $\delta^B_k$ extended with stars of $\Phi'$ in exactly the same way as how $\phi^B_k$ can be connected with $\Phi'$.
We have $\varphi(\delta^B_k) \in \csatdiags_{A \cup B}(\Phi \uplus \Phi')$ since it connects stars of both $\Phi$ and $\Phi'$.

It remains to show that $\varphi$ is invertible so that we have an isomorphism between $\csatdiags_{A \cup B}(\exec_B(\Phi) \uplus \Phi')$ and $\csatdiags_{A \cup B}(\Phi \uplus \Phi')$. Assume we have
\[\delta^{A \cup B} \in \csatdiags_{A \cup B}(\Phi \uplus \Phi').\] We would like to define $\varphi^{-1}(\delta^{A \cup B})$.
By the confluence of fusion (which is a consequence of the correspondence between fusion and actualisation, \cf Theorem~\ref{thm:eqfusionactu}), we can contract first the stars coming from $\Phi$ using colours in $B$ and we obtain a diagram $\varphi^{-1}(\delta^{A \cup B})$. We have $\varphi^{-1}(\delta^{A \cup B}) \in \csatdiags_{A \cup B}(\exec_B(\Phi) \uplus \Phi')$.

It is obvious that $\varphi(\varphi^{-1}(\delta)) = \delta$ and $\varphi^{-1}(\varphi(\delta)) = \delta$ because $\varphi$ is defined from the diagrams from which the stars of a normal form come from (star expansion), which is exactly the reverse operation of contracting stars by fusion.
\end{proof}

\begin{thm}[Confluence]
\label{thm:confluence}
For any constellation $\Phi$, and $A, B \subseteq C$ two disjoint sets of colours such that $\Phi \ucap_{A \cup B} \Phi' = \emptyset$, we have $\exec_B(\exec_A(\Phi)) = \exec_{A\cup B}(\Phi) = \exec_A(\exec_B(\Phi))$.
\end{thm}
\begin{proof}
By Lemma~\ref{lem:ppe} with $\Phi' := \emptyset$ (in this case we trivially have $\Phi \ucap_{A \cup B} \emptyset = \emptyset$ which is the required precondition) we have $\exec_{A \cup B}(\exec_B(\Phi)) = \exec_{A\cup B}(\Phi)$.
Since $\exec_B(\Phi)$ already uses all colours in $B$, we have $\exec_{A \cup B}(\exec_B(\Phi)) = \exec_A(\exec_B(\Phi))$, hence $\exec_{A}(\exec_B(\Phi)) = \exec_{A\cup B}(\Phi)$.
Since $A \cup B = B \cup A$, we also have $\exec_{A\cup B}(\Phi) = \exec_{B\cup A}(\Phi)$. By using again Lemma~\ref{lem:ppe}, we finally obtain $\exec_{B\cup A}(\Phi) = \exec_B(\exec_A(\Phi))$.
\end{proof}

\begin{rem}
In Girard's first paper on Transcendental Syntax \cite[Section 2.4]{transyn1}, the constellation $\Phi = [+a(X), -a(X), +b(X)]$ is mentioned as a counter-example for the confluence of $\exec$ (which only considers tree-shaped diagrams). Here, we have $\exec_{\{a\}}(\exec_{\{b\}}(\Phi)) = \exec_{\{b\}}(\exec_{\{a\}}(\Phi)) = \emptyset$ (because no saturated diagram on $a$ nor on $b$ can be constructed).
Our understanding of Girard's failure comes from his limitation to strongly normalising constellations, so that $\exec_{\{a\}}(\Phi)$ was not defined because of the cyclic dependence between $+a(X)$ and $-a(X)$.

Also remark that $\exec$ is analogous to the computation of all answers we can infer from a logic program, meaning that all possible paths of computation are considered, hence naturally leading to confluence.
\end{rem}

Before ending our computational journey, let us stress (again) the fact that there are several differences between the stellar resolution and approaches in logic programming (although identical objects are used). Our approach is indeed a liberalised variant of first-order resolution but we are not aware of any similar uses of resolution. We suggest some comparisons with other approaches in the literature:
\begin{description}
    \item[Original first-order resolution] It is almost identical. We add unpolarised rays which cannot be connected (it can still be simulated in resolution by using special unused predicates). In resolution, we are usually interested in the reachability of the empty clause ($\emptystar$ in our case) representing a contradiction. In the stellar resolution, it does not have any meaning and we use objects as query-free logic programs. Usual resolution is limited to tree derivations (corresponding to tree-like diagrams) whereas stellar resolution allows cyclic diagrams in order to interpret tiling-based computation. There are graph-based models \cite{sickel1976search, kowalski1975proof, eisinger1991deduction} which are very similar to stellar resolution but they are still different for the reasons mentioned above.
    \item[Horn clauses and logic programming] By logic programming, we mean that we are interested in answering a query represented by a first-order atom (such as in \textsc{Prolog} for instance). In order to answer the query, logic programming use a backward reasoning by going up from the unique conclusion to the premises. The stellar resolution is naturally query-free (although queries can be simulated, there is no such distinguished objects). In particular, we can have several outputs and we do not distinguish between input and output. For instance, if we have a star representing an implication $A \Rightarrow B$, then we can connect a star to the output and only the input will survive. This does not make sense in logic programming because a direction output$\rightarrow$inputs is imposed in the inference.
    \item[Stable model semantics] There are several languages based on stable model semantics such as disjunctive logic programming \cite{minker1994overview, lobo1991semantics} itself based on a subset of \textsc{Prolog} called \textsc{Datalog}. The notion of stable model is also the basis of answer set programming (ASP) \cite{gelfond2008answer, eiter2009answer}. In these languages, a primitive handling of logical negation is used whereas we want our model to be purely computational, without any reference to logic.
\end{description}

\section{Emergence of proofs}\label{sec:proofs}

In order to reconstruct logic, it is natural to start from linear logic \cite{linearlogic} which decomposes both classical and intuitionistic logic. We choose to work with Girard's representation of proofs called ``proof-nets"\footnote{Our definitions differ from the usual definitions of the literature but are more convenient for the results presented in this paper. }. We begin by defining the fragment of linear logic we work with. Useful definitions about hypergraphs are recalled in \Cref{sec:hypergraphs}.

\subsection{Multiplicative proofs}\label{subsec:mll}

\begin{figure}[t]
    \centering
    \begin{subfigure}{0.8\textwidth}
        \begin{equation*}\tag{$\fmll$}
            A, B = X_i \mid X_i^\bot \mid A \otimes B \mid A \parr B \qquad i \in \nat
        \end{equation*}
        \caption{MLL Formulas.}
        \label{subfig:fmll}
    \end{subfigure}

    \bigskip
    \begin{subfigure}{.8\textwidth}
        \scalebox{0.9}{
        \begin{prooftree}
            \infer0[ax]{ \vdash A, A^\bot }
        \end{prooftree}
        \qquad
        \begin{prooftree}
            \hypo{\vdash \Gamma, A \quad \vdash \Delta, A^\bot}
            \infer1[cut]{\vdash \Gamma, \Delta}
        \end{prooftree}
        \qquad
        \begin{prooftree}
            \hypo{ \vdash \Gamma, A \quad \vdash \Delta, B }
            \infer1[$\otimes$]{ \vdash \Gamma, \Delta, A \otimes B }
        \end{prooftree}
        \qquad
        \begin{prooftree}
            \hypo{\vdash \Gamma, A, B}
            \infer1[$\parr$]{\vdash \Gamma, A \parr B}
        \end{prooftree}}
        \caption{MLL sequent calculus rules.}
        \label{subfig:mllsequent}
    \end{subfigure}

    \bigskip
    \begin{subfigure}{.6\textwidth}
        \centering
        \scalebox{0.8}{
        \begin{tikzpicture}
          \node[dot] at (0, 0) (a) {};
          \node[dot] at (2, 0) (ad) {};
          \node at (1, 0.75) (ax) {ax};
          \draw[-latex, rounded corners=5pt] (ax) -| (a);
          \draw[-latex, rounded corners=5pt] (ax) -| (ad);
          \node at (1, -1) (label) {Axiom};
        \end{tikzpicture}
        \qquad
        \begin{tikzpicture}
          \node[dot] at (0, 0) (a) {};
          \node[dot] at (2, 0) (ad) {};
          \node at (1, -0.75) (cut) {cut};
          \draw[-latex, rounded corners=5pt] (a) |- (cut);
          \draw[-latex, rounded corners=5pt] (ad) |- (cut);
          \node at (1, -1.5) (label) {Cut};
        \end{tikzpicture}
        \qquad
        \begin{tikzpicture}
          \node[dot] at (-0.75, 0.75) (a) {};
          \node[dot] at (0.75, 0.75) (b) {};
          \node at (0, 0) (tens) {$\otimes$};
          \node[dot] at (0, -0.75) (ab) {};
          \draw[-stealth] (a) -- (tens);
          \draw[-stealth] (b) -- (tens);
          \draw[-stealth] (tens) -- (ab);
          \node at (0, -1.5) (label) {Tensor};
        \end{tikzpicture}
        \qquad
        \begin{tikzpicture}
          \node[dot] at (-0.75, 0.75) (a) {};
          \node[dot] at (0.75, 0.75) (b) {};
          \node at (0, 0) (par) {$\parr$};
          \node[dot] at (0, -0.75) (ab) {};
          \draw[-stealth] (a) -- (par);
          \draw[-stealth] (b) -- (par);
          \draw[-stealth] (par) -- (ab);
          \node at (0, -1.5) (label) {Par};
        \end{tikzpicture}}
        \caption{Links/constructors of proof-structures as hyperedges.}
        \label{subfig:proofstructures}
    \end{subfigure}

    \bigskip
    \begin{subfigure}{1\textwidth}
        \centering
        \begin{minipage}{0.175\textwidth}
            \begin{tikzpicture}
                \node[dot] at (0, 0) (a) {};
                \node[dot] at (1, 0) (ad) {};
                \node at (0.5, 0.75) (ax) {ax};
                \node at (1.75, -0.75) (cut) {cut};
                \node[dot] at (2.5, 0) (a2) {};
                \draw[-latex, rounded corners=5pt] (ax) -| (a);
                \draw[-latex, rounded corners=5pt] (ax) -| (ad);
                \draw[-latex, rounded corners=5pt] (ad) |- (cut);
                \draw[-latex, rounded corners=5pt] (a2) |- (cut);
            \end{tikzpicture}
        \end{minipage}
        $\quad\overset{\axl/\cutl}{\leadsto}\quad$
        \begin{minipage}{0.05\textwidth}
            \begin{tikzpicture}
                \node[dot] at (0, 0) (a) {};
            \end{tikzpicture}
        \end{minipage}
        
        \begin{minipage}{0.3\textwidth}
            \begin{tikzpicture}
                \node[dot] at (-0.75, 0.75) (a) {};
                \node[dot] at (0.75, 0.75) (b) {};
                \node at (0, 0) (tens) {$\otimes$};
                \node[dot] at (0, -0.75) (ab) {};
                \draw[-stealth] (a) -- (tens);
                \draw[-stealth] (b) -- (tens);
                \draw[-stealth] (tens) -- (ab);

                \node[dot] at (2, 0.75) (c) {};
                \node[dot] at (3.5, 0.75) (d) {};
                \node at (2.75, 0) (par) {$\parr$};
                \node[dot] at (2.75, -0.75) (cd) {};
                \draw[-stealth] (c) -- (par);
                \draw[-stealth] (d) -- (par);
                \draw[-stealth] (par) -- (cd);

                \node at (1.25, -1.25) (cut) {cut};
                \draw[-latex, rounded corners=5pt] (ab) |- (cut);
                \draw[-latex, rounded corners=5pt] (cd) |- (cut);
            \end{tikzpicture}
        \end{minipage}
         $\quad\overset{\otimes/\parr}{\leadsto}\qquad$
         \begin{minipage}{0.5\textwidth}
            \begin{tikzpicture}
              \node[dot] at (0, 0) (a) {};
              \node[dot] at (1.5, 0) (b) {};
              \node[dot] at (3, 0) (c) {};
              \node[dot] at (4.5, 0) (d) {};
              \node at (0.75, -0.5) (cut1) {cut};
              \node at (3, -1) (cut2) {cut};
              \draw[-latex, rounded corners=5pt] (a) |- (cut1);
              \draw[-latex, rounded corners=5pt] (c) |- (cut1);
              \draw[-latex, rounded corners=5pt] (b) |- (cut2);
              \draw[-latex, rounded corners=5pt] (d) |- (cut2);
           \end{tikzpicture}
        \end{minipage}
        \caption{Cut-elimination reductions. The $\axl/\cutl$ case is a graph contraction and $\otimes/\parr$ is a rewiring.}
        \label{subfig:cutelimination}
    \end{subfigure}
\caption{Syntax of Multiplicative Linear Logic (MLL).}
\label{fig:mll}
\end{figure}

Multiplicative linear logic (MLL) is a fragment of linear logic \cite{linearlogic} restricted to the tensor $\otimes$ and par $\parr$ connectives which are respectively a sort of conjunction and disjunction. The set $\fmll$ of MLL formulas is defined by the grammar of \Cref{subfig:fmll}. Linear negation $(\cdot)^\bot$ is extended to formulas by involution and De Morgan laws: $X_i^{\bot\bot} = X_i$, $(A \otimes B)^\bot = A^\bot \otimes B^\bot$, and $(A \parr B)^\bot = A^\bot \parr B^\bot$.

MLL proofs can be written in the traditional sequent calculus fashion by constructing trees using the set of rules shown in \Cref{subfig:mllsequent}. These rules use sequents $\vdash \Gamma$ stating the provability of a set of formulas $\Gamma \subseteq \fmll$.
Instead of the MLL sequent calculus, we choose to work with Girard's proof-nets, a ``parallel" syntax for proofs akin to Gentzen's natural deduction which captures the essence of proofs by forgetting the order of rules. In order to define proof-nets, we first define \emph{proof-structures} which are purely computational and structural objects with no logical meaning. They represent skeletons for proofs. In the same spirit as Girard's ludics \cite{girard2001locus}, ``only location matters" at this point. The idea is that when considering the structure of proofs, formulas are nothing more than decorative labels which can be forgotten.
In this syntax, we consider directed hypergraphs constructed with the hyperedges of \Cref{subfig:proofstructures}.

\begin{defi}[Proof-structure, \Cref{fig:proofstructure}]
A \emph{proof-structure} is defined by a tuple $\mathcal{S} = (V,E,\ein{},\eout{},\ell_E)$ where $(V,E,\ein{},\eout{})$ is a directed hypergraph and $\ell_E: E\rightarrow\{\otimes,\parr,\axl,\cutl\}$ is a labelling map on hyperedges.
A proof-structure is subject to these additional constraints:
\begin{itemize}[label=$\triangleright$]
    \item the hyperedges satisfy the arities and labelling constraints shown in \Cref{subfig:proofstructures};
    \item each vertex must be the target of exactly one hyperedge, and the source of at most one hyperedge;
    \item cut hyperedges must connect either:
    \begin{itemize}
    	\item the conclusion of a $\parr$ hyperedge with the conclusion of a $\otimes$ hyperedge, or
    	\item two atoms.
    \end{itemize}
\end{itemize}
\end{defi}

\begin{conv}[Left and right sources]
For practical purposes, the sources of hyperedges are ordered, and we will talk about the ``left" and ``right" sources since there are never more than two; illustrations in \Cref{subfig:proofstructures} implicitly represent the left (\resp right) source on the left (\resp right).
\end{conv}

\begin{nota}[Axioms and cuts]
Let $\mathcal{S}$ be a proof-structure. We write $\axioms{\mathcal{S}}$ (\resp $\cuts{\mathcal{S}}$) the set of axioms (\resp cut) hyperedges in $\mathcal{S}$.
Given $e\in\axioms{\mathcal{S}}$ ($e\in\cuts{\mathcal{S}})$), we write $\concl{e}$ and $\concr{e}$ the left and right conclusion (\resp sources) of $e$ respectively.
\end{nota}

\begin{nota}[Conclusions and atoms]
The \emph{conclusions} of $\mathcal{S}$ are defined by the set $\conclu{\mathcal{S}} = \{v \in V \mid \text{there is no } e \in E \text{ such that } v \in \ein{e}\}$.
Similarly, the \emph{atoms} of $\mathcal{S}$ are defined by the set $\atoms{\mathcal{S}} = \{v \in V \mid \exists e \in \axioms{\mathcal{S}} \text{ such that } v \in \eout{e}\}$. They are conclusions of axiom hyperedges.
\end{nota}

The hyperedges of \Cref{subfig:proofstructures} are the elementary bricks for proof-structures. Notice that the $\otimes$ and $\parr$ hyperedge are structurally identical and that their label is irrelevant. As a first step, we will consider them identical. It is only later, in \Cref{subsec:correctness}, that these two constructions will be distinguished by the logical meaning we associate to them.

\begin{figure}
    \scalebox{0.75}{
    \begin{tikzpicture}
        \node at (-0.75, 0.75) (a) {1};
        \node at (0.75, 0.75) (b) {2};
        \node at (0, 0) (par) {$\parr$};
        \node at (0, -1) (ab) {7};
        \draw[-stealth] (a) -- (par);
        \draw[-stealth] (b) -- (par);
        \draw[-stealth] (par) -- (ab);

        \node at (1.75, 0.75) (f1) {3};
        \node at (6.25, 0.75) (f2) {6};

        \node at (3.25, 0.75) (c) {4};
        \node at (4.75, 0.75) (d) {5};
        \node at (4, 0) (tens) {$\otimes$};
        \node at (4, -1) (cd) {8};
        \draw[-stealth] (c) -- (tens);
        \draw[-stealth] (d) -- (tens);
        \draw[-stealth] (tens) -- (cd);

        \node at (1.75, -1.75) (cut) {cut};
        \draw[-latex, rounded corners=5pt] (ab) |- (cut);
        \draw[-latex, rounded corners=5pt] (cd) |- (cut);
        \node at (0, 1.5) (ax) {ax};
        \draw[-latex, rounded corners=5pt] (ax) -| (a);
        \draw[-latex, rounded corners=5pt] (ax) -| (b);
        \node at (2.5, 1.5) (ax) {ax};
        \draw[-latex, rounded corners=5pt] (ax) -| (f1);
        \draw[-latex, rounded corners=5pt] (ax) -| (c);
        \node at (5.5, 1.5) (ax) {ax};
        \draw[-latex, rounded corners=5pt] (ax) -| (f2);
        \draw[-latex, rounded corners=5pt] (ax) -| (d);
    \end{tikzpicture}}
    \caption{Example of unlabelled proof-structure with vertices in $\nat$.}
    \label{fig:proofstructure}
\end{figure}
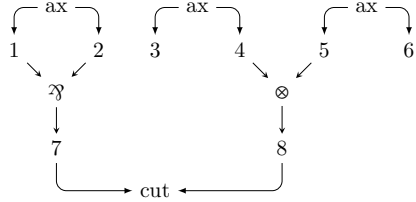

The cut-elimination procedure on proof-structures (corresponding to program execution) is defined as a graph-rewriting system on proof-structures, defined by the two rewriting rules in \Cref{subfig:cutelimination}.

\begin{figure}
    \begin{minipage}{0.25\textwidth}
        \begin{prooftree}
            \hypo{\vdash \Gamma}
            \hypo{\vdash \Delta}
            \infer2[mix]{\vdash \Gamma, \Delta}
        \end{prooftree}
    \end{minipage}
    \begin{minipage}{0.4\textwidth}
        \begin{prooftree}
            \infer0[ax]{\vdash X_1, X_1^\bot}
            \infer0[ax]{\vdash X_2, X_2^\bot}
            \infer2[mix]{\vdash X_1^\bot, X_2^\bot, X_1, X_2}
        \end{prooftree}
    \end{minipage}
    \caption{The MIX rule and a of sequent calculus proof of $X_1 \otimes X_2 \multimap X_1 \parr X_2$ using the MIX rule.}
    \label{fig:mix}
\end{figure}

There exists a remarkable extension of MLL with a rule called MIX (\cf\Cref{fig:mix}), initially studied by Fleury and Rétoré \cite{fleury1994mix}. This rule corresponds to the axiom scheme $A\otimes B\multimap A\parr B$ and constitutes, together with the other rules of MLL, a new proof system called MLL+MIX. Beside this new rule, MLL+MIX works with the same formulas as MLL. In particular, all MLL sequent calculus proofs are MLL+MIX sequent calculus proofs as well.

We now would like to define the underlying proof-structure of an MLL+MIX sequent calculus proof. In order to do so, we define a labelling of the vertices of proof-structures by formulas in order to make proof-structures look like actual proofs. By doing so, we already give a little bit of meaning to the purely computational proof-structures but which is only superficial for the moment.

\begin{defi}[Labelled proof-structure]
A \emph{labelled proof-structure} is a tuple \[\mathcal{S} = (V,E,\ein{},\eout{},\ell_V,\ell_E)\] where $(V,E,\ein{},\eout{},\ell_E)$ is a proof-structure and $\ell_V : V \rightarrow \fmll$ is a function labelling vertices of $V$ by formulas.

We write $\vdash \mathcal{S} : \Gamma$ for a set of formula $\Gamma := \{\ell_V(v) \mid v \in \conclu{\mathcal{S}}\}$ in order to specify the formulas associated to the conclusions of $\mathcal{S}$.
\end{defi}

\begin{figure}
\scalebox{0.9}{
\begin{minipage}{0.15\textwidth}
    \begin{prooftree}
        \infer0[ax]{\vdash A, A^\bot}
    \end{prooftree}
\end{minipage}}
$\rightarrow^{\interp{\cdot}}$
\scalebox{0.75}{
\begin{minipage}{0.15\textwidth}
    \begin{tikzpicture}
      \node at (0, 0) (a) {$A$};
      \node at (1, 0) (ad) {$A^\bot$};
      \node at (0.5, 0.75) (ax) {$ax$};
      \draw[-latex, rounded corners=5pt] (ax) -| (a);
      \draw[-latex, rounded corners=5pt] (ax) -| (ad);
    \end{tikzpicture}
\end{minipage}}
\qquad
\scalebox{0.9}{
\begin{minipage}{0.25\textwidth}
    \begin{prooftree}
        \hypo{\pi_1}
        \ellipsis{}{\vdash \Gamma, A}
        \hypo{\pi_2}
        \ellipsis{}{\vdash \Delta, A^\bot}
        \infer2[cut]{\vdash \Gamma, \Delta}
    \end{prooftree}
\end{minipage}}
$\rightarrow^{\interp{\cdot}}$
\scalebox{0.75}{
\begin{minipage}{0.4\textwidth}
    \begin{tikzpicture}
      \draw (-0.75,0.5) rectangle (0.5,1.25);
      \node at (-0.15, 0.9) (label) {$\interp{\pi_1}$};
      \draw (1.5,0.5) rectangle (2.75,1.25);
      \node at (2.15, 0.9) (label) {$\interp{\pi_2}$};
      \node at (-0.5, 0) (gamma) {$\Gamma$};
      \node at (2.5, 0) (delta) {$\Delta$};
      \node at (0, 0) (a) {$A$};
      \node at (2, 0) (ad) {$A^\bot$};
      \draw[-latex] (-0.75, 0.5) -| (a);
      \draw[-latex] (-0.75, 0.5) -| (gamma);
      \draw[-latex] (1.5, 0.5) -| (ad);
      \draw[-latex] (1.5, 0.5) -| (delta);
      \node at (1, -0.75) (cut) {cut};
      \draw[-latex, rounded corners=5pt] (a) |- (cut);
      \draw[-latex, rounded corners=5pt] (ad) |- (cut);
    \end{tikzpicture}
\end{minipage}}

\bigskip
\scalebox{0.9}{
\begin{minipage}{0.18\textwidth}
    \begin{prooftree}
        \hypo{\pi}
        \ellipsis{}{\vdash \Gamma, A, B}
        \infer1[$\parr$]{\vdash \Gamma, A \parr B}
    \end{prooftree}
\end{minipage}}
$\rightarrow^{\interp{\cdot}}$
\scalebox{0.75}{
\begin{minipage}{0.25\textwidth}
    \begin{tikzpicture}
      \draw (-1,1.25) rectangle (1.75,2);
      \node at (0.4, 1.6) (label) {$\interp{\pi}$};
      \node at (1.5, 0.75) (gamma) {$\Gamma$};
      \node at (-0.75, 0.75) (a) {$A$};
      \node at (0.75, 0.75) (b) {$B$};
      \draw[-latex] (-1, 1.25) -| (gamma);
      \draw[-latex] (-1, 1.25) -| (a);
      \draw[-latex] (-1, 1.25) -| (b);
      \node at (0, 0) (par) {$\parr$};
      \node at (0, -0.75) (ab) {$A \parr B$};
      \draw[-stealth] (a) -- (par);
      \draw[-stealth] (b) -- (par);
      \draw[-stealth] (par) -- (ab);
    \end{tikzpicture}
\end{minipage}}
\qquad
\scalebox{0.9}{
\begin{minipage}{0.23\textwidth}
    \begin{prooftree}
        \hypo{\pi_1}
        \ellipsis{}{\vdash \Gamma, A}
        \hypo{\pi_2}
        \ellipsis{}{\vdash \Delta, B}
        \infer2[$\otimes$]{\vdash \Gamma, \Delta, A \otimes B}
    \end{prooftree}
\end{minipage}}
$\rightarrow^{\interp{\cdot}}$
\scalebox{0.75}{
\begin{minipage}{0.25\textwidth}
    \begin{tikzpicture}
      \draw (-1.5,1.25) rectangle (-0.5,2);
      \draw (0.5,1.25) rectangle (1.5,2);
      \node at (-1, 1.65) (label) {$\interp{\pi_1}$};
      \node at (1, 1.65) (label) {$\interp{\pi_2}$};
      \node at (-1.25, 0.75) (gamma) {$\Gamma$};
      \node at (1.25, 0.75) (delta) {$\Delta$};
      \node at (-0.75, 0.75) (a) {$A$};
      \node at (0.75, 0.75) (b) {$B$};
      \draw[-latex] (-1.5, 1.25) -| (gamma);
      \draw[-latex] (-1.5, 1.25) -| (a);
      \draw[-latex] (1.5, 1.25) -| (delta);
      \draw[-latex] (1.5, 1.25) -| (b);
      \node at (0, 0) (par) {$\otimes$};
      \node at (0, -0.75) (ab) {$A \otimes B$};
      \draw[-stealth] (a) -- (par);
      \draw[-stealth] (b) -- (par);
      \draw[-stealth] (par) -- (ab);
    \end{tikzpicture}
\end{minipage}}

\bigskip
\scalebox{0.9}{
\begin{minipage}{0.2\textwidth}
    \begin{prooftree}
        \hypo{\pi_1}
        \ellipsis{}{\vdash \Gamma}
        \hypo{\pi_2}
        \ellipsis{}{\vdash \Delta}
        \infer2[mix]{\vdash \Gamma, \Delta}
    \end{prooftree}
\end{minipage}}
$\rightarrow^{\interp{\cdot}}$
\scalebox{0.75}{
\begin{minipage}{0.3\textwidth}
    \begin{tikzpicture}
      \draw (-1.5,1.25) rectangle (-0.5,2);
      \draw (0.5,1.25) rectangle (1.5,2);
      \node at (-1, 1.65) (label) {$\interp{\pi_1}$};
      \node at (1, 1.65) (label) {$\interp{\pi_2}$};
      \node at (-1.25, 0.75) (gamma) {$\Gamma$};
      \node at (1.25, 0.75) (delta) {$\Delta$};
      \draw[-latex] (-1.5, 1.25) -| (gamma);
      \draw[-latex] (1.5, 1.25) -| (delta);
    \end{tikzpicture}
\end{minipage}}
\caption{Translation of MLL+MIX sequent calculus proofs into labelled proof-structures.}
\label{fig:sequent2structure}
\end{figure}
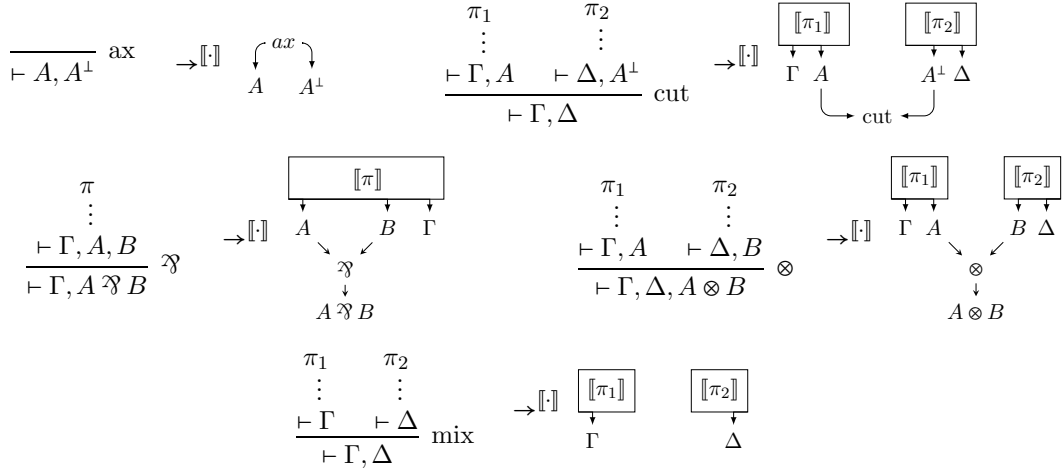

In \Cref{fig:sequent2structure}, we define a translation $\interp{\cdot}$ from MLL+MIX sequent calculus derivations to labelled proof-structures. Notice that this translation is not surjective, and that some proof-structures do not represent sequent calculus proofs. This is tackled by the \emph{correctness criterion}, which characterises those proof-structures that do translate sequent calculus proofs through topological properties and which are considered ``correct". This is discussed in \Cref{subsec:correctness} but for the time being, we give a preliminary definition of proof-net, the proof-structures coming from sequent calculus proofs.

Also notice that the MIX rule corresponds to allowing disjoint union of proof-structures as being ``correct". Although not ``logical" (\ie not coming from MLL sequent calculus which decomposes intuitionistic and classical logic), MLL+MIX proofs keep interesting computational properties which naturally appear in various models of linear logic such as coherence spaces \cite[Chapter 4]{linearlogic}.

\begin{defi}[MLL and MLL+MIX proof-nets]
\label{def:proofnets}
An MLL (\resp MLL+MIX) \emph{proof-net} is a proof-structure $\mathcal{S}$ such that there exists an MLL (\resp MLL+MIX) sequent calculus proof $\pi$ such that $\mathcal{S}=\interp{\pi}$.
\end{defi}

In this paper, we show that both MLL and MLL+MIX can be interpreted in the stellar resolution.

\subsection{Simulation of cut-elimination}\label{subsec:cutelim}

Early investigations on the cut-elimination \cite{goi1}, simplified with Seiller's interaction graphs \cite{seiller2012interaction}, show that cut-elimination can be considered much simpler than the standard graph rewriting of \Cref{subfig:cutelimination}.
The $\otimes/\parr$ cut-elimination rule pushes cuts to the top of the proof-structure (axioms) and the $\axl/\cutl$ rule identifies some atoms by contraction. It shows that we can see a proof-structure as a connexion between a permutation of atoms representing axioms and a partial permutation on atoms representing cuts (\cf\Cref{fig:goi}). The cut-elimination procedure is then seen as a computation of maximal alternating paths between the graph of these two permutations or equivalently as the complete edge contraction of a bipartite graph.

When considering the computational content of proofs, the connectives $\otimes$ and $\parr$ are then irrelevant since no reference to logic exists at this point and that the $\otimes/\parr$ cut-elimination is a simple rewiring. For that reason, the simulation of cut-elimination in the stellar resolution only deals with the translation of axioms and cuts as binary stars with rays representing the \emph{address} of atoms. The interpretation is the same for both MLL and MLL+MIX since they have the same cut-elimination.

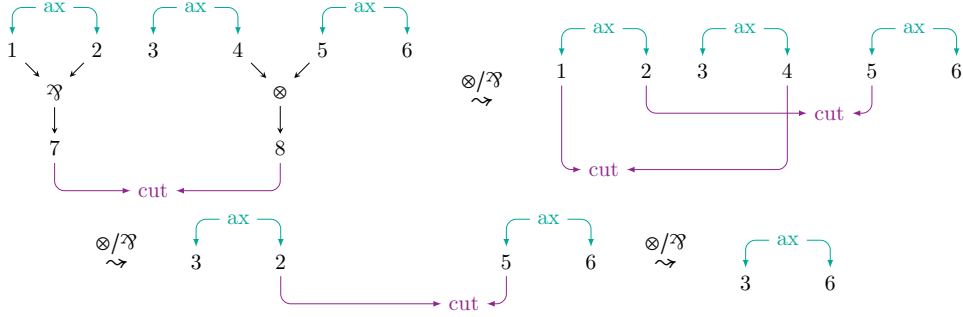
\begin{figure}
    \begin{minipage}{0.4\textwidth}
        \scalebox{0.75}{
        \begin{tikzpicture}
            \node at (-0.75, 0.75) (a) {1};
            \node at (0.75, 0.75) (b) {2};
            \node at (0, 0) (par) {$\parr$};
            \node at (0, -1) (ab) {7};
            \draw[-stealth] (a) -- (par);
            \draw[-stealth] (b) -- (par);
            \draw[-stealth] (par) -- (ab);

            \node at (1.75, 0.75) (f1) {3};
            \node at (6.25, 0.75) (f2) {6};

            \node at (3.25, 0.75) (c) {4};
            \node at (4.75, 0.75) (d) {5};
            \node at (4, 0) (tens) {$\otimes$};
            \node at (4, -1) (cd) {8};
            \draw[-stealth] (c) -- (tens);
            \draw[-stealth] (d) -- (tens);
            \draw[-stealth] (tens) -- (cd);

            \node[Plum] at (1.75, -1.75) (cut) {cut};
            \draw[Plum, -latex, rounded corners=5pt] (ab) |- (cut);
            \draw[Plum, -latex, rounded corners=5pt] (cd) |- (cut);
            \node[JungleGreen] at (0, 1.5) (ax) {ax};
            \draw[JungleGreen, -latex, rounded corners=5pt] (ax) -| (a);
            \draw[JungleGreen, -latex, rounded corners=5pt] (ax) -| (b);
            \node[JungleGreen] at (2.5, 1.5) (ax) {ax};
            \draw[JungleGreen, -latex, rounded corners=5pt] (ax) -| (f1);
            \draw[JungleGreen, -latex, rounded corners=5pt] (ax) -| (c);
            \node[JungleGreen] at (5.5, 1.5) (ax) {ax};
            \draw[JungleGreen, -latex, rounded corners=5pt] (ax) -| (f2);
            \draw[JungleGreen, -latex, rounded corners=5pt] (ax) -| (d);
        \end{tikzpicture}}
    \end{minipage}
    $\overset{\otimes/\parr}{\leadsto}\quad$
    \begin{minipage}{0.45\textwidth}
        \scalebox{0.75}{
        \begin{tikzpicture}
            \node at (-0.75, 0.75) (a) {1};
            \node at (0.75, 0.75) (b) {2};

            \node at (1.75, 0.75) (f1) {3};
            \node at (6.25, 0.75) (f2) {6};

            \node at (3.25, 0.75) (c) {4};
            \node at (4.75, 0.75) (d) {5};

            \node[Plum] at (0, -1) (cut) {cut};
            \draw[Plum, -latex, rounded corners=5pt] (a) |- (cut);
            \draw[Plum, -latex, rounded corners=5pt] (c) |- (cut);

            \node[Plum] at (4, 0) (cut) {cut};
            \draw[Plum, -latex, rounded corners=5pt] (b) |- (cut);
            \draw[Plum, -latex, rounded corners=5pt] (d) |- (cut);

            \node[JungleGreen] at (0, 1.5) (ax) {ax};
            \draw[JungleGreen, -latex, rounded corners=5pt] (ax) -| (a);
            \draw[JungleGreen, -latex, rounded corners=5pt] (ax) -| (b);
            \node[JungleGreen] at (2.5, 1.5) (ax) {ax};
            \draw[JungleGreen, -latex, rounded corners=5pt] (ax) -| (f1);
            \draw[JungleGreen, -latex, rounded corners=5pt] (ax) -| (c);
            \node[JungleGreen] at (5.5, 1.5) (ax) {ax};
            \draw[JungleGreen, -latex, rounded corners=5pt] (ax) -| (f2);
            \draw[JungleGreen, -latex, rounded corners=5pt] (ax) -| (d);
        \end{tikzpicture}}
    \end{minipage}

    $\overset{\otimes/\parr}{\leadsto}\quad$
    \begin{minipage}{0.4\textwidth}
        \scalebox{0.75}{
        \begin{tikzpicture}
            \node at (-0.75, 0.75) (a) {3};
            \node at (0.75, 0.75) (b) {2};

            \node at (6.25, 0.75) (f2) {6};

            \node at (4.75, 0.75) (d) {5};
            \node[Plum] at (4, 0) (cut) {cut};
            \draw[Plum, -latex, rounded corners=5pt] (b) |- (cut);
            \draw[Plum, -latex, rounded corners=5pt] (d) |- (cut);

            \node[JungleGreen] at (0, 1.5) (ax) {ax};
            \draw[JungleGreen, -latex, rounded corners=5pt] (ax) -| (a);
            \draw[JungleGreen, -latex, rounded corners=5pt] (ax) -| (b);
            \node[JungleGreen] at (5.5, 1.5) (ax) {ax};
            \draw[JungleGreen, -latex, rounded corners=5pt] (ax) -| (f2);
            \draw[JungleGreen, -latex, rounded corners=5pt] (ax) -| (d);
        \end{tikzpicture}}
    \end{minipage}
    $\overset{\otimes/\parr}{\leadsto}\quad$
    \begin{minipage}{0.2\textwidth}
        \scalebox{0.75}{
        \begin{tikzpicture}
            \node at (-0.75, 0.75) (a) {3};
            \node at (0.75, 0.75) (b) {6};

            \node[JungleGreen] at (0, 1.5) (ax) {ax};
            \draw[JungleGreen, -latex, rounded corners=5pt] (ax) -| (a);
            \draw[JungleGreen, -latex, rounded corners=5pt] (ax) -| (b);
        \end{tikzpicture}}
    \end{minipage}
    \caption{Cut-elimination for the proof-structure of \Cref{fig:proofstructure} represented as the juxtaposition of two partial permutations or graphs on atoms as suggested in the GoI.}
    \label{fig:goi}
\end{figure}

In order to encode proof-structures, we fix a \emph{basis of representation} $\mathcal{B}$ with variables $V=\{X\}$, colours $C=\{+c, -c, +t, -t\}$, function symbols $F=\{c, t, \lc, \rc, \gc, \cdot, p_A, q_A\}$ for $A \in \fmll$ such that $c, t, p_A$ and $q_A$ are unary, $\cdot$ is binary and the symbols $\lc, \rc$ and $\gc$ are constants. We define $\op(+c) = -c$, $\op(+t) = -t$, $\floor{\pm c} = c$ and $\floor{\pm t} = t$ for $\pm \in \{+, -\}$.

Similarly to unlabelled proof-structures, constellations are purely locative: only the locations appearing in a proof-structure $\mathcal{S}$ are translated, without regard to labels. We would like to associate a unique address in $\terms{\mathbb{B}}$ to the atoms $v \in \atoms{\mathcal{S}}$ of a proof-structure $\mathcal{S}$.
The address of $v$ will be a term $\ulocus{v'}{t}$ where $t$ is a path encoded as a sequence of $\lc$ (left) and $\rc$ (right) symbols representing the direction to follow in $\mathcal{S}$ to get from the conclusion $v' \in \conclu{\mathcal{S}}$ to the atom $v$.

For convenience, we suggest an inductive definition of proof-structures based on their underlying hypergraph.

\begin{rem}[Inductive definition of proof-structures]
\label{rem:inductiveps}
A proof-structure with only one hyperedge is necessarily an axiom with two conclusions, written $\indax{u,v}$. Then a proof-structure $\mathcal{S}$ with $n$ hyperedges is either built from the union of two proof-structures, with respectively $k$ and $n-k$ hyperedges (written $\indunion{\mathcal{S}_1}{\mathcal{S}_2}$), or from a proof-structure with $n-1$ hyperedges extended by either a $\otimes$, $\parr$, or cut hyperedge on two of its conclusions $u$ (left) and $v$ (right). This is written $\indtens{u,v}{\mathcal{S}'}$, $\indpar{u,v}{\mathcal{S}'}$ and $\indcut{u,v}{\mathcal{S}'}$.
\end{rem}

We use this inductive definition to define the \emph{address} of atoms in a proof-structure.

\begin{defi}[Vertex above another one]
A vertex $v$ is \emph{above} another vertex $u$, written in a proof-structure if there exists a directed path (\cf \Cref{sec:hypergraphs}) from $v$ to $u$ going through only $\otimes$ and $\parr$ hyperedges.
\end{defi}

\begin{figure}
    \scalebox{0.85}{
    \begin{tikzpicture}
        \node at (-0.75, 0.75) (a) {$\hdots$};
        \node at (0.75, 0.75) (b) {$\hdots$};
        \node at (0, 0) (par) {$\parr$};
        \node at (0, -1) (ab) {0};
        \draw[-stealth] (a) -- (par);
        \draw[-stealth, Emerald] (b) -- (par) node[midway, right, xshift=2pt, yshift=-2pt] {$\rc$};
        \draw[-stealth, Emerald] (par) -- (ab);
        \node at (2.25, 2.5) (dots) {$\hdots$};
        \node[Bittersweet] at (3.25, 2.5) (oa) {3};
        \node at (4.25, 2.5) (dots) {$\hdots$};

        \node[Emerald] at (0, 2.5) (c) {1};
        \node at (1.25, 2.5) (d) {2};
        \node at (0.75, 1.75) (tens) {$\otimes$};
        \draw[-stealth, Emerald] (c) -- (tens) node[midway, left, yshift=-4pt] {$\lc$};
        \draw[-stealth] (d) -- (tens);
        \draw[-stealth, Emerald] (tens) -- (b);
        
        \node[Emerald] at ($(ab)+(-4,2)$) (l1) {$p_0(\rc \cdot \lc \cdot X)$};
        \draw[-stealth, Emerald] (l1) edge[bend left] (c);
        
        \node[Bittersweet] at ($(ab)+(5,1)$) (l2) {$p_3(X)$};
        \draw[-stealth, Bittersweet] (l2) edge[bend left] (oa);
    \end{tikzpicture}}
    \caption{Addressing of the atoms $1$ and $3$ in a proof-structure relatively to the conclusion they come from.}
    \label{fig:addresses}
\end{figure}
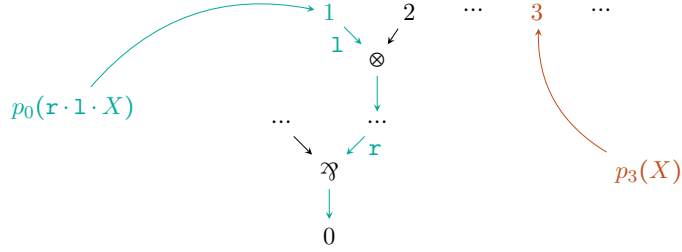

\begin{defi}[Address of an atom]
\label{def:addr}
We define the \emph{path address} $\paddr_\mathcal{S}(v,X)$ to an atom $v$ in a proof-structure $\mathcal{S}$ \wrt the variable $X$ inductively (\cf Remark~\ref{rem:inductiveps}):
\begin{itemize}[label=$\triangleright$]
    \item $\paddr_\mathcal{S}(v,X)=X$ when $\mathcal{S} = \indax{v,\ast}$ or $\mathcal{S} = \indax{\ast,v}$;
    \item $\paddr_\mathcal{S}(v,X)=\paddr_{\mathcal{S}_i}(v,X)$ if $\mathcal{S} = \mathcal{S}_1 \uplus \mathcal{S}_2$ and $v \in V^{\mathcal{S}_i}$;
    \item $\paddr_\mathcal{S}(v,X)= \lc \cdot \paddr_{\mathcal{S}'}(v,X)$ if $\mathcal{S} = \indpar{v,\ast}{\mathcal{S}'}$ or $\mathcal{S} = \indtens{v,\ast}{\mathcal{S}'}$;
    \item $\paddr_\mathcal{S}(v,X)= \rc \cdot \paddr_{\mathcal{S}'}(v,X)$ if $\mathcal{S} = \indpar{\ast,v}{\mathcal{S}'}$ or $\mathcal{S} = \indtens{\ast,v}{\mathcal{S}'}$ and $\paddr_\mathcal{S}(v,X)=\paddr_{\mathcal{S}'}(v,X)$ otherwise.
\end{itemize}
The path address to $v$ is uniquely defined \wrt to a conclusion $c \in \conclu{\mathcal{S'}}$ where $\mathcal{S}'$ is $\mathcal{S}$ without cuts, \ie $E^{\mathcal{S}'} = E^\mathcal{S}\backslash\cuts{\mathcal{S}}$ and the rest of $\mathcal{S}$ is defined as in $\mathcal{S}'$.

The \emph{address} of $v$ is then defined as the term $\addr_\mathcal{S}(v,X) := \ulocus{c}{\paddr_\mathcal{S}(v,X)}$.
\end{defi}

\begin{exa}
\Cref{fig:addresses} illustrates the idea of addressing of atoms. The address of the atom $1$ in \Cref{fig:proofstructure} is $\ulocus{7}{\lc \cdot X}$ because it is reachable from the conclusion $7$ by going to the left premise and the address of the atom $3$ is $\ulocus{3}{X}$ because it is directly reachable.
\end{exa}

\begin{defi}[Set of addresses]
\label{def:setofaddresses}
We define $\setofaddresses{S}$ as the \emph{set of addresses} of the shape $\addr_\mathcal{S}(\_,X)$, \ie the countable set of all terms of the form $\ulocus{c}{f_1 \cdot ... \cdot f_n \cdot X}$ where $c$ ranges over conclusions of $\mathcal{S}$ and $f_i \in \{\lc, \rc\}$.
\end{defi}

\begin{nota}[Unary colour]
We will often write $\ccol{+c.t}$ instead of $\ccol{+c(t)}$ for rays having a unary colour as prefix.
\end{nota}

\begin{defi}[Colour change]
Let $\mathbb{B} = (V, C, F, \op, \floor{\cdot})$ and $\mathbb{B}' = (V', C', F', \op', \floor{\cdot}')$ be two signatures. A \emph{colour change} from $\mathbf{B}$ to $\mathbf{B}'$ is a total injective function $\mu : C \rightarrow C'$.

A colour change of a constellation $\Phi$ is a constellation $\mu(\Phi)$ over $\mathbb{B}'$ where the function $\mu : \colours{\Phi} \rightarrow C'$ is a colour change such that $\colours{\Phi}$ is the set of colours in $\Phi$ and $\mu(\Phi)$ is defined by replacing the colours $c$ by $\mu(c)$ in $\Phi$. 
\end{defi}

\begin{defi}[Colour shift]
\label{def:colourshift}
A \emph{colour shift} from a signature $\mathbb{B} = (V, C, F, \op, \floor{\cdot})$ to $\mathbb{B}' = (V', C', F', \op', \floor{\cdot}')$ is a colour change $\mu$ from $\mathbb{B}$ to $\mathbb{B}'$ such that $\floor{\mu(c)} = \op(\floor{\mu(\op(c)))})$.
\end{defi}

\begin{prop}
\label{prop:colourshift}
Let $\Phi$ be a constellation, $\mu$ a colour shift and $A \subseteq C$ a set of colours. We have $\dgraph[A]{\Phi} \simeq \dgraph[A]{\mu(\Phi)}$.
\end{prop}
\begin{proof}
We show that the matchability is preserved. Let $r$ and $r'$ two dual rays. By induction on $r$. If $r$ is a variable, then even if we change the symbols of $r'$, $x$ is still $\alpha$-unifiable with $r'$. If $r$ is $f(r_1, ..., r_n)$, then we must have $r' := f(r_1', ..., r_n')$. If $f$ is a colour for which $\mu$ is defined the two $f$ of $r$ and $r'$ become $\mu(f)$ which preserves the $\alpha$-unifiability. We conclude with the induction hypothesis for the rays $r_i$ and $r_i'$. Conversely, if $r$ and $r'$ are not matchable, they must be two terms with a mismatch of function symbol between $f$ and $g$ such that $f \neq g$. In this case, since $\mu$ is injective and total, we cannot have $\mu(f) = \mu(g)$. Hence, the non $\alpha$-unifiability is preserved as well.
\end{proof}

\begin{defi}[Full colouration of constellation]
\label{def:colouration}
The \emph{full colouration} $\ccol{c}.\Phi$ of a constellation $\Phi$ is defined by a constellation $\mu(\Phi)$ where $\mu$ is a colour shift such that $\mu(x) = c$.
\end{defi}

\begin{defi}[Translation of the computational content of proof]
\label{def:vehicle}
The \emph{vehicle} and the \emph{cuts} of a proof-structure $\mathcal{S}$ are respectively defined by the following constellations: \[\axof{\mathcal{S}} := \sum_{e \in \axioms{\mathcal{S}}} [\addr_\mathcal{S}(\concl{e},X)), \addr_\mathcal{S}(\concr{e},X)],
\qquad
\cutof{\mathcal{S}} := \sum_{e \in \mathrm{Cut}(\mathcal{S})}[\ulocus{\concl{e}}{X}, \ulocus{\concr{e}}{X}].\]
We define the \emph{computational content} of $\mathcal{S}$ as the constellation $\compof{\mathcal{S}} := \ccol{+c.\axof{\mathcal{S}}} \uplus \ccol{-c.\cutof{\mathcal{S}}}$.
\end{defi}

We now show that the execution $\exec(\compof{\mathcal{S}})$, which can be understood as an interaction between $\ccol{+c.\axof{\mathcal{S}}}$ and $\ccol{-c.\cutof{\mathcal{S}}}$, has the same behaviour as the cut-elimination on $\mathcal{S}$. This shows that cut-elimination for proof-structures can be simulated in the stellar resolution.

\begin{lem}[Simulation of cut-elimination]
\label{lem:cutelim}
Let $\mathcal{R}$ be a proof-structure such that $\mathcal{R} \leadsto \mathcal{S}$. We have $\exec(\compof{\mathcal{R}}) = \exec(\compof{\mathcal{S}})$.
\end{lem}
\begin{proof}
For convenience, we will sometimes use the notations of Definition~\ref{rem:inductiveps}. By case analysis on a reducible cut selected in $\mathcal{R}$.
\begin{itemize}[label=$\blacktriangleright$]
    \item Assume we have an $\axl/\cutl$ cut on two vertices $v_1$ and $v_2$ with $v_1$ conclusions of an axiom $\indax{v_0,v_1}$ and $v_0 \neq v_2$ such that $v_2$ is conclusion of some $e \in E^{\mathcal{R}_\Gamma}$ (\ie $v_2 \in \eout{e}$) where $\mathcal{R}_\Gamma$ is the rest of the proof-structure.
    Then we have an $\axl/\cutl$ cut-elimination which removes $v_1$ and $v_2$ then updates the targets $\eout{e}$ of $e$ with $\eout{e} := \eout{e}\subst{v_2}{v_0}$ ($v_2$ and $v_0$ now refer to the same location by replacing $v_2$ by $v_0$ in the targets of $e$). This produces a new proof-structure $\mathcal{S}$ where the only remaining cuts are the ones of $\mathcal{R}_\Gamma$.
    We have $\compof{\mathcal{R}} = \compof{\mathcal{R}_\Gamma} + [\vcut{v_1}{X}, \vcut{v_2}{X}]$.
    By the definition of vehicle (\cf Definition~\ref{def:vehicle}) and addresses (\cf Definition~\ref{def:addr}), we necessarily have a star $[\varphi(\clocus{v_0}{t}), \clocus{v_1}{X}] \in \axof{\mathcal{R}_\Gamma}$ for some $t$ (with a colouring $\varphi$ depending on whether the ray is related to a cut or not and $t=X$ if $v_0$ is not source of another hyperedge) and $[\clocus{v_2}{X}, r] \in \axof{\mathcal{R}_\Gamma}$ for some ray $r$.
    By fusion, the cut star will merge these two stars and form $\phi := [\varphi(\clocus{v_0}{t}), r]$. We finally obtain the constellation $\compof{\Gamma} + \phi$.
    The translation of $\mathcal{S}$ coincides with the previous constellation obtained by fusion because what what connected to $v_2$ ($r$) by an axiom is now connected $v_0$ in $\phi$, hence $\compof{\mathcal{S}} = \compof{\Gamma} + \phi$. It is indeed a relocation of atom. By partial pre-execution (\cf Lemma~\ref{lem:ppe}), we can focus on this step of fusion and preserve the correct diagrams without adding more diagrams because there is only one choice of connexion (this subgraph of $\dgraph{\compof{\mathcal{R}}}$ corresponds to a deterministic constellation).
    Therefore, $\exec(\compof{\mathcal{R}}) = \exec(\compof{\mathcal{R}_\Gamma} + [\vcut{v_1}{X}, \vcut{v_2}{X}]) = \exec(\compof{\Gamma} + \phi) = \exec(\compof{\mathcal{S}})$.

    \item Assume we have a $\otimes/\parr$ cut between two vertices: $v_1$ conclusion of a $\parr$ hyperedge of inputs $\concl{v_1},\concr{v_1}$ and $v_2$ conclusion of a $\otimes$ hyperedge of inputs $\concl{v_2},\concr{v_2}$. We call $\mathcal{R}_\Gamma$ the rest of the hypergraph.
    We have $\compof{\mathcal{R}} = \compof{\mathcal{R}_\Gamma} + [\vcut{v_1}{X}, \vcut{v_2}{X}]$.
    By the definition of vehicle (\cf Definition~\ref{def:vehicle}) and addresses (\cf Definition~\ref{def:addr}), we necessarily have rays of the shape $\clocus{v_1}{\lc \cdot t_1}, \clocus{v_1}{\rc \cdot t_2}, \clocus{v_2}{\lc \cdot t_3}, \clocus{v_2}{\rc \cdot t_4}$ for some $t_1, t_2, t_3, t_4$.
    Since the lower part of $\mathcal{R}$ is organised as a tree, the path address $t_1$ (\resp $t_2$) is designed to be equal to $t_3$ (\resp $t_4$) when they have the same position relatively to $v_1$ and $v_2$ or only $\alpha$-unifiable when one is a subpath of the other.
    A cut star between the locations $v_1$ and $v_2$ uses the same variable for its two rays, hence it forces a connexion between two rays $\clocus{v_1}{t}$ and $\clocus{v_2}{u}$ where $t$ and $u$ are $\alpha$-unifiable. There is only one possible such connexion by duplicating the cut and forming one diagram. Hence, the corresponding subgraph of $\dgraph{\compof{\mathcal{R}}}$ corresponds to a deterministic constellation.
    By partial pre-execution (\cf Lemma~\ref{lem:ppe}), we can focus on some steps of fusion while preserving correct diagrams.
    By fusion and by duplicating the cut star for each pair of rays $\clocus{v_1}{t}$ and $\clocus{v_2}{u}$ where $t$ and $u$ are $\alpha$-unifiable, we can merge stars of the shape $[r, \clocus{v_1}{\lc \cdot t_i}]$ and $[r', \clocus{v_2}{\lc \cdot u_i}]$ and produce $[r, r']$ (same idea for $\rc$ instead of $\lc$).
    This has exactly the same effect as relocating the atoms $\concl{v_1},\concr{v_1}, \concl{v_2}$ and $\concr{v_2}$ in order to get rays of the shape $\clocus{\concl{v_1}}{t_1}, \clocus{\concr{v_1}}{t_2}, \clocus{\concl{v_2}}{t_3}, \clocus{\concr{v_2}}{t_4}$ for the previous path addresses $t_1, t_2, t_3$ and $t_4$ with the cuts $[\vcut{\concl{v_1}}{X}, \vcut{\concl{v_2}}{X}]$ and $[\vcut{\concr{v_1}}{X}, \vcut{\concr{v_2}}{X}]$.
    This preserves the $\alpha$-unifiability between rays and exactly coincides with $\compof{\mathcal{S}}$ which removes the conclusions $v_1$ and $v_2$, then relocates their sources which become conclusions. Therefore, the translation of $\mathcal{R}$ and $\mathcal{S}$ both have the same execution because of their structural equivalence which has no impact on the normal form.
\end{itemize}
\end{proof}

\begin{thm}[Simulation of reduction for proof-nets]
\label{thm:dynamics}
For an MLL+MIX proof-net $\mathcal{R}$ of normal form $\mathcal{S}$, we have $\exec(\compof{\mathcal{R}}) = \axof{\mathcal{S}}$.
\end{thm}
\begin{proof}
This result is a consequence of Lemma~\ref{lem:cutelim} by induction of the number of cut-elimination steps from $\mathcal{R}$ to $\mathcal{S}$, as well as the fact that $\exec(\compof{\mathcal{S}}) = \exec(\axof{\mathcal{R}}) = \axof{\mathcal{R}}$ since $\mathcal{S}$ does not contain cuts.
\end{proof}

\begin{exa}[Correct cut-elimination]
\label{ex:mllcorrect}
We have the following reduction $\mathcal{S} \leadsto^* \mathcal{S}'$ of proof-structure:
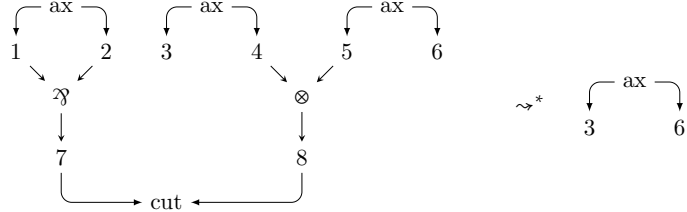
\begin{center}
\scalebox{0.8}{
\begin{minipage}{0.55\textwidth}
\begin{tikzpicture}
    \node at (-0.75, 0.75) (a) {1};
    \node at (0.75, 0.75) (b) {2};
    \node at (0, 0) (par) {$\parr$};
    \node at (0, -1) (ab) {7};
    \draw[-stealth] (a) -- (par);
    \draw[-stealth] (b) -- (par);
    \draw[-stealth] (par) -- (ab);

    \node at (1.75, 0.75) (f1) {3};
    \node at (6.25, 0.75) (f2) {6};

    \node at (3.25, 0.75) (c) {4};
    \node at (4.75, 0.75) (d) {5};
    \node at (4, 0) (tens) {$\otimes$};
    \node at (4, -1) (cd) {8};
    \draw[-stealth] (c) -- (tens);
    \draw[-stealth] (d) -- (tens);
    \draw[-stealth] (tens) -- (cd);

    \node at (1.75, -1.75) (cut) {cut};
    \draw[-latex, rounded corners=5pt] (ab) |- (cut);
    \draw[-latex, rounded corners=5pt] (cd) |- (cut);
    \node at (0, 1.5) (ax) {ax};
    \draw[-latex, rounded corners=5pt] (ax) -| (a);
    \draw[-latex, rounded corners=5pt] (ax) -| (b);
    \node at (2.5, 1.5) (ax) {ax};
    \draw[-latex, rounded corners=5pt] (ax) -| (f1);
    \draw[-latex, rounded corners=5pt] (ax) -| (c);
    \node at (5.5, 1.5) (ax) {ax};
    \draw[-latex, rounded corners=5pt] (ax) -| (f2);
    \draw[-latex, rounded corners=5pt] (ax) -| (d);
\end{tikzpicture}
\end{minipage}
$\leadsto^*\quad$
\begin{minipage}{0.15\textwidth}
\begin{tikzpicture}
    \node at (-0.75, 0.75) (a) {3};
    \node at (0.75, 0.75) (b) {6};
    \node at (0, 1.5) (ax) {ax};
    \draw[-latex, rounded corners=5pt] (ax) -| (a);
    \draw[-latex, rounded corners=5pt] (ax) -| (b);
\end{tikzpicture}
\end{minipage}
} \end{center}

The proof-structure $\mathcal{S}$ is translated into $\compof{\mathcal{S}} =$
\[[\clocus{7}{\lc \cdot X}, \clocus{7}{\rc \cdot X}] + [\clocus{3}{X}, \clocus{8}{\lc \cdot X}] + [\clocus{8}{\rc \cdot X}, \clocus{6}{X}] + \]
\[[\vcut{7}{X}, \vcut{8}{X}].\]
The ray $\vcut{7}{X}$ can match either $\clocus{7}{\lc \cdot X}$ or $\clocus{7}{\rc \cdot X}$ and the same occurs for $\vcut{8}{X}$. In order to satisfy these $\alpha$-unifications, the cut star must be duplicated and each occurrence of cut must connect rays with the same address, \ie the path addresses $\lc \cdot X$ together and not $\lc \cdot X$ with $\rc \cdot X$. We obtain the following diagram:

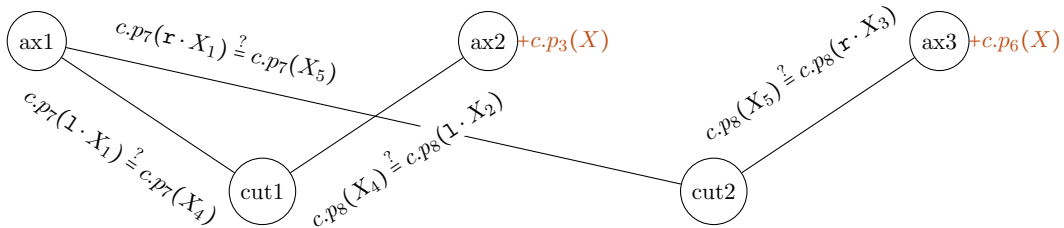
\begin{center}
\begin{tikzpicture}[every node/.style={scale=0.8}]
    \node[circle,draw] at (0, 0) (a1) {ax1};
    \node[circle,draw] at (6, 0) (a2) {ax2};
    \node[circle,draw] at (12, 0) (a3) {ax3};
    \node[circle,draw] at (3, -2) (c1) {cut1};
    \node[circle,draw] at (9, -2) (c2) {cut2};
    \node at (7, 0) (f1) {$\clocus{3}{X}$};
    \node at (13, 0) (f2) {$\clocus{6}{X}$};

    \draw (c1) edge node[midway,sloped,below=3.5mm] {$c.p_7(\lc \cdot X_1) \eqq c.p_7(X_4)$} (a1);
    \draw (c2) edge node[near end,sloped,above=2mm] {$c.p_7(\rc \cdot X_1) \eqq c.p_7(X_5)$} (a1);
    \draw (c2) edge node[midway,sloped,above=5mm] {$c.p_8(X_5) \eqq c.p_8(\rc \cdot X_3)$} (a3);
    \draw (c1) edge node[midway,sloped,below=5.5mm,fill=white,inner sep=0] {$c.p_8(X_4) \eqq c.p_8(\lc \cdot X_2)$} (a2);
\end{tikzpicture}
\end{center}

By case analysis, it is easy to check that it is the only possible diagram. Since the $\alpha$-unification is exact (\cf \Cref{sec:unification}), it is simply a graph contraction doing no more than renamings and we get $\exec(\compof{\mathcal{S}}) = [\clocus{3}{X}, \clocus{6}{X}] = \exec(\compof{\mathcal{S'}})$.
\end{exa}

\begin{exa}[Incorrect cut-elimination]
\label{ex:mllwrong}
We have the following reduction $\mathcal{S} \leadsto^* \mathcal{S}'$ of proof-structure:

\begin{center}
\scalebox{0.8}{
\begin{minipage}{0.35\textwidth}
\begin{tikzpicture}
    \node at (-0.75, 0.75) (a) {1};
    \node at (0.75, 0.75) (b) {2};
    \node at (0, 0) (par) {$\parr$};
    \node at (0, -1) (ab) {5};
    \draw[-stealth] (a) -- (par);
    \draw[-stealth] (b) -- (par);
    \draw[-stealth] (par) -- (ab);

    \node at (1.75, 0.75) (c) {3};
    \node at (3.25, 0.75) (d) {4};
    \node at (2.5, 0) (tens) {$\otimes$};
    \node at (2.5, -1) (cd) {6};
    \draw[-stealth] (c) -- (tens);
    \draw[-stealth] (d) -- (tens);
    \draw[-stealth] (tens) -- (cd);

    \node at (1.25, -1.75) (cut) {cut};
    \draw[-latex, rounded corners=5pt] (ab) |- (cut);
    \draw[-latex, rounded corners=5pt] (cd) |- (cut);
    \node at (0, 1.5) (ax) {ax};
    \draw[-latex, rounded corners=5pt] (ax) -| (a);
    \draw[-latex, rounded corners=5pt] (ax) -| (c);
    \node at (2, 1.75) (ax) {ax};
    \draw[-latex, rounded corners=5pt] (ax) -| (b);
    \draw[-latex, rounded corners=5pt] (ax) -| (d);
\end{tikzpicture}
\end{minipage}
$\leadsto^*\qquad$
\begin{minipage}{0.3\textwidth}
\begin{tikzpicture}
    \node at (-0.75, 0.75) (a) {1};
    \node at (0.75, 0.75) (b) {2};

    \node at (1.75, 0.75) (c) {3};
    \node at (3.25, 0.75) (d) {4};

    \node at (0, 0) (cut) {cut};
    \draw[-latex, rounded corners=5pt] (a) |- (cut);
    \draw[-latex, rounded corners=5pt] (c) |- (cut);
    \node at (2, -0.5) (cut) {cut};
    \draw[-latex, rounded corners=5pt] (b) |- (cut);
    \draw[-latex, rounded corners=5pt] (d) |- (cut);
    \node at (0, 1.5) (ax) {ax};
    \draw[-latex, rounded corners=5pt] (ax) -| (a);
    \draw[-latex, rounded corners=5pt] (ax) -| (c);
    \node at (2, 1.75) (ax) {ax};
    \draw[-latex, rounded corners=5pt] (ax) -| (b);
    \draw[-latex, rounded corners=5pt] (ax) -| (d);
\end{tikzpicture}
\end{minipage}}
\end{center}

The proof-structure $\mathcal{S}$ is translated into $\compof{\mathcal{S}} =$
\[[\clocus{5}{\lc \cdot X}, \clocus{6}{\lc \cdot X}] + [\clocus{5}{\rc \cdot X}, \clocus{6}{\rc \cdot X}] + [\vcut{5}{X}, \vcut{6}{X}]\]
with the following dependency graph $\dgraph{\compof{S}}$:
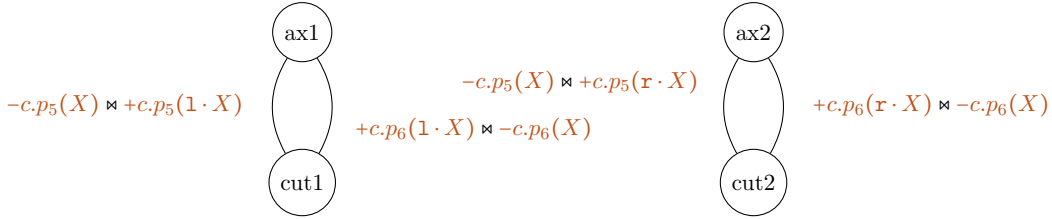
\begin{center}
\begin{tikzpicture}[every node/.style={scale=0.8}]
    \node[circle,draw] at (0, 0) (a1) {ax1};
    \node[circle,draw] at (6, 0) (a2) {ax2};
    \node[circle,draw] at (0, -2) (c1) {cut1};
    \node[circle,draw] at (6, -2) (c2) {cut2};

    \draw (c1) edge[bend left] node[midway,left=3.5mm] {$\vcut{5}{X} \match \clocus{5}{\lc \cdot X}$} (a1);
    \draw (c1) edge[bend right] node[near start,right=3mm] {$\clocus{6}{\lc \cdot X} \match \vcut{6}{X}$} (a1);

    \draw (c2) edge[bend left] node[near end,left=3.5mm] {$\vcut{5}{X} \match \clocus{5}{\rc \cdot X}$} (a2);
    \draw (c2) edge[bend right] node[midway,right=3.5mm] {$\clocus{6}{\rc \cdot X} \match \vcut{6}{X}$} (a2);
\end{tikzpicture}
\end{center}
The cycles in $\dgraph{\compof{S}}$ can be unfolded and yield infinitely many saturated correct diagrams, all actualising into $\emptystar$. We have $\exec(\compof{\mathcal{S}}) = \sum_{i=1}^\infty \emptystar = \exec(\compof{\mathcal{S'}})$.
\end{exa}

\subsection{Simulation of logical correctness}\label{subsec:correctness}

Since proof-structures are more general than proof-nets (\cf Definition~\ref{def:proofnets}), we have to check which proof-structures are ``logically correct". The idea of logical correctness traditionally corresponds to the fact of coming from a sequent calculus proof (\cf \Cref{subfig:mllsequent}), which is taken as the natural understanding of what ``being an actual proof" means.

A beautiful result of Girard, analysed by many subsequent works \cite{danos1989structure,danosPHD,lafont1995proof,murawski2000dominator,de2011correctness,retore2003handsome,bagnol2015dependencies}, is that the proof-structures that are proof-nets can be characterised by a topological/combinatorial property called a \emph{correctness criterion}. While Girard's original criterion, called the long-trip criterion \cite[Section III.2]{linearlogic}, is about the set of walks in a proof-structure, we will here work with Danos and Regnier's simplified criterion  \cite[Section 3.2]{danos1989structure} which is the most standard and which could not be treated by previous GoI models. Similarly to how a product has to pass several tests in order to be certified, this criterion defines tests to pass in order to be logically correct.

\begin{figure}
    \begin{tabular}{|c|c|c|c|}
        \hline
        Structure & Axioms & Test 1 & Test 2 \\
        \hline
        \scalebox{0.75}{\begin{tikzpicture}
          \node at (-0.75, 0.75) (a) {1};
          \node at (0.75, 0.75) (b) {2};
          \node at (0, 0) (tens) {$\otimes$};
          \node at (0, -0.75) (ab) {5};
          \draw[-] (a) -- (tens);
          \draw[-] (b) -- (tens);
          \draw[-] (tens) -- (ab);

          \node at (1.25, 0.75) (ad) {3};
          \node at (2.75, 0.75) (bd) {4};
          \node at (2, 0) (par) {$\parr$};
          \node at (2, -0.75) (adbd) {6};
          \draw[-] (ad) -- (par);
          \draw[-] (bd) -- (par);
          \draw[-] (par) -- (adbd);

          \node at (0, 1.25) (ax1) {ax};
          \node at (2, 1.5) (ax2) {ax};
          \draw[-, rounded corners=5pt] (ax1) -| (a);
          \draw[-, rounded corners=5pt] (ax1) -| (ad);
          \draw[-, rounded corners=5pt] (ax2) -| (b);
          \draw[-, rounded corners=5pt] (ax2) -| (bd);
        \end{tikzpicture}}
        &
        \scalebox{0.75}{\begin{tikzpicture}
          \node at (-0.75, 0.75) (a) {1};
          \node at (0.75, 0.75) (b) {2};

          \node at (1.25, 0.75) (ad) {3};
          \node at (2.75, 0.75) (bd) {4};

          \node at (0, 1.25) (ax1) {ax};
          \node at (2, 1.5) (ax2) {ax};
          \draw[-, rounded corners=5pt] (ax1) -| (a);
          \draw[-, rounded corners=5pt] (ax1) -| (ad);
          \draw[-, rounded corners=5pt] (ax2) -| (b);
          \draw[-, rounded corners=5pt] (ax2) -| (bd);
        \end{tikzpicture}}
        &
        \scalebox{0.75}{\begin{tikzpicture}
          \node at (-0.75, 0.75) (a) {1};
          \node at (0.75, 0.75) (b) {2};
          \node at (0, 0) (tens) {$\otimes$};
          \node at (0, -0.75) (ab) {5};
          \draw[-] (a) -- (tens);
          \draw[-] (b) -- (tens);
          \draw[-] (tens) -- (ab);

          \node at (1.25, 0.75) (ad) {3};
          \node at (2.75, 0.75) (bd) {4};
          \node at (2, 0) (par) {$\parr_L$};
          \node at (2, -0.75) (adbd) {6};
          \draw[-] (ad) -- (par);
          \draw[-] (par) -- (adbd);
        \end{tikzpicture}}
        &
        \scalebox{0.75}{\begin{tikzpicture}
          \node at (-0.75, 0.75) (a) {1};
          \node at (0.75, 0.75) (b) {2};
          \node at (0, 0) (tens) {$\otimes$};
          \node at (0, -0.75) (ab) {5};
          \draw[-] (a) -- (tens);
          \draw[-] (b) -- (tens);
          \draw[-] (tens) -- (ab);

          \node at (1.25, 0.75) (ad) {3};
          \node at (2.75, 0.75) (bd) {4};
          \node at (2, 0) (par) {$\parr_R$};
          \node at (2, -0.75) (adbd) {6};
          \draw[-] (bd) -- (par);
          \draw[-] (par) -- (adbd);
        \end{tikzpicture}}
        \\
        \hline
    \end{tabular}
    \caption{The axioms and tests of a proof-structure. The combination of axioms and a test corresponds to a correctness hypergraph representing a testing of the proof-structure.}
    \label{fig:tests}
\end{figure}
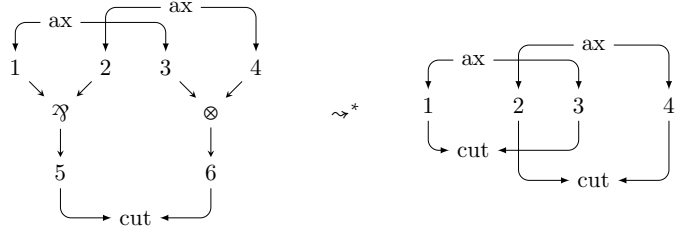

We define the \emph{correctness hypergraphs} associated to a proof-structure $\mathcal{S}$ as undirected copies of $\mathcal{S}$ with one source of each $\parr$-labelled hyperedge removed. They correspond to a testing between the upper part made of axioms, the \emph{tested} (corresponding to the vehicle in the stellar resolution\footnote{It also gives meaning to Girard's terminology of vehicle \cite{transyn1} which can be understood from its abstract definition, \eg language as the vehicle of thought, or by its concrete definition, \eg a car in an industry which is tested in order to be certified.}), and the lower part which is the \emph{test}. This decomposition of a proof-structure into axioms and tests is illustrated in \Cref{fig:tests}. The Danos-Regnier criterion states that a proof-structure is an MLL proof-net if and only if all its correctness hypergraphs are all connected and acyclic (\cf \Cref{sec:hypergraphs}).

Since the idea of logical correctness is purely structural as well, our definitions still deals with unlabelled proof-structures.

\begin{nota}
Given a proof-structure $\mathcal{S} = (V,E,\ein{},\eout{},\ell_E)$, we write $\parr(\mathcal{S})$ the subset $P\subseteq E$ of $\parr$-labelled edges, \ie $\parr(\mathcal{S}) = \{e\in E\mid \ell_E(e) = \parr \}$.
\end{nota}

\begin{defi}[Correctness hypergraph]
\label{def:correctnessgraph}
Let $\mathcal{S} = (V,E,\ein{},\eout{},\ell_E)$ be a proof-structure. A \emph{switching} is a map $\varphi : \parr(\mathcal{S}) \rightarrow \{\parr_L, \parr_R\}$.
Its associated \emph{correctness hypergraph} is the undirected hypergraph with labelled hyperedges $\mathcal{S}^\varphi = (V,E,\mathtt{in'}, \eout{}, \ell_E')$ induced by the switching $\varphi$ which is defined with
\begin{itemize}
	\item $\mathtt{in'}(e) = \{u\}$ where $u$ is the left premise of $e$ when $e \in \parr(\mathcal{S})$ and $\varphi(e) = \parr_L$;
	\item $\mathtt{in'}(e) = \{u\}$ where $u$ is the right premise of $e$ when $e \in \parr(\mathcal{S})$ and $\varphi(e) = \parr_R$;
	\item $\mathtt{in'}(e) = \ein{e} \cup \eout{e}$ in all other cases.
\end{itemize}
The labelling $\ell_E'$ is defined by $\ell_E'(e) = \varphi(e)$ when $e \in \parr(\mathcal{S})$ and $\ell_E'(e) = \ell_E(e)$ otherwise.
\end{defi}

If testing proof-structures is seen as certifying products in a factory, the correctness criterion also gives them a label/certificate: the sequent they prove. This is represented as the possibility of labelling a proof-structure so that it corresponds to a proof-net. Hence, a proof-structure can actually correspond to several sequent calculus proofs depending on the labelling we choose.

\begin{defi}[MLL-certification and MLL+MIX-certification]
A proof-structure $\mathcal{S} = (V,E,\ein{},\eout{},\ell_E)$ is \emph{MLL-certifiable} (\resp \emph{MLL+MIX-certifiable}) with $\vdash A_1, ..., A_n$ when there exists a vertex-labelling function $\ell_V$ such that $(V,E,\ein{},\eout{},\ell_V,\ell_E)$ is an MLL (\resp MLL+MIX) proof-net.

When there exists $\vdash A_1, ..., A_n$ such that $\mathcal{S}$ is MLL(+MIX)-certifiable with $\vdash A_1, ..., A_n$ then we simply say that $\mathcal{S}$ is MLL(+MIX)-certifiable.
\end{defi}

\begin{thm}[MLL+MIX correctness]
\label{thm:mixcorrectness}
A proof-structure $\mathcal{S}$ is MLL+MIX-certifiable if and only if $\mathcal{S}^\varphi$ is acyclic for all switching $\varphi$.
\end{thm}
\begin{proof}
Proven in \cite[Theorem 4.7 and 4.8]{fleury1994mix}.
\end{proof}

\begin{thm}[Danos-Regnier correctness]
\label{thm:correctness}
A proof-structure $\mathcal{S}$ is MLL-certifiable if and only if it is MLL+MIX-certifiable and $\mathcal{S}^\varphi$ is connected for all switching $\varphi$.
\end{thm}
\begin{proof}
Proven in \cite[Theorem 4]{danos1989structure}.
\end{proof}

The previous section handled cut-elimination by means of binary stars translating axiom and cut hyperedges. We now need to translate tests which contain $\parr$ and $\otimes$ hyperedges. Binary stars are no longer sufficient for a natural and satisfactory treatment of logical correctness since we now have to deal with a ternary hyperedge\footnote{As previous GoI models used binary objects (\eg edge of a graph), ternary hyperedges could not be encoded.}: the tensor link.

For minor technical reasons, instead of directly translating hyperedges (which would be more natural), the vertices are translated\footnote{Although they could have been added without problem, conclusion hyperedge were omitted for convenience in Definition~\ref{subfig:proofstructures}. However, since we will need to translate conclusions as rays, a translation of hyperedge would lack this information.}.

For our encoding, we use two colours: $\ccol{c}$ (for computation) and $\tcol{t}$ (for testing). A vehicle will be coloured with the colour $\ccol{c}$ when we want its execution by connecting it with cuts and and with $\tcol{t}$ when being subject to an interaction against tests.

\begin{defi}[MLL test]
\label{def:test}
Let $\mathcal{S}$ be a proof-structure and $\varphi$ one of its switchings. The \emph{test} associated to $\mathcal{S}^\varphi$ is the constellation defined by $\testof{\mathcal{S}}{\varphi} := \sum_{v \in V^{\mathcal{S}^\varphi}} v^\bigstar$.
We define the translation $v^\bigstar$ of a vertex $v$ conclusion of an hyperedge $e$ as follows:
\begin{itemize}
    \item if $\ell_E(e) = \axl$ then $v^\bigstar = [\tcol{-\addr_\mathcal{S}(v, X)}, \qray[+]{v}]$;
    \item if $\ell_E(e) = \parr_L$ and $\eout{e} = \{u, w\}$ then $v^\bigstar = \gparlc{u}{w}{v}$;
    \item if $\ell_E(e) = \parr_R$ and $\eout{e} = \{u, w\}$ then $v^\bigstar = \gparrc{u}{w}{v}$;
    \item if $\ell_E(e) = \otimes$ and $\eout{e} = \{u, w\}$ then $v^\bigstar = \gtensc{u}{w}{v}$;
    \item if $v \in \conclu{\mathcal{S}}$ then $v^\bigstar = \gconc{v}$;
    \item for each cut (hyperedge $e$ such that $\ell_E(e) = \cutl$), we add a star $[\vcut{\concl{e}}{X}, \vcut{\concr{e}}{X}]$.
\end{itemize}
\end{defi}

The technical purpose of the constant $\gc$ is to make the terms of the tests distinct from the ones of the vehicle so that the cuts can be applied to both the vehicle and the tests by simply changing colours ($\vcut{v}{X}$ matches both $\clocus{v}{X}$ and $\qray[+]{v}$). This allows more flexibility in the definitions.

\begin{figure}
\begin{tikzpicture}[every node/.style={scale=0.75}]
    \node[circle,draw] at (0, 0) (a1) {$\phi_1$};
    \node[circle,draw] at (3, 0) (a2) {$\phi_2$};
    \node[circle,draw] at (6, 0) (a3) {$\phi_3$};
    \node[circle,draw] at (9, 0) (a4) {$\phi_4$};

    \node at (0, .6) (b1) {$\tlocus[-]{5}{\lc\cdot X}$};
    \node at (3, .6) (b2) {$\tlocus[-]{5}{\rc\cdot X}$};
    \node at (6, .6) (b3) {$\tlocus[-]{6}{\lc\cdot X}$};
    \node at (9, .6) (b4) {$\tlocus[-]{6}{\rc\cdot X}$};

    \node[circle,draw] at (1.5, -2) (t) {$\phi_\otimes$};
    \node[circle,draw] at (7, -2) (p1) {$\phi_\emptyset$};
    \node[circle,draw] at (8, -2) (p2) {$\phi_{\parr_R}$};

    \node[circle,draw] at (1.5, -4) (c1) {$\phi_{c_1}$};
    \node[circle,draw] at (8, -4) (c2) {$\phi_{c_2}$};
    \node at (1.5, -4.6) (r1) {$\ulocus{5}{X}$};
    \node at (8, -4.6) (r2) {$\ulocus{6}{X}$};

    \draw (a1) edge node[midway,sloped,below=6mm] {$\qray[+]{1} \match \qray[-]{1}$} (t);
    \draw (a2) edge node[midway,sloped,below=6mm] {$\qray[-]{2} \match \qray[+]{2}$} (t);
    \draw (a3) edge node[midway,sloped,below=6mm] {$\qray[+]{3} \match \qray[-]{3}$} (p1);
    \draw (a4) edge node[midway,sloped,below=6mm] {$\qray[-]{4} \match \qray[+]{4}$} (p2);
    \draw (t) edge node[midway,left=6mm] {$\qray[+]{5} \match \qray[-]{5}$} (c1);
    \draw (p2) edge node[midway,right=6mm] {$\qray[+]{6} \match \qray[-]{6}$} (c2);
\end{tikzpicture}
    \caption{Dependency graph of the constellation corresponding to Test 2 in \Cref{fig:tests}.}
    \label{fig:structural}
\end{figure}
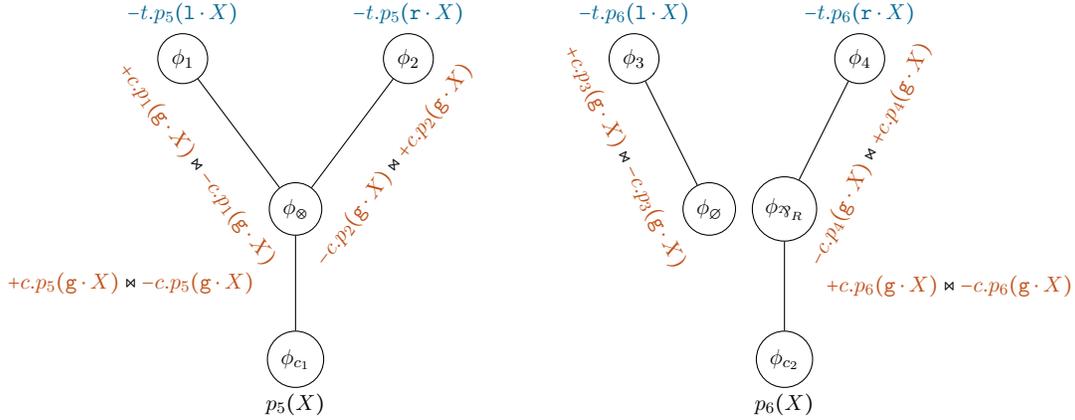

Tests for a proof-structure $\mathcal{S}$ are actually designed so that $\dgraph{\tcol{+t.\axof{\mathcal{S}}} \uplus \testof{\mathcal{S}}{\varphi}}$ is structurally equivalent to $\mathcal{S}^\varphi$, as illustrated in \Cref{fig:structural}. We make this idea explicit by the following lemma and show that it leads to a simulation of logical correctness in the stellar resolution.

\begin{lem}[Structural realisation]
\label{lem:structural}
Let $\mathcal{S}$ be a proof-structure and $\varphi$ one of its switchings. We have $\dgraph{\tcol{+t.\axof{\mathcal{S}}} \uplus \testof{\mathcal{S}}{\varphi}} \simeq \mathcal{S}^\varphi$.
\end{lem}
\begin{proof}
Remark that a test (\cf Definition~\ref{def:test}) could alternatively be defined by a translation of hyperedges instead of vertices. Let $v$ be a vertex of $\mathcal{S}^\varphi$. It is target of an hyperedge $e$ with sources $\ein{e} = \{u_1, ..., u_k\}$ with $k \leq 2$. We define $e^\bigstar := v^\bigstar$.
By case analysis on $e$ and by Definition~\ref{def:test}, there is a correspondence between $v, u_1, ..., u_k$ and rays of $e^\bigstar$. We write $v^\ast_e$ for the ray corresponding to $v$ in $e$.
The tests are designed so that two hyperedges $e_1, e_2 \in E^{\mathcal{S}^\varphi}$ share a vertex $v$ if and only if $v_{e_1}^\ast$ is $\alpha$-unifiable with $v_{e_2}^\ast$. Notice that they both correspond to the same formula seen as input or output depending on the $e$ in $v_e^\ast$. These two points of view correspond to matchable rays of opposite colours.
In conclusion, $\dgraph{\tcol{+t.\axof{\mathcal{S}}} \uplus \testof{\mathcal{S}}{\varphi}}$ preserves the structure of $\mathcal{S}^\varphi$ as hypergraph and links vertices in the same way.
\end{proof}

A technical corollary is that the translation of correctness hypergraph is deterministic and exact (\cf Definition~\ref{def:propconst}). It ensures that all diagrams are correct but also that if its dependency graph is connected and acyclic, there is no branching leading to non-deterministic choices for a ray.

\begin{cor}
\label{cor:deterministicexact}
Let $\mathcal{S}$ be a proof-structure and $\varphi$ one of its switchings. Then $\tcol{+t.\axof{\mathcal{S}}} \uplus \testof{\mathcal{S}}{\varphi}$ is deterministic and exact.
\end{cor}
\begin{proof}
By Lemma~\ref{lem:structural}, $\dgraph{\tcol{+t.\axof{\mathcal{S}}} \uplus \testof{\mathcal{S}}{\varphi}} \simeq \mathcal{S}^\varphi$.
The hypergraph $\mathcal{S}^\varphi$ has vertices which are uniquely connected, \ie for each $v \in V^{\mathcal{S}^\varphi}$, there is only one $e \in E^{\mathcal{S}^\varphi}$ such that $v \in \ein{e}$ or $v \in \eout{e}$.
Therefore, rays are uniquely connected in $\dgraph{\tcol{+t.\axof{\mathcal{S}}} \uplus \testof{\mathcal{S}}{\varphi}}$ and $\tcol{+t.\axof{\mathcal{S}}} \uplus \testof{\mathcal{S}}{\varphi}$ is deterministic.
A simple analysis on Definition~\ref{def:test} shows that it is also exact.
\end{proof}

The idea for MLL correctness is that we would like to have $\exec(\tcol{+t.\axof{\mathcal{S}}} \uplus \testof{\mathcal{S}}{\varphi}) = [\ulocus{v_1}{X}, ..., \ulocus{v_n}{X}]$ with $\conclu{\mathcal{S}} = \{v_1, ..., v_n\}$, which would imply that $\mathcal{S}$ has an arborescent shape. In his original paper \cite[Section 2.3]{transyn1}, Girard only considers the case of cut-free proofs by forbidding cyclic diagrams and the empty star. However, it seems that these definitions lead to a technical mistakes in case of a cyclic correctness hypergraph. We describe the problem and explain why it does not occur in our case.

\begin{figure}
    \begin{minipage}{0.15\textwidth}
        \scalebox{0.75}{\begin{tikzpicture}
        \node at (-0.75, 0.75) (a) {1};
        \node at (0.75, 0.75) (b) {2};
        \node at (0, 0) (tens) {$\otimes$};
        \node at (0, -0.75) (ab) {3};
        \draw[-] (a) -- (tens);
        \draw[-] (b) -- (tens);
        \draw[-] (tens) -- (ab);
        \node at (0, 1.5) (ax) {ax};
        \draw[-, rounded corners=5pt] (ax) -| (a);
        \draw[-, rounded corners=5pt] (ax) -| (b);
        \end{tikzpicture}}
    \end{minipage}
    \begin{minipage}{0.6\textwidth}
        $[\tlocus{3}{\lc\cdot X}, \tlocus{3}{\rc\cdot X}]+$ \\ $[\tlocus[-]{3}{\lc\cdot X}, \qray[+]{1}]+[\tlocus[+]{3}{\rc\cdot X}, \qray[+]{2}]$ \\
        $\gtensc{1}{2}{3}+$ \\
        $\gconc{3}$
    \end{minipage}
    \caption{Incorrect correctness hypergraph for a proof-structure $\mathcal{S}$ and its translation. Notice that the cycle is turned into a computational cycle (a loop in a program).}
    \label{fig:incorrectps}
\end{figure}
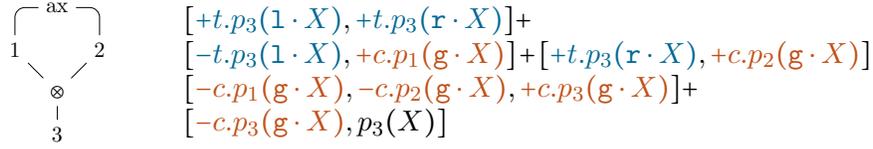

\begin{figure}
\begin{tikzpicture}[every node/.style={scale=0.75}]
    \node[circle,draw] at (0, 0) (a1) {ax};
    \node[circle,draw] at (4, 0) (e2) {$\phi_2$};
    \node[circle,draw] at (8, 0) (t1) {$\phi_\otimes$};
    \node[circle,draw] at (12, 0) (e1) {$\phi_1$};

    \node[circle,draw] at (12, -2) (a2) {ax};
    \node[circle,draw] at (8, -2) (e2p) {$\phi_2$};
    \node[circle,draw] at (4, -2) (t2) {$\phi_\otimes$};
    \node[circle,draw] at (0, -2) (e1p) {$\phi_1$};

    \node[circle,draw] at (8, 2) (c1) {$\phi_{c_1}$};
    \node[circle,draw] at (4, -4) (c2) {$\phi_{c_2}$};
    \node at (9, 2) (r1) {$\ulocus{3}{X_9}$};
    \node at (3, -4) (r2) {$\ulocus{3}{X_{10}}$};

    \draw (a1) edge node[midway,above] {$p_3(\rc\cdot X_1) \eqq p_3(\rc\cdot X_2)$} (e2);
    \draw (e2) edge node[midway,above] {$p_2(\gc\cdot X_2) \eqq p_2(\gc\cdot X_3)$} (t1);
    \draw (t1) edge node[midway,above] {$p_1(\gc\cdot X_3) \eqq p_1(\gc\cdot X_4)$} (e1);
    \draw (e1) edge node[midway,left=3mm] {$p_3(\lc\cdot X_3) \eqq p_3(\lc\cdot X_5)$} (a2);
    \draw (a2) edge node[midway,below] {$p_2(\rc\cdot X_6) \eqq p_2(\rc\cdot X_5)$} (e2p);
    \draw (e2p) edge node[midway,below] {$p_2(\gc\cdot X_7) \eqq p_2(\gc\cdot X_6)$} (t2);
    \draw (t2) edge node[midway,below] {$p_1(\gc\cdot X_8) \eqq p_1(\gc\cdot X_7)$} (e1p);
    \draw (e1p) edge node[midway,right=3mm] {$p_3(\lc\cdot X_9) \eqq p_3(\lc\cdot X_8)$} (a1);
    \draw (t2) edge node[near end,right=3mm] {$p_3(\gc\cdot X_7) \eqq p_3(\gc\cdot X_{10})$} (c2);
    \draw (t1) edge node[near end,left=3mm] {$p_3(\gc\cdot X_9) \eqq p_3(\gc\cdot X_3)$} (c1);
\end{tikzpicture}
    \caption{Example of a correct and saturated cyclic diagram for the constellation from \Cref{fig:incorrectps} actualising into $[\ulocus{3}{X_9}, \ulocus{3}{X_{10}}]$. The cycle can extended infinitely many times by adding copies of three stars of the constellation.}
    \label{fig:cycleunfolding}
\end{figure}
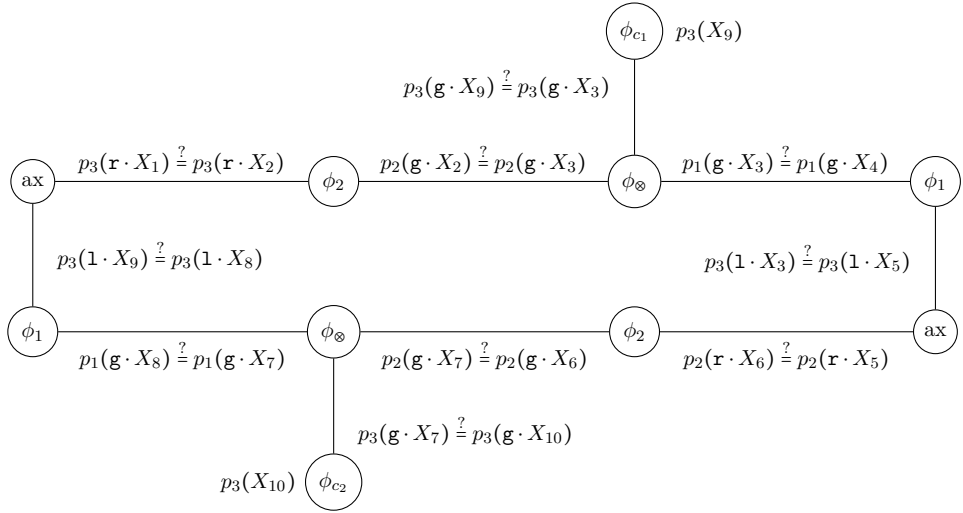

We write $\exec[\restrees]$ for the execution restricted to tree-shaped diagrams (execution used by Girard) which applies the operator $\lightning$ and $\flat$. If we consider the incorrect proof-structure $\mathcal{S}'$ in Example~\ref{ex:mllwrong}, we have $\exec[\restrees](\tcol{+t.\axof{\mathcal{S}'}} \uplus \testof{\mathcal{S}'}{\varphi}) = \emptyset$ making this incorrect proof-structure invisible in the normal form. Assume we have a correct proof-structure $\mathcal{R}$.
By Lemma~\ref{lem:independence}, the union of $\mathcal{R}$ and $\mathcal{S}'$ is translated into a correct constellation although it should not (because $\mathcal{S}'$ is incorrect). This is due to the absence of conclusion caused by cuts.
However, even without cuts, the same phenomena can occur. For instance, if we consider the correctness hypergraph of \Cref{fig:incorrectps}, infinitely many diagrams can be constructed because of the loop of dependencies but all the corresponding diagrams have free polarised rays. Such diagrams are erased by the operator $\lightning$. Therefore, $\exec[\restrees](\tcol{+t.\axof{\mathcal{S}}} \uplus \testof{\mathcal{S}}{\varphi}) = \emptyset$.

A solution is to make incorrect proofs visible in the normal form. Our definition of execution makes this possible. In case of closed cycles, we can construct closed cyclic diagrams (which are accepted). Since proof-structures are always translated into exact constellations, such diagrams will always actualise into the empty star. As for cyclic diagrams with free rays, in the case of proof-structure, they will yield infinitely many stars by unfolding the loop infinitely many times. This makes incorrectness visible in the output of execution.

This problem is actually not new and already existed in previous GoI models. For instance, in Seiller's works, it was necessary to be able to detect cycles. The problem has been solved with a notion of \emph{wager} which is a value associated to proofs indicating the presence of cycles but we were able to simulate this idea by modifying the notion of execution instead.

We can now state the Danos-Regnier correctness criterion in the stellar resolution.

\begin{prop}
\label{prop:cycleconc}
If a connected multiplicative correctness hypergraph $\mathcal{S}^\varphi$ has no conclusion then it is cyclic.
\end{prop}
\begin{proof}
Since $\mathcal{S}^\varphi$ has no conclusion, all vertices are source of exactly one hyperedge. By the definition of proof-structure, all vertices are target of exactly one hyperedge. Hence, all vertices have a degree at least $2$. Therefore, by basic properties of cycles in graph theory, $\mathcal{S}^\varphi$ must be cyclic.
\end{proof}

\begin{thm}[Stellar correctness criterion]
\label{thm:correctness2}
A proof-structure $\mathcal{S}$ such that $\conclu{\mathcal{S}} = \{v_1, ..., v_n\}$ is MLL-certifiable if and only if for all switchings $\varphi$, we have: \[\exec(\tcol{+t.\axof{\mathcal{S}}} \uplus \testof{\mathcal{S}}{\varphi}) = [\ulocus{v_1}{X}, ..., \ulocus{v_n}{X}].\]
\end{thm}
\begin{proof}
Let $\varphi$ be a switching of $\mathcal{S}$. By Lemma~\ref{lem:structural}, we know that $\dgraph{\tcol{+t.\axof{\mathcal{S}}} \uplus \testof{\mathcal{S}}{\varphi}}$ is structurally equivalent to $\mathcal{S}^\varphi$.
\begin{itemize}
    \item $(\Rightarrow)$ Assume $\mathcal{S}^\varphi$ is connected and acyclic and so is $\dgraph{\tcol{+t.\axof{\mathcal{S}}} \uplus \testof{\mathcal{S}}{\varphi}}$.
    By Lemma~\ref{lem:uniqueness} and \ref{lem:exactness}, $\dgraph{\tcol{+t.\axof{\mathcal{S}}} \uplus \testof{\mathcal{S}}{\varphi}}$ has a unique correct diagram.
    We can construct a diagram $\delta$ by following the links of $\mathcal{S}^\varphi$. We have that $G_\delta$ is a tree containing all conclusions $v_1, ..., v_n$. Hence $\actu\delta$ is the star $[\ulocus{v_1}{X}, ..., \ulocus{v_n}{X}]$.

    \item $(\Leftarrow)$ Assume that $\exec(\tcol{+t.\axof{\mathcal{S}}} \uplus \testof{\mathcal{S}}{\varphi}) = [\ulocus{v_1}{X}, ..., \ulocus{v_n}{X}]$.
    Assume by contradiction that $\mathcal{S}^\varphi$ has at least two connected components. Assume that a component has no conclusion (because of cuts). Then, by Proposition~\ref{prop:cycleconc}, there is a cycle yielding infinitely many closed diagrams normalising into the empty star $\emptystar$. Hence, all connected components must have free rays corresponding to conclusion. By the independence of connected component (\cf Lemma~\ref{lem:independence}), we can independently execute each connected component. Since they correspond to deterministic and exact subconstellation, the normalisation produces the constellation $\phi_1 + ... + \phi_k$ for the $k$ connected component, contradicting the hypothesis that we normalise into a single star. Therefore, $\mathcal{S}^\varphi$ must be connected. Now, assume by contradiction that $\mathcal{S}^\varphi$ is cyclic.
    The cycle can either yield a closed diagram actualising into the empty star $\emptystar$ or pass through a conclusion and produce infinitely many stars containing conclusion rays. In both case, the normalisation is different from $[\ulocus{v_1}{X}, ..., \ulocus{v_n}{X}]$, contradicting the hypothesis. Therefore, $\mathcal{S}^\varphi$ must also be acyclic. This proves that $\mathcal{S}^\varphi$ must be a tree for any switching $\varphi$, \ie $\mathcal{S}$ is MLL-certifiable.
\end{itemize}
\end{proof}

The following corollary finally extends the logical correctness to MLL+MIX and suggest a more general variant which also captures MLL.

\begin{cor}
\label{cor:correctness}
Let $\mathcal{S}$ be a proof-structure and $\Phi := \tcol{+t.\axof{\mathcal{S}}} \uplus \testof{\mathcal{S}}{\varphi}$ be the constellation corresponding to the correctness hypergraph $\mathcal{S}^\varphi$ for some switching $\varphi$. We have:
\begin{itemize}
    \item $\mathcal{S}^\varphi$ is acyclic $\Leftrightarrow$ $\dgraph{\Phi}$ is acyclic $\Leftrightarrow$ $|\exec_A(\Phi)| < \infty$;
    \item $\mathcal{S}^\varphi$ is connected and acyclic $\Leftrightarrow$ $\dgraph{\Phi}$ is a deterministic tree $\Leftrightarrow$ $|\exec(\Phi)| = 1$;
    \item $\mathcal{S}^\varphi$ is connected and acyclic $\Leftrightarrow$ $\Phi$ normalises into the star of its uncoloured rays.
\end{itemize}
\end{cor}
\begin{proof}
The first equivalence of each point are direct consequences of Lemma~\ref{lem:structural}. It only remains to show the last equivalences.
\begin{itemize}
    \item If $\dgraph{\Phi}$ is acyclic, then by Lemma~\ref{lem:acyclicity}, $|\exec_A(\Phi)| < \infty$ because $\Phi$ is exact and deterministic (\cf Corollary~\ref{cor:deterministicexact}). Now assume that $|\exec_A(\Phi)| < \infty$. The proof of Theorem~\ref{thm:correctness2} shows that the presence of cycles in correctness hypergraphs is linked to the generation of infinitely many correct saturated diagrams. Hence it cannot be both cyclic and strongly normalising and has to be acyclic.
    \item We start from the previous point. If $\dgraph{\Phi}$ is also connected, then by Lemma~\ref{lem:uniqueness} and \ref{lem:exactness} and the fact that $\Phi$ is deterministic and exact (\cf \ref{cor:deterministicexact}), there is exactly a unique correct saturated diagram, hence a single star in the normal form. Conversely, if $\Phi$ normalises into a single star, its dependency graph must be both connected and acyclic, otherwise we would end up with either several stars (\cf proof of Theorem~\ref{thm:correctness2}) or infinitely many correct saturated diagrams (since cycles are related to non-termination as stated in the previous point).
    \item The third case corresponds to an alternative characterisation of correct proof-structures. Assume $\mathcal{S}^\varphi$ is connected and acyclic. Then $\dgraph{\Phi}$ is a deterministic tree by the previous point. By definition, uncoloured rays are the only free rays in $\Phi$. Since $\Phi$ is exact, it must produce a unique diagram corresponding to the cover tree of $\dgraph{\Phi}$. By definition, such a diagram reduces into the star of its free rays, hence the star of its uncoloured rays. Now, assume $\Phi$ normalises into the star of its uncoloured rays. The reasoning is the same as for the previous point.
\end{itemize}
\end{proof}

The analysis of the computational and logical content of proofs in the Transcendental Syntax leads to a decomposition of proof-structures and give a new outlook on what being a ``correct" proof means in proof theory.

\begin{defi}[Translation of a proof-structure]
\label{def:transps}
The translation of a proof-structure $\mathcal{S}$ is defined as the constellation $\mathcal{S}^\bigstar = (\axcutof{\mathcal{S}}, \formatof{\mathcal{S}})$ where $\formatof{\mathcal{S}}$ is called \emph{format}\footnote{\emph{Gabarit} in Girard's original papers.
We choose the term ``format" because it is less awkward in English and reminds of file formats in a computer.} and is defined by $\formatof{\mathcal{S}} := \{\testof{\mathcal{S}}{\varphi} \mid \varphi \text{ is a switching of $\mathcal{S}$}\}$.
\end{defi}

As shown in Theorem~\ref{thm:correctness2}, $\mathcal{S}^\bigstar$ corresponds to a proof-net if and only if it passes all the tests $\Phi \in \formatof{\mathcal{S}}$. In particular, any proof-structure can be seen as a program (its set of axioms) already coming with some implicit constraining tests. This corresponds to a sort of hidden pre-typing. Proof-nets are programs coming with tests it can passes. Hence tests corresponds to a \emph{certification} for programs. This demonstrates what Girard means by ``making the hidden assumptions of logic explicit" (\cf\Cref{sec:intro}).

Notice that these tests are entirely definable by MLL formulas (and thus, dependent of them) because only vertices of the lower part of $\mathcal{S}^\varphi$ are used in the translation of Definition~\ref{def:test}. We obtain a more general meaning of the idea of proof: a proof is a computational entity passing the tests corresponding to a certain notion of formula/specification yet to be defined.

\section{Emergence of formulas}\label{sec:formulas}

Generalising the correctness criterion of proof-nets actually gives rise to a notion of type (or formula). We need to fix a symmetric binary relation between constellations formalising what we mean by ``correctly passing a test". For instance, Corollary~\ref{cor:correctness} suggests three such relations we call $\perp_\mathrm{fin}$, $\perp^1$ and $\perp^R$ but others can be designed depending on what we want. The intention behind orthogonality relations is that they define linear negations for linear logic.

\begin{defi}[Orthogonality]
\label{def:ortho}
We define binary relations of \emph{orthogonality} between two constellations $\Phi_1$ and $\Phi_2$ \wrt a set of colours $A \subseteq C$:
\begin{itemize}
    \item $\Phi_1 \perp^\mathrm{fin}_A \Phi_2$ when $|\exec_A(\Phi_1 \uplus \Phi_2)| < \infty$;
    \item $\Phi_1 \perp^1_A \Phi_2$ when $|\exec_A(\Phi_1 \uplus \Phi_2)| = 1$;
    \item $\Phi_1 \perp^R_A \Phi_2$ when $\exec(\Phi_1 \uplus \Phi_2) = \{\roots{\Phi_1 \uplus \Phi_2}\}$ where $\roots{\Phi}$ is the star of uncoloured rays in $\Phi$.
\end{itemize}
The orthogonal of a set of constellations $\mathbf{A}$ is defined by $\mathbf{A}^{\bot_A} := \{\Phi \mid \forall \Phi' \in \mathbf{A}, \Phi \perp_A \Phi'\}$ for a relation of orthogonality $\bot$.
\end{defi}

In order to allow typing for partial evaluations, the orthogonality relation $\bot_A$ has to be parametrised by a set of colours $A$ but we omit this parameter when considering all colours in $C$.

The orthogonal $\mathbf{A}^{\bot_A}$ corresponds to the set of all constellations passing the tests of $\mathbf{A}$. But since test and tested are both constellations and that orthogonality relations are symmetric, they have interchangeable roles, hence $\mathbf{A}$ is also the set of constellations passing the tests of $\mathbf{A}^\bot$.

The orthogonality $\perp^\mathrm{fin}$ will define a fully complete model of MLL+MIX, while $\perp^1$ and $\perp^R$ (which captures more directly the correctness criterion for MLL) will define a fully complete model of MLL. However, those notions of orthogonality share most of the properties needed, and we therefore use the generic notation $\perp$ in the following to state results valid for all of them.

\begin{lem}[Invariance of orthogonality under execution]
\label{lem:invortho}
Let $\Phi$ and $\Phi'$ be constellations such that $\Phi \ucap_{A \cup B} \Phi' = \emptyset$. We have $\Phi \perp \Phi'$ if and only if $\exec(\Phi) \perp \Phi'$ for $\perp \in \{\perp^1, \perp^\mathrm{fin}, \perp^R\}$.
\end{lem}
\begin{proof}
These relation are satisfied when $P(\exec(\Phi \uplus \Phi'))$ is satisfied for some property $P$. By the lemma of partial pre-execution (\cf Lemma~\ref{lem:ppe}), we have $\exec(\exec(\Phi) \uplus \Phi') = \exec(\Phi \uplus \Phi')$.
Hence we have $P(\exec(\Phi \uplus \Phi'))$ if and only if $P(\exec(\exec(\Phi) \uplus \Phi'))$, meaning that we have $\Phi \perp \Phi'$ if and only if $\exec(\Phi) \perp \Phi'$.
\end{proof}

\subsection{Types as labels certified by tests (l'Usine)}\label{subsec:testing}

In this section, we construct formulas by generalising the logical correctness of \Cref{subsec:correctness}.

\begin{defi}[Type label]
\label{def:typelabel}
A \emph{type} is an object (or label) $A$ associated to a finite set of constellations $\tests{A}$ called its \emph{tests}. We say that a constellation $\Phi$ is of type $A$ \wrt $\perp$ if and only if $\Phi \in \tests{A}^\bot$.
\end{defi}

A type corresponds to a \emph{specification} for a computational entity (typically a program) certified by an associated set of \emph{tests} as we do in software engineering or formal methods. For instance, in model checking \cite{baier2008principles}, given an automata $\Phi$ (or labelled transition system), we would like to know whether it satisfies a specification $S$ (often written as a formula of a logic called LTL). It is then possible to check if $\Phi$ satisfies $S$ by turning $\lnot S$ into an automaton $\Phi_{\lnot S}$ and verifying if $\mathcal{L}(\Phi) \cap \mathcal{L}(\Phi_{\lnot S}) = \emptyset$, by analysing paths of the state graph of the automaton \cite[Section 3.6.3]{huth2004logic}. This is similar to how we turn a sequent $\vdash\Gamma$ into a set of tests (defined as constellations) allowing us to label/certify a constellation as a proof of $A$. Moreover, the Danos-Regnier's tests can also be considered as proofs of $A^\bot$ as we will see in Observation~\ref{obs:testsproofdual}.

The purpose of having finite set of tests is to make type checking computable. However, this is only happens under some conditions such as the orthogonality relation between computable. Even under these conditions, it is possible to ``trick'' tests so to create infinite loops and make effective type checking impossible. It shows that we need to consider testing \wrt a specific class of objects (for instance the universe of proof-structures) so to prevent such tricks to happen.

Although similar, typing with finite tests is not quite the type checking with typing rules which appears in typed $\lambda$-calculus. Girard's Usine is meant to check \emph{cut-free proofs} only, whereas it is possible to verify the type of normalisable terms for a sequent $\vdash (\lambda x. M) N : A$ without actually doing the normalisation. This is because the transcendental syntax distinguishes between:
\begin{itemize}
    \item characterising the shape of our logical objects (cut-free proofs), which corresponds to Usine and to the logical rules of sequent calculus;
    \item defining the use of our logical objects (interaction with cuts), which corresponds to Usage and to the cut rule of sequent calculus.
\end{itemize}
These notions are often mixed in proof theory: in order to even have the right to write an elimination rule such as modus ponens, we implicitly assume that we are given an object which has the shape of a proof of implication $A \Rightarrow B$ and that its interaction with any proof of $A$ will produce a proof of $B$. In other words, we assume an \emph{adequation} between Usine and Usage or that we have primitive objects (defined by finite tests) which will behave soundly \wrt some \emph{use} (behaving like functions in the case of modus ponens).

The definition of orthogonality and interactive testing leads to a reformulation of correctness criterion, showing that MLL sequents define type labels by themselves, independently of a proof-structure. This is based on the fact that the bottom part of proof-structure corresponds to the syntax tree of a sequent which is already a sort of pre-typing constraining atomic cut-elimination. By constructing a syntax hypergraph from a sequent, Definition~\ref{def:test} can be used.

\begin{defi}[Test of a sequent]
\label{def:mlltest}
Let $\vdash\Gamma$ be a sequent of MLL where $\Gamma \subseteq \fmll$ and all variables are distinct. We define the syntax tree of an MLL formula $A$ inductively:
\begin{itemize}
    \item $ST(X_i)$ and $ST(X_i^\bot)$ are vertex labelled by $X_i$ and $X_i^\bot$ respectively;
    \item $ST(A \otimes B)$ is an hyperedge labelled by $\otimes$ linking the conclusion of $ST(A)$ and $ST(B)$ as sources and having a vertex labelled by $A \otimes B$ as target;
    \item $ST(A \parr B)$ is an hyperedge labelled by $\parr$ linking the conclusion of $ST(A)$ and $ST(B)$ as sources and having a vertex labelled by $A \parr B$ as target.
\end{itemize}
The syntax hypergraph $ST(\vdash\Gamma)$ of $\vdash\Gamma$ is defined as the hypergraph disjoint union of all $ST(A_i)$ for $A_i \in \Gamma$. A switching (\cf Definition~\ref{def:correctnessgraph}) $\varphi$ still applies on $ST(\vdash\Gamma)$ as for correction hypergraphs. We write $ST(\vdash\Gamma)^\varphi$ for the switching $\varphi$ applied on the syntax hypergraph $ST(\vdash\Gamma)$.

The \emph{test} associated to the sequent $\vdash\Gamma$ and the switching $\varphi$ is defined as the constellation $\test{\vdash\Gamma}^\varphi$ such that $I_{\test{\vdash\Gamma}^\varphi} := V^{ST(\vdash\Gamma)^\varphi}$ (it is indexed by vertices of the syntax tree) and $\test{\vdash\Gamma}^\varphi[v] := v^\bigstar$ where $\ell(v)$ is not an atomic formula and $v^\bigstar$ is the translation of Definition~\ref{def:test}. Notice that we reject the translation of atomic formulas because they depend upon a proof-structure $\mathcal{S}$. This dependency is actually not necessary.

The \emph{set of tests} associated to the sequent $\vdash\Gamma$ is defined by $\tests{\vdash\Gamma} := \{\test{\vdash\Gamma}^\varphi \mid \varphi \text{ is a switching of } ST(\vdash\Gamma)\}$.
\end{defi}

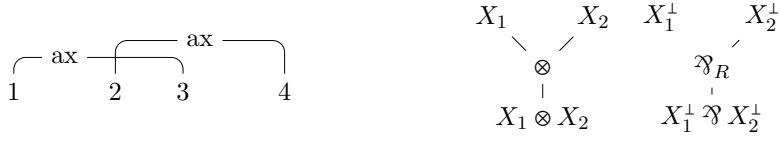
\begin{figure}
	\begin{minipage}{0.4\textwidth}
    \scalebox{0.9}{\begin{tikzpicture}
      \node at (-0.75, 0.75) (a) {1};
      \node at (0.75, 0.75) (b) {2};
      \node at (1.75, 0.75) (ad) {3};
      \node at (3.25, 0.75) (bd) {4};
      
      \node at (0, 1.25) (ax1) {ax};
      \node at (2, 1.5) (ax2) {ax};
      \draw[-, rounded corners=5pt] (ax1) -| (a);
      \draw[-, rounded corners=5pt] (ax1) -| (ad);
      \draw[-, rounded corners=5pt] (ax2) -| (b);
      \draw[-, rounded corners=5pt] (ax2) -| (bd);
    \end{tikzpicture}}
    \end{minipage}
    \begin{minipage}{0.4\textwidth}
    \scalebox{0.9}{\begin{tikzpicture}
      \node at (-0.75, 0.75) (a) {$X_1$};
      \node at (0.75, 0.75) (b) {$X_2$};
      \node at (0, 0) (tens) {$\otimes$};
      \node at (0, -0.75) (ab) {$X_1 \otimes X_2$};
      \draw[-] (a) -- (tens);
      \draw[-] (b) -- (tens);
      \draw[-] (tens) -- (ab);
      
      \node at (1.75, 0.75) (ad) {$X_1^\bot$};
      \node at (3.25, 0.75) (bd) {$X_2^\bot$};
      \node at (2.5, 0) (par) {$\parr_R$};
      \node at (2.5, -0.75) (adbd) {$X_1^\bot \parr X_2^\bot$};
      \draw[-] (bd) -- (par);
      \draw[-] (par) -- (adbd);
    \end{tikzpicture}}
    \end{minipage}
    \caption{We expect this proof-structure to be able to pass any test of $\tests{\vdash X_1 \otimes X_2, X_1^\bot \parr X_2^\bot}$. However, since the function symbols used in tests are not compatible with the ones of the proof-structure, we need a conversion of function symbols to allow interaction.}
    \label{fig:readdressing}
\end{figure}

We now defined MLL sequents as type labels in the sense of Definition~\ref{def:typelabel}. However, there is a minor technical problem: arbitrary constellations may not match with the tests we defined because of a difference of function symbols, as illustrated in \Cref{fig:readdressing}. One solution is to extended the notion of colour shifts of Definition~\ref{def:colourshift} to change rays in order to force the $\alpha$-unification

\begin{defi}[Conjugation]
\label{def:conjugation}
A \emph{conjugation} $\mu : \idrays{\mathbb{S}} \rightarrow \idrays{\mathbb{S}'}$ between two signatures $\mathbb{S}$ and $\mathbb{S}'$ is a function replacing the rays of a constellation such that it preserves its structure, \ie $\dgraph{\mu(\Phi)} \simeq \dgraph{\Phi}$.
\end{defi}

Another solution is to use a computational realisation of conjugations by using stars $[-t, +u]$ where $+t$ is a ray of the vehicle and $-u$ is a ray of a test. This corresponds to a sort of generalised cut allowing a trivial connexion between two rays. We give a more general definition of that idea.

\begin{defi}[Adapter]
\label{def:adapter}
Let $\Phi$ be a constellation. An \emph{adapter} for $\Phi$ is a star $[t, u]$ where $t'$ and $u'$ are rays in $\Phi$ which are respectively $t$ and $u$ with opposite polarity.
\end{defi}

Conjugations induce adapters. Whenever we have a conjugation $\mu$ such that $\mu(r_1) = r_2$, we can construct an adapter $[r_1', r_2']$ where $r_i'$ has a polarity opposite to $r_i$. We use adapters and conjugations indistinctly in this paper.

Notice that the translation of atomic formulas in Definition~\ref{def:mlltest} actually corresponds to adapters between a vehicle and a test, hence tests are indeed independent of vehicles. This dependency is only artificial and appears when considering proof-structures as an entity which cannot be decomposed.

\begin{prop}[Correspondence between proof-structure tests and sequent tests]
\label{prop:eqtest}
Let $\mathcal{S}$ be a cut-free proof-structure. For all switching $\varphi$ of $\mathcal{S}$, there exists an MLL sequent $\vdash\Gamma$ and a constellation of adapters $\Phi$ such that $\exec(\testof{\mathcal{S}}{\varphi}) = \exec(\test{\vdash\Gamma}^\varphi \uplus \Phi)$.
\end{prop}
\begin{proof}
The constellation $\testof{\mathcal{S}}{\varphi}$ corresponds to the syntax tree of a formula with exactly one premise cut for each $\parr$ vertex. Hence, it naturally induces a sequent $\vdash \Delta$ and we defined $\Gamma := \Delta$. 
The constellation $\test{\vdash\Gamma}^\varphi$ structurally corresponds to $\testof{\mathcal{S}}{\varphi}$ without the upper rays of colour $\tcol{-t}$ which allows connexion with the right vehicle. Apart from that, they both use the same translation function $(\cdot)^\bigstar$ on vertices for correctness hypergraphs. Assume we have a star for atom $[\tcol{-t.p_v(t)}, \ccol{+c.q_w(x)}]$ related to some star $[\ccol{-c.q_w(x)}, ...]$ in $\test{\vdash\Gamma}^\varphi$. During the execution, they will merge into $[\tcol{-t.p_v(t)}, ...]$. However, it is possible to construct $\Phi$ so to reproduce this step with an adapter $[\tcol{-t.p_v(t)}, \ccol{+c.q_A(x)}]$ (by definition, the star $[\ccol{-c.q_A(x)}, ...]$ which is isomorphic to $[\ccol{-c.q_w(x)}, ...]$ must be present in $\testof{\mathcal{S}}{\varphi}$). Moreover, because of the structural equivalence between the two constellations, they only differ by conjugation. It is then possible to extend $\Phi$ so that $\test{\vdash\Gamma}^\varphi$ is turned exactly into $\testof{\mathcal{S}}{\varphi}$. It follows that the two constellations must have the same normal form. 
\end{proof}

\begin{defi}[Typing]
We say that a constellation $\Phi$ is of type $\vdash\Gamma$, written $\vdash \Phi : \Gamma$ when $\Phi \in (\testof{\vdash\Gamma}{\varphi} \uplus \Phi_\mu)^\bot$ for a set of adapters $\Phi_\mu$ and all switchings $\varphi$ of $\vdash\Gamma$.
\end{defi}

\begin{prop}[Reformulation of logical correctness]
\label{prop:reformulationcorrectness}
A cut-free proof-structure $\mathcal{S}$ is MLL-certifiable if and only if there exists a sequent $\vdash\Gamma$ and a constellation of adapters $\Phi$ such that $\vdash \tcol{+t.\axof{\mathcal{S}}} : \Gamma$ with $\bot \in \{\bot^1, \bot^R\}$. The same statement holds for MLL+MIX \wrt $\perp^{\mathrm{fin}}$.
\end{prop}
\begin{proof}
By Proposition~\ref{prop:eqtest}, there exist some sequent $\vdash\Gamma$ such that $\testof{\mathcal{S}}{\varphi}$ is simulated by $\Phi \uplus \test{\vdash\Gamma}^\varphi$ for some constellation of adapters $\Phi$.
By invariance of orthogonality under execution (\cf Lemma~\ref{lem:invortho}), this connexion is equivalent to a connexion between $\tcol{+t.\axof{\mathcal{S}}}$ and $\testof{\mathcal{S}}{\varphi}$. The orthogonality $\tcol{+t.\axof{\mathcal{S}}} \in \tests{\vdash\Gamma}^\bot$ and the same statement for MLL+MIX (\wrt $\perp^{\mathrm{fin}}$) both hold by a direct consequence of Corollary~\ref{cor:correctness}.
\end{proof}

Now that we have finite tests able to certify computational entities, what remains is to be able to express the \emph{real use} of these objects (defined by the set of their potential partners in interaction), which is usually infinite. Girard's Usine is then only an effective approximation of this ideal use.

\subsection{Interactive typing (l'Usage)}\label{subsec:behaviours}

By using an idea of \emph{interactive typing} which was already present in ludics \cite{girard2001locus} and in the Geometry of Interaction \cite{goi5,seiller2012interaction}, it is possible to define ``semantic-free formulas". Such formulas are defined as set of constellations, not from a given semantics but from how the constellations interact with each other. We need two ingredients: a notion of \emph{interaction} (the execution of constellations) and a symmetric and binary \emph{orthogonality} relation which opposes constellations. This relation represents a \emph{point of view} on interaction and formalises what it means to ``interact correctly".

This actually extends the previous idea of type but instead of arbitrary tests, a constellation is given a meaning by all its possible interaction with other constellations, relatively to a specific point of view. Since these potential opponents still define the meaning of a constellation, we keep the term of ``test'' (although effective testing is no more possible in general because a set of tests can be infinite).

The constellations are grouped into arbitrary sets called \emph{pre-behaviours}, giving rise a notion of formula corresponding to a computational version of phase semantics \cite[Section II.5]{linearlogic}.

\begin{defi}[Pre-behaviour]
A \emph{pre-behaviour} $\mathbf{A}$ is a set of constellations.
\end{defi}

We now define the notion of \emph{behaviour} which corresponds to the formulas/types appearing in linear logic. They represent idealised logical notions that we can only approximate if we wish for an effective type checking.

\begin{defi}[Behaviour]
\label{def:behaviour}
A pre-behaviour $\mathbf{A}$ is a \emph{behaviour} when there exists a pre-behaviour $\mathbf{B}$ such that $\mathbf{A} = \mathbf{B}^\bot$.
\end{defi}

More intuitively, a behaviour is a group of computational objects which is entirely characterised by a (potentially infinite) set of tests: a pre-behaviour $\mathbf{A}$ is a behaviour when there exists a set of tests (constellations) $\mathbf{B}$ such that $\mathbf{A}$ is exactly the set of constellations passing all the tests of $\mathbf{B}$. In other words, $\mathbf{A}$ is a behaviour if and only if it is \emph{testable}.

\begin{lem}[Invariance of typing under execution]
\label{lem:typeinv}
Let $\Phi$ be a constellation and $\mathbf{A}$ a behaviour. We have $\Phi \in \mathbf{A}$ if and only if $\exec(\Phi) \in \mathbf{A}$.
\end{lem}
\begin{proof}
If $\mathbf{A}$ is a behaviour then $\mathbf{A} = \mathbf{A}^{\bot\bot}$, meaning that $\mathbf{A}$ is characterised by some tests $\mathbf{A}^\bot$. Hence we have to show that $\Phi \perp \Phi'$ for any $\Phi' \in \mathbf{A}^\bot$ if and only if $\exec(\Phi) \perp \Phi'$. This is the consequence of the invariance of orthogonality under execution (\cf Lemma~\ref{lem:invortho}).
\end{proof}

There is an alternative (more standard) definition of behaviours which is called \emph{bi-orthogonal closure}. It states a sort of balance between tests and tested. This is actually something very important we require in linear logic and which is not true in intuitionistic logic\footnote{In intuitionistic logic, we do not have $\lnot\lnot A = A$ for any formula $A$.}: the linear negation is involutive.

\begin{prop}[Bi-orthogonal closure]
A pre-behaviour $\mathbf{A}$ is a behaviour if and only if $\mathbf{A} = \mathbf{A}^{\bot\bot}$.
\end{prop}
\begin{proof}
The proof can be found in the literature \cite[Proposition 15]{joinet2021abstraction}.
\end{proof}

In order to define the tensor of two behaviours (corresponding to an interactive version of the usual tensor type label), we have to exclude any interaction between them because we want the tensor to connect two independent proof-structures. Definition~\ref{def:sharedvar} of set of variables shared by two constellations makes this possible.

\begin{defi}[Disjointness of behaviours]
Let $\mathbf{A}$ and $\mathbf{B}$ be two behaviours and a set of colours $C' \subseteq C$. They are \emph{disjoint} when for all $\Phi_A \in \mathbf{A}$ and $\Phi_B \in \mathbf{B}$, we have $\Phi_A \ucap_{C'} \Phi_B = \emptyset$.
\end{defi}

When two behaviours $\mathbf{A}$ and $\mathbf{B}$ are disjoint, for any pair of constellations $\Phi_A \in \mathbf{A}$ and $\Phi_B \in \mathbf{B}$, there is no path from one constellation to the other in $\dgraph{\Phi_A \uplus \Phi_B}$: for instance, if we had a path from $\Phi_1$ to a variable $X$ in $\Phi_2$, this variable is still accessible from $\Phi_2$, hence $X$ is shared by the two constellations.

\begin{defi}[Pre-tensor]
\label{def:pretensor}
Let $\mathbf{A}$ and $\mathbf{B}$ be disjoint pre-behaviours. We define their pre-tensor by $\mathbf{A} \odot \mathbf{B} = \{ \Phi_1 \uplus \Phi_2 \mid \Phi_1 \in \mathbf{A}, \Phi_2 \in \mathbf{B} \}$.
\end{defi}

\begin{defi}[Tensor]
\label{def:tensor}
Let $\mathbf{A}$ and $\mathbf{B}$ be disjoint behaviours. We define their tensor by
\[\mathbf{A} \otimes \mathbf{B} = (\mathbf{A} \odot \mathbf{B})^{\bot\bot}.\]
\end{defi} 

The pre-tensor is the natural definition of the tensor product pairing constellations of two pre-behaviours. The real tensor product adds a bi-orthogonal closure $(\cdot)^{\bot\bot}$ in order to ensure that we get a behaviour (it is not necessarily the case without the closure, depending on the orthogonality we consider). It is indeed a generalisation of the usual tensor because depending on the orthogonality relation, its orthogonal can contain way more than what we expect from proof-structures because of the huge space of objects provided by stellar resolution. In case $\mathbf{A} \odot \mathbf{B} = \mathbf{A} \otimes \mathbf{B}$, we have what we call an \emph{internal completeness} property.

\begin{prop}[Commutativity and associativity of tensor]
\label{prop:proptensor}
Given $\mathbf{A}, \mathbf{B}, \mathbf{C}$ pairwise disjoint behaviours, we have (1) $\mathbf{A} \otimes \mathbf{B} = \mathbf{B} \otimes \mathbf{A}$ and  (2) $\mathbf{A} \otimes (\mathbf{B} \otimes \mathbf{C}) = (\mathbf{A} \otimes \mathbf{B}) \otimes \mathbf{C}$.
\end{prop}
\begin{proof}
(1) By the definition of tensor, we have $\Phi_1 \uplus \Phi_2 \in \mathbf{A} \otimes \mathbf{B}$ when $\Phi_1 \uplus \Phi_2 \in \{ \Phi_1 \uplus \Phi_2 \mid \Phi_1 \in \mathbf{A}, \Phi_2 \in \mathbf{B} \}^{\bot\bot}$. We also have $\Phi_2 \uplus \Phi_1 \in \mathbf{B} \otimes \mathbf{A}$.
But since $\Phi_1 \uplus \Phi_2 = \Phi_2 \uplus \Phi_1$ by commutativity of multiset disjoint union, we obtain $\mathbf{A} \otimes \mathbf{B} = \mathbf{B} \otimes \mathbf{A}$.
(2) In the same fashion, by using the associativity of multiset disjoint union, we obtain $\mathbf{A} \otimes (\mathbf{B} \otimes \mathbf{C}) = (\mathbf{A} \otimes \mathbf{B}) \otimes \mathbf{C}$.
\end{proof}

The other connectives are then defined by interactive testing, \eg the elements of $\mathbf{A} \parr \mathbf{B}$ are the elements passing the tests of $\mathbf{A}^\bot \otimes \mathbf{B}^\bot$. This is why we can speak about \emph{interactive types} as we did in the introduction of this paper.

\begin{defi}[Par and linear implication]
\label{def:parimpl}
Let $\mathbf{A}, \mathbf{B}$ be disjoint behaviours. We define:
\[\mathbf{A} \parr \mathbf{B} = (\mathbf{A}^\bot \otimes \mathbf{B}^\bot)^\bot \quad\text{and}\quad \mathbf{A} \multimap \mathbf{B} = \mathbf{A}^\bot \parr \mathbf{B}.\]
\end{defi}

\begin{rem}[Implicit exchange]
\label{rem:implex}
The commutativity and associativity of $\otimes$ are preserved for the $\parr$. For instance $\mathbf{A} \parr \mathbf{B} = (\mathbf{A}^\bot \otimes \mathbf{B}^\bot)^\bot = (\mathbf{B}^\bot \otimes \mathbf{A}^\bot)^\bot = \mathbf{B} \parr \mathbf{A}$. This corresponds to the fact that the exchange rule is implicit in usual linear logic.
\end{rem}

\begin{figure}
    \begin{tikzpicture}
        \node at (0, 0) (c1) {$\Phi_1 = [X, +c(X)]$};
        \node at (6, 0.5) (c2) {$[-c(\lc \cdot X)] = \Phi_2$};
        \node at (6, -0.5) (c3) {$[-c(\rc \cdot X)] = \Phi_3$};
        \draw[dotted] (c1.0) edge[out=45, in=180] (c2.180);
        \draw[dotted] (c1.0) edge[out=-45, in=180] (c3.180);
    \end{tikzpicture}
    \caption{Counter-example of non-associativity. We have $\exec_{\{c\}}(\Phi_1 \uplus \Phi_2) = [\lc \cdot X]$ and $\exec_{\{c\}}(\exec_{\{c\}}(\Phi_1 \uplus \Phi_2) \uplus \Phi_3) = [-c(\rc \cdot X)]+[\lc \cdot X]$, but $\exec_{\{c\}}(\Phi_1 \uplus \exec_{\{c\}}(\Phi_2 \uplus \Phi_3)) = [-c(\lc \cdot X)]+[\rc \cdot X]$ which is different.}
    \label{fig:nonassociativity}
\end{figure}

In \Cref{fig:nonassociativity}, we show that associativity fails when execution is treated as a binary operator on constellations. However, this property is fundamental when speaking about (categorical \cite{mellies2009categorical, schalk2004categorical}) models of linear logic. We need a restriction on the interaction between constellations as in Seiller's works \cite[Proposition 12]{seiller2012interaction}\cite[Theorem 24]{seiller2016interaction} where the same problem exists.

A technical precondition is defined for the associativity, and \emph{trefoil property} \cite[Theorem 40]{seiller2016interaction} is stated as a corollary. In particular, the trefoil property ensures that one can define a $\ast$-autonomous category, which characterises denotational models of MLL \cite{seely1987linear}.

\begin{thm}[Associativity of execution]
\label{thm:assocexec}
Choose a set of colours $A\subseteq C$. For constellations $\Phi_1, \Phi_2$ and $\Phi_3$ such that $\Phi_1 \ucap_A \Phi_2 \ucap_A \Phi_3 = \emptyset$, we have:
\[\exec_A(\Phi_1 \uplus \exec_A(\Phi_2 \uplus \Phi_3)) =  \exec_A(\exec_A(\Phi_1 \uplus \Phi_2) \uplus \Phi_3).\]
\end{thm}
\begin{proof}
Assume we have $\Phi_1 \ucap_A \Phi_2 \ucap_A \Phi_3 = \emptyset$. Hence, by definition, no variable is shared by the three constellations. Let $P(x^i_j)$ with $x^i_j$ a variable using the notations of Definition~\ref{def:sharedvar}, be the set of paths reaching $x^i_j$ in $\dgraph[A]{\Phi_1 \uplus \Phi_2 \uplus \Phi_3}$. By the previous statement, these paths traverse at most two constellations in $\{\Phi_1, \Phi_2, \Phi_3\}$. By using the reasoning of the proof of partial pre-execution (\cf Lemma~\ref{lem:ppe}), the paths $P(x^i_j)$ traversing $\Phi_2$ and $\Phi_3$ can be reduced with no effect on other connexions (since no variables are shared). Hence, the stars of $\Phi_1$ can connect to the stars of $\exec_A(\Phi_2 \uplus \Phi_3)$ in the same way as in $\Phi_2 \uplus \Phi_3$. It follows that $\exec_A(\Phi_1 \uplus \exec_A(\Phi_2 \uplus \Phi_3)) = \exec_A(\Phi_1 \uplus \Phi_2 \uplus \Phi_3)$. By the same reasoning, we also have $\exec_A(\exec_A(\Phi_1 \uplus \Phi_2) \uplus \Phi_3) = \exec_A(\Phi_1 \uplus \Phi_2 \uplus \Phi_3)$, hence execution is associative.
\end{proof}

\begin{thm}[Trefoil Property for execution-based orthogonality]
\label{thm:trefoil}
Choose a set of colours $A\subseteq C$. For constellations $\Phi_1, \Phi_2, \Phi_3$ and for $i, j, k \in \{1, 2, 3\}$ such that $\Phi_1 \ucap_A \Phi_2 \ucap_A \Phi_3 = \emptyset$, we have:
\[\Phi_1 \perp_A \exec_A(\Phi_2 \uplus \Phi_3) \quad\text{if and only if}\quad \exec_A(\Phi_1 \uplus \Phi_2) \perp_A \Phi_3.\]
\end{thm}
\begin{proof}
Assume that $P$ is a property corresponding to the orthogonality relation $\bot$ based on execution, \ie we have $\Phi_1 \perp_A \Phi_2$ if and only if $P(\exec(\Phi_1 \uplus \Phi_2))$.
The statement can be rewritten as follows: $P(\exec_A(\Phi_1 \uplus \exec_A(\Phi_2 \uplus \Phi_3))) \text{ if and only if } P(\exec_A(\exec_A(\Phi_1 \uplus \Phi_2) \uplus \Phi_3))$. This is a direct consequence of the associativity (\cf Theorem~\ref{thm:assocexec}).
\end{proof}

The trefoil property leads to the \emph{adjunction}\footnote{Corresponding to the categorical adjunction in cartesian closed categories.} which has been stated in previous models of GoI \cite[Theorem 3]{goi5} or in ludics.

\begin{cor}[Adjunction]
\label{cor:adjunction}
Choose a set of colours $A \subseteq C$. For all constellations $\Phi_f$, $\Phi_a$ and $\Phi_b$ such that $\Phi_a \ucap_A \Phi_b = \emptyset$, we have:
\[
\Phi_f \perp_A \Phi_a \uplus \Phi_b \quad\text{if and only if}\quad \exec_A(\Phi_f \uplus \Phi_a) \perp_A \Phi_b.
\]
\end{cor}
\begin{proof}
By symmetry of orthogonality relations and invariance of orthogonality under execution (\cf Lemma~\ref{lem:invortho}), we have $\Phi_f \perp_A \Phi_a \uplus \Phi_b$ if and only if $\Phi_f \perp_A \exec_A(\Phi_a \uplus \Phi_b)$. In order to conclude with the trefoil property, it remains to show the precondition, \ie that we have $\Phi_f \ucap_A \Phi_a \ucap_A \Phi_b = \emptyset$. We assumed $\Phi_a \ucap_A \Phi_b = \emptyset$, meaning that no variable were shared by both $\Phi_a$ and $\Phi_b$. It follows that a variable cannot be shared by $\Phi_f, \Phi_a$ and $\Phi_b$ at the same time because otherwise, it would be shared by $\Phi_a$ and $\Phi_b$ as well.
\end{proof}

Thanks to the adjunction, it is possible to define a more intuitive linear implication seeing a constellation $\Phi_f$ as a function interacting with a constellation $\Phi_a$ as argument.

\begin{prop}[Alternative linear implication]
\label{prop:altlin}
Let $\mathbf{A}, \mathbf{B}$ be two disjoint behaviours. We have $\mathbf{A} \multimap \mathbf{B} = \{\Phi_f \mid \forall\  \Phi_a \in \mathbf{A}, \Phi_f \bot \Phi_a \text{ and } \exec(\Phi_f \uplus \Phi_a) \in \mathbf{B}\}$.
\end{prop}
\begin{proof}
By Definition~\ref{def:parimpl}, we have $\mathbf{A} \multimap \mathbf{B} = (\mathbf{A} \otimes \mathbf{B}^\bot)^\bot$. We have $\Phi_f \in \mathbf{A} \multimap \mathbf{B}$ if and only if for all $\Phi_a \in \mathbf{A}$, $\Phi_f \bot \Phi_a$ and $\exec(\Phi_f \uplus \Phi_a) \in \mathbf{B}$.
Since $\mathbf{B}$ is a behaviour, by Definition~\ref{def:behaviour}, there exists $\Phi_{b'} \in \mathbf{B}^\bot$ such that $\exec(\Phi_f \uplus \Phi_a) \bot \Phi_{b'}$.
By the adjunction (\cf Corollary~\ref{cor:adjunction}), $\Phi_f \bot (\Phi_a \uplus \Phi_{b'})$, hence $\Phi_f \in (\mathbf{A} \otimes \mathbf{B}^\bot)^\bot$. The proof only relies on equivalences hence a bi-inclusion is proved.
\end{proof}

\subsection{The case of multiplicative units}\label{subsec:units}

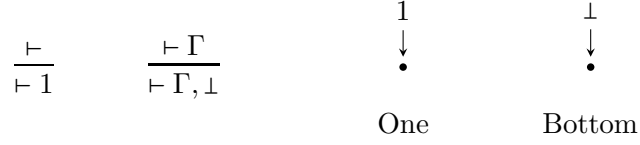
\begin{figure}
\centering
\begin{minipage}{0.3\textwidth}
    \begin{prooftree}
    \hypo{\vdash}
    \infer1{\vdash 1}
    \end{prooftree}
    \hspace{1cm}
    \begin{prooftree}
    \hypo{\vdash \Gamma}
    \infer1{\vdash \Gamma, \bot}
    \end{prooftree}
\end{minipage}
\begin{minipage}{0.3\textwidth}
\begin{tikzpicture}
    \node at (0, 0) (o) {$1$};
    \node[dot] at (0, -0.75) (c) {};
    \draw[-stealth] (o) -- (c);
    \node at (0, -1.5) (label) {One};
\end{tikzpicture}
\hspace{1cm}
\begin{tikzpicture}
    \node at (0, 0) (b) {$\bot$};
    \node[dot] at (0, -0.75) (c) {};
    \draw[-stealth] (b) -- (c);
    \node at (0, -1.5) (label) {Bottom};
\end{tikzpicture}
\end{minipage}
\caption{Rules for the units of MLL. Two hyperedges with no input and a single output are added in the construction of proof-structures. }
\label{fig:units}
\end{figure}

In this paper, we only mentioned MLL without units but linear logic is often presented with two formulas $1$ and $\bot$ corresponding to neutral elements for the $\otimes$ and $\parr$ connectives respectively. New rules and links for MLL units are presented in \Cref{fig:units}.

Now, we look for behaviours corresponding to neutral elements for $\otimes$ and $\parr$ respectively. It is possible to define a pre-behaviour $\Perp$ called a \emph{pole} \cite[Definition 3.5]{miquey2017classical} such that $\Phi \perp \Phi'$ if and only if $\exec(\Phi \uplus \Phi') \in \Perp$ for an execution-based orthogonality $\perp$ and $\Perp$ must be closed under \emph{anti-evaluation}, \ie if $\Phi \in \Perp$ and $\exec(\Phi') = \Phi$, then $\Phi' \in \Perp$. For instance, if we consider $\perp^R$, then $\Perp$ is the set of all constellations normalising into a single uncoloured star. The pole will be useful for a definition of neutral elements.

A natural choice of behaviour for the neutral element of $\otimes$ \wrt $\perp^R$ is the pre-behaviour $\{\emptyset\}$ only containing the empty constellation since $\Phi \uplus \emptyset = \Phi$ for any constellation $\Phi$. Fortunately, it is a behaviour.

\begin{prop}
The pre-behaviour $\{\emptyset\}$ is a behaviour.
\end{prop}
\begin{proof}
A constellation of $\{\emptyset\}^\bot$ must self-normalise into the set of its roots since $\emptyset$ has not effect when in interaction with another constellation. We have $\{\emptyset\}^\bot = \Perp$. Now, a constellation $\Phi \in \Perp^\bot$ is a constellation such that when it interacts with a constellation $\Phi' \in \Perp^\bot$, we have $\exec(\Phi \uplus \Phi') \in \Perp$. We can theoretically imagine that $\Phi$ has rays linked to $\Phi$ but this is impossible because $\Phi'$ is self-normalising into an element of $\Perp$ by constructing a saturated diagram which cannot be extended and which must be present in the normal form. Actually, $\Phi$ must be the empty constellation because otherwise we would get more than the star of roots. Therefore, $\Perp^\bot = \{\emptyset\}^{\bot\bot} = \{\emptyset\}$.
\end{proof}

\begin{defi}[One]
We define the behaviour $\mathbf{1} := \{\emptyset\} = \Perp^\bot$.
\end{defi}

\begin{prop}
We have $\mathbf{A} \otimes \mathbf{1} = \mathbf{A}$ for any behaviour $\mathbf{A}$.
\end{prop}
\begin{proof}
By definition, we have $\mathbf{A} \otimes \mathbf{1} = \{\Phi_A \uplus \emptyset \mid \Phi_A \in \mathbf{A}\}^{\bot\bot} = \{\Phi_A \mid \Phi_A \in \mathbf{A}\}^{\bot\bot} = \mathbf{A}^{\bot\bot} = \mathbf{A}$.
\end{proof}

As for bottom, as usual in linear logic, we define it as $\mathbf{1}^\bot = \Perp$.

\begin{prop}
The pre-behaviour $\mathbf{1}^\bot = \{\emptyset\}^\bot = \Perp$ is a behaviour.
\end{prop}
\begin{proof}
Since it is known that $\mathbf{A}^\bot = \mathbf{A}^{\bot\bot\bot}$ for any behaviour $\mathbf{A}$ \cite[Corollary 9]{joinet2021abstraction}, it follows that $\mathbf{1}^\bot$ (and thus $\{\emptyset\}^\bot$) is a behaviour.
\end{proof}

\begin{defi}[Bottom]
We define the behaviour $\simbot := \mathbf{1}^\bot$.
\end{defi}

\begin{prop}
We have $\mathbf{A} \parr \simbot = \mathbf{A}$ for any behaviour $\mathbf{A}$ when considering $\perp^R$.
\end{prop}
\begin{proof}
We have $\mathbf{A} \parr \simbot = (\mathbf{A}^\bot \otimes \simbot^\bot)^\bot = (\mathbf{A}^\bot \otimes \{\emptyset\}^{\bot\bot})^\bot = (\mathbf{A}^\bot \otimes \{\emptyset\})^\bot = \mathbf{A}^{\bot\bot} = \mathbf{A}$ (since $\mathbf{A}$ is a behaviour).
\end{proof}

\begin{prop}
We have $\mathbf{A}^\bot = \mathbf{A} \multimap \simbot$ for any behaviour $\mathbf{A}$ when considering $\perp^R$.
\end{prop}
\begin{proof}
We have $\mathbf{A} \multimap \simbot = \mathbf{A}^\bot \parr \simbot$. Since $\simbot$ is a neutral element for $\parr$, it follows that $\mathbf{A}^\bot \parr \simbot = \mathbf{A}^\bot$.
\end{proof}

We defined interactive types for units which correspond to idealised neutral elements (Girard's \emph{Usage}). Now, considering a constellation $\Phi$ in the wild, are we able to effectively tell whether it is in $\mathbf{1}$ (respectively $\simbot$) or not ? (Girard's \emph{Usine}).

We consider $\perp^R$. In order to tell if $\Phi \in \simbot$, we can use the fact that $\simbot = \{\emptyset\}^\bot$. Hence, it is sufficient to consider the set of tests $\{\emptyset\}$. When testing $\Phi$ against the empty constellation $\emptyset$, if we have $\Phi \perp \emptyset$ then $\Phi \in \simbot$.
As for $\mathbf{1}$, we just need to be able to tell if $\Phi = \emptyset$. This can be done with any constellation of $\mathbf{1}^\bot = \simbot$.

This provides a notion of correct constellations for multiplicative units. However, although they fulfil their role as constellations having the behaviour of neutral elements for multiplicative connectives, it is not quite the \emph{real thing} as they do not exactly correspond to the units of proof-nets. In particular, if we look at the rule for $\bot$, the constant $\bot$ is introduced in a given context $\Gamma$ to which it is dependent. Hence, it will either be disconnected when considering a switching in a correct proof-structure. This breaks the connectedness condition of the Danos-Regnier correctness criterion. The usual hack is usually to consider links (called ``jumps'' \cite[Section A.2]{girard1996proof}) between $\bot$ nodes and either axioms or $1$ nodes to represent the dependency between $\bot$ and its context. Girard's idea \cite[Section 2.1.1]{girard2018logique} is to encode multiplicative units in second order linear logic because of this non-local dependency but we do not discuss it in this paper.

\section{Conservativity and Adequacy}\label{sec:conservativityadequacy}

In this section, we propose two links:
\begin{itemize}
    \item a proof of conservativity \wrt the original model of proof-nets for $\perp^\mathrm{fin}$ in order to capture MLL+MIX provability and for both $\perp^R$ and $\perp^1$ in order to capture MLL provability. For that purpose, we state \emph{soundness} and \emph{completeness} theorems;
    \item a link between Usine and Usage (called \emph{adequacy} by Girard) showing that the correctness criterion is sufficient to guarantee a sound use of proofs (interaction by cuts).
\end{itemize}

We interpret formula labels by behaviours where distinct behaviours are associated to occurrences of variables by a function called \emph{basis of interpretation}.
Following previous works of Seiller \cite[Definition 46]{seiller2012interaction}, the behaviours corresponding to formula labels are \emph{localised} formulas: they are defined using the same grammar as MLL formulas, except that variables are of the form $X_i(t)$, where $t$ is a term (here representing the path address described in Definition~\ref{def:addr}) used to distinguish occurrences of atomic formulas.
Two behaviours $X_i(t)$ and $X_i(u)$ with $t \neq u$ represent the same atom at different locations and should correspond to the same behaviour modulo conjugation.

\begin{defi}
\label{def:basisinterp}
A \emph{basis of interpretation} is a function $\Omega$ producing a behaviour $\Omega(A, i, t)$ when given a formula $A \in \fmll$, a natural number $i$ (index of occurrence) and a term $t \in \setofaddresses{S}$ (\cf Definition~\ref{def:addr}). A basis of interpretation has to satisfy the condition that $\Omega(A, i, t) + [\tlocus{A}{t}, \tlocus{B}{u}] = \Omega(B, j, u)$ when $i=j$ and otherwise $\Omega(A, i, t)$ and $\Omega(B, j, u)$ are disjoint, such that $\mathbf{A}+\phi = \{\Phi+\phi \mid \Phi \in \mathbf{A}\}$ for a behaviour $\mathbf{A}$ and a star $\phi$.
\end{defi}

\begin{defi}[Interpretation of MLL formulas]
Given a basis of interpretation $\Omega$, a formula $C$ representing the conclusion of a sequent, and an MLL formula occurrence $A$ identified by a unique unary function symbol $\ulocus{A}{x}$ (\cf Definition~\ref{def:addr}). We define the \emph{interpretation} $\interp{A, t}_\Omega$ along $\Omega$ and a term $t$ (encoding the address of $A$ \wrt a conclusion $C$) inductively:
\begin{itemize}
    \item $\interp{C, X_i, t}_\Omega = \Omega(C, i, t)$;
    \item $\interp{C, X_i^\bot, t}_\Omega = \Omega(C, i, t)^\bot$;
    \item $\interp{C, A\otimes B, t}_\Omega = \interp{C, A, \lc \cdot t} _\Omega \otimes \interp{C, B, \rc \cdot t}_\Omega$;
    \item $\interp{C, A\parr B, t}_\Omega = \interp{C, A, \lc \cdot t}_\Omega \parr \interp{C, B, \rc \cdot t}_\Omega$.
\end{itemize}
We write $\interp{C}$ for $\interp{C, C, X}$ and extend the interpretation to sequents with $\interp{\vdash C_1, ..., C_n}_\Omega := \interp{C_1}_\Omega \parr ... \parr \interp{C_n}_\Omega$.
\end{defi}

\begin{rem}
\label{rem:basisinterpax}
The interpretation of an axiom under an basis of interpretation $\Omega$ is defined by $\interp{\vdash X_1, X_1^\bot}_\Omega = \interp{X_1}_\Omega \parr \interp{X_1^\bot} = \Omega(X_1, 1, X) \parr \Omega(X_1^\bot, 1, X)^\bot$.
\end{rem}

\subsection{A complete model of MLL+MIX}\label{subsec:mllmixmodel}

We prove soundness and completeness for MLL+MIX. Theorem~\ref{thm:correctness2} shows that asking for a strongly normalising union between vehicle and test corresponds to MLL+MIX correctness. This is the key ingredient in the proof of completeness. In this section, we consider the orthogonality $\perp^\mathrm{fin}$ exclusively.

Instead of the usual soundness property, we prove an extension called \emph{full soundness} \cite[Theorem 55]{seiller2012interaction} which takes cut-elimination into account. In terms of the adequacy used in realisability interpretations, proving the soundness property corresponds to showing that $\Phi : \Gamma$ implies $\interp{\vdash \Gamma}_\Omega$ for some basis of interpretation $\Omega$, except that for $\Phi : \Gamma$ we only consider constellations coming from proof-nets.

\begin{lem}
\label{lem:tensorbehaviour}
Let $A, B$ be MLL formulas, $\Gamma = C_1, ..., C_n, \Delta = D_1, ..., D_m$ be sets of MLL formulas and $\Omega$ be a basis of interpretation. We have $(\interp{\vdash\Gamma}_\Omega \parr \interp{A}_\Omega) \otimes (\interp{\vdash\Delta}_\Omega \parr \interp{B}_\Omega) \subseteq \interp{\vdash\Gamma}_\Omega \parr \interp{\vdash\Delta}_\Omega \parr \interp{A \otimes B}_\Omega$.
\end{lem}
\begin{proof}
The idea is to show $(\interp{C_1^\bot \otimes ... \otimes C_n^\bot}_\Omega \multimap \interp{A}_\Omega) \otimes (\interp{D_1^\bot \otimes ... \otimes D_m^\bot}_\Omega \multimap \interp{B}_\Omega) \subseteq (\interp{C_1^\bot \otimes ... \otimes C_n^\bot}_\Omega \otimes \interp{D_1^\bot \otimes ... \otimes D_m^\bot}_\Omega) \multimap \interp{A \otimes B}_\Omega$ which is equivalent to $(\interp{C}_\Omega \multimap \interp{A}_\Omega) \otimes (\interp{D}_\Omega \multimap \interp{B}_\Omega) \subseteq (\interp{C}_\Omega \otimes \interp{D}_\Omega) \multimap \interp{A \otimes B}_\Omega$ for $C := C_1^\bot \otimes ... \otimes C_n^\bot$ and $D := D_1^\bot \otimes ... \otimes D_m^\bot$.
Assume we have two functions $\Phi_{C,A} \in \interp{C}_\Omega \multimap \interp{A}_\Omega$ and $\Phi_{D,B} \in \interp{D}_\Omega \multimap \interp{B}_\Omega$. We can construct their disjoint union $\Phi_{C,A} \uplus \Phi_{D,B} \in (\interp{C}_\Omega \multimap \interp{A}_\Omega) \otimes (\interp{D}_\Omega \multimap \interp{B}_\Omega)$. If we provide to $\Phi_{C,A} \uplus \Phi_{D,B}$ an argument $\Phi \in \interp{C}_\Omega \otimes \interp{D}_\Omega$, then since $C$ and $D$ are disjoint, each function $\Phi_{C,A}$ and $\Phi_{D,B}$ will take their argument separately and produce $\Phi' \in \interp{A \otimes B}_\Omega$. Therefore, $\Phi_{C,A} \uplus \Phi_{D,B} \in (\interp{C}_\Omega \otimes \interp{D}_\Omega) \multimap \interp{A \otimes B}_\Omega$.
\end{proof}

\begin{lem}
\label{lem:noempty}
If $\mathbf{A}$ is a pre-behaviour then $\mathbf{A}^\bot \neq \emptyset$.
\end{lem}
\begin{proof}
Any constellation with only uncoloured rays strongly normalise with any constellation so it is always part of the orthogonal of a pre-behaviour.
\end{proof}

\begin{lem}
\label{lem:bhterminates}
If $\mathbf{A}$ is a behaviour and $\Phi \in \mathbf{A}$ then $|\exec(\Phi)| < \infty$.
\end{lem}
\begin{proof}
By definition of behaviour, we have $\mathbf{A} = \mathbf{A}^{\bot\bot}$. By Lemma~\ref{lem:noempty} there must be some $\Phi' \in \mathbf{A}^\bot$ such that $\Phi \uplus \Phi'$ is strongly normalising. Assume $\Phi$ is not strongly normalising. Then, $\Phi$ can produce infinitely many \emph{saturated} correct diagrams. Such diagrams cannot be extended with stars of $\Phi'$, hence these infinitely many saturated diagrams are preserved and $\Phi \uplus \Phi'$ cannot be strongly normalising, which is contradictory. Therefore, $\Phi$ must be strongly normalising.
\end{proof}

\begin{thm}[Full soundness for MLL+MIX]
\label{thm:soundness}
Let $\vdash \mathcal{S} : \Gamma$ be an MLL+MIX proof-net and $\Omega$ a basis of interpretation. We have $\exec(\axcutof{\mathcal{S}}) \in \interp{\vdash \Gamma}_\Omega$.
\end{thm}
\begin{proof}
We start with the case of cut-free proofs normalising into themselves. The proof is done by induction on the proof-net structure of $\mathcal{S}$.
\begin{itemize}
	\item Assume we have $\vdash \mathcal{S} : X_i, X_i^\bot$. We would like to show that $\axof{\mathcal{S}} \in \interp{X_i}_\Omega \parr \interp{X_i^\bot}_\Omega
    = \interp{X_i, X_i, X}_\Omega \parr \interp{X_i^\bot, X_i, X}_\Omega^\bot
    = \Omega(X_i, i, X) \parr \Omega(X_i^\bot, i, X)^\bot
    = \big(\Omega(X_i, i, X)^\bot \otimes \Omega(X_i^\bot, i, X)\big)^\bot$. Let $\Phi_1 \uplus \Phi_2 \in \Omega(X_i, i, X)^\bot \otimes \Omega(X_i^\bot, i, X)$ with $\Phi_1 \in \Omega(X_i, i, X)^\bot$ and $\Phi_2 \in \Omega(X_i^\bot, i, X)$. It is sufficient to show that $|\exec(\Phi_1 \uplus \Phi_2 \uplus \axof{\mathcal{S}})| < \infty$, \ie that the axiom strongly normalises with its tests. By Definition~\ref{def:basisinterp} since we have $\axof{\mathcal{S}} = [\tlocus{X_i}{X}, \tlocus{X_i^\bot}{X}]$, we have $\axof{\mathcal{S}} \uplus \Phi_2 \in \Omega(X_i, i, X)$ which is orthogonal to $\Phi_1$. It follows that $|\exec(\Phi_1 \uplus \Phi_2 \uplus \axof{\mathcal{S}})| < \infty$.
	
    \item Assume we have $\vdash \mathcal{S} : \Gamma, \Delta, A \otimes B$ coming from $\vdash \mathcal{S}_1 : \Gamma, A$ and $\vdash \mathcal{S}_2 : \Delta, B$. We have to show $\axof{\mathcal{S}} \in \interp{\vdash \Gamma}_\Omega \parr \interp{\vdash\Delta}_\Omega \parr \interp{A \otimes B}_\Omega$.
    By induction hypothesis, we have $\axof{\mathcal{S}_1} \in \interp{\vdash \Gamma, A}_\Omega = \interp{\vdash\Gamma}_\Omega \parr \interp{A}_\Omega$ and $\axof{\mathcal{S}_2} \in \interp{\vdash \Delta, B}_\Omega = \interp{\vdash\Delta}_\Omega \parr \interp{B}_\Omega$.
    By a conjugation $\mu$ such that $\axof{\mathcal{S}_1}$ and $\axof{\mathcal{S}_2}$ are made distinct, we can relocale the atoms and obtain a constellation $\Phi_\mu \in (\interp{\vdash\Gamma}_\Omega \parr \interp{A}_\Omega) \otimes (\interp{\vdash\Delta}_\Omega \parr \interp{B}_\Omega)$ such that $\Phi_\mu = \mu(\axof{\mathcal{S}_1}) \uplus \axof{\mathcal{S}_2}$. Now, by the definition of tensor for proof-structures, we have a preservation of axioms and $\axof{\mathcal{S}}$ equivalent to $\Phi_\mu$ up to conjugation (and this conjugation could be chosen for $\mu$).
    By Lemma~\ref{lem:tensorbehaviour}, we have $(\interp{\vdash\Gamma}_\Omega \parr \interp{A}_\Omega) \otimes (\interp{\vdash\Delta}_\Omega \parr \interp{B}_\Omega) \subseteq \interp{\vdash\Gamma}_\Omega \parr \interp{\vdash\Delta}_\Omega \parr \interp{A \otimes B}_\Omega$, hence $\axof{\mathcal{S}} \in \interp{\vdash \Gamma}_\Omega \parr \interp{\vdash\Delta}_\Omega \parr \interp{A \otimes B}_\Omega$.
	
    \item Assume we have $\vdash \mathcal{S} : \Gamma, A \parr B$ coming from $\vdash \mathcal{S}' : \Gamma, A, B$. We would like to show that $\axof{\mathcal{S}} \in \interp{\vdash \Gamma}_\Omega \parr \interp{A}_\Omega \parr \interp{B}_\Omega$. This directly follows from the induction hypothesis and the fact that we have $\interp{\vdash \Gamma, A, B}_\Omega = \interp{\vdash \Gamma}_\Omega \parr \interp{A}_\Omega \parr \interp{B}_\Omega$ by definition.
    
    \item Assume we have $\vdash \mathcal{S} : \Gamma, \Delta$ coming from $\vdash \mathcal{S}_1 : \Gamma$ and $\vdash \mathcal{S}_2 : \Delta$ (by using the MIX rule). We have to show $\axof{\mathcal{S}} \in \interp{\vdash\Gamma}_\Omega \parr \interp{\vdash\Delta}_\Omega$ knowing that the induction hypothesis states that $\axof{\mathcal{S}_1} \in \interp{\vdash\Gamma}_\Omega$ and $\axof{\mathcal{S}_2} \in \interp{\vdash\Delta}_\Omega$. Since the MIX rule only places two proofs next to each other, we have $\axof{\mathcal{S}} = \axof{\mathcal{S}_1} \uplus \axof{\mathcal{S}_2}$ by definition. By definition of tensor, we have $\axof{\mathcal{S}} \in \interp{\vdash\Gamma}_\Omega \otimes \interp{\vdash\Delta}_\Omega$. It remains to show that $\mathbf{A} \otimes \mathbf{B} \subseteq \mathbf{A} \parr \mathbf{B}$ in general, which would imply $\axof{\mathcal{S}} \in \interp{\vdash\Gamma}_\Omega \parr \interp{\vdash\Delta}_\Omega$. We have $\mathbf{A} \parr \mathbf{B} \subseteq (\mathbf{A}^\bot \otimes \mathbf{B}^\bot)^\bot$, hence we have to show that $\mathbf{A} \otimes \mathbf{B} \subseteq (\mathbf{A}^\bot \otimes \mathbf{B}^\bot)^\bot$. Let $\Phi_1 \uplus \Phi_2 \in \mathbf{A} \otimes \mathbf{B}$ and $\Phi_1' \uplus \Phi_2' \in \mathbf{A}^\bot \otimes \mathbf{B}^\bot$. We know that $\Phi_1 \perp \Phi_1'$ and $\Phi_2 \perp \Phi_2'$. We have $\Phi_1 \uplus \Phi_2 \perp \Phi_1' \uplus \Phi_2'$, \ie that $\Phi_1 \uplus \Phi_2 \uplus \Phi_1' \uplus \Phi_2'$ is strongly normalising. In particular, we cannot have a crossed infinite interaction between $\Phi_1$ and $\Phi_2'$ or between $\Phi_2$ and $\Phi_1'$ because otherwise one constellation would have to not be strongly normalising (because a strongly normalising constellation produces finitely many saturated diagrams which cannot be extended so to make an infinite execution) but this would contradict Lemma~\ref{lem:bhterminates}.
\end{itemize}
If the proof has cuts, then by Theorem~\ref{thm:dynamics}, we can execute its translation (a constellation) so that the normal form corresponds to the normal form of the proof. This proof is necessarily cut-free, hence the case of cut-free proofs also applies to this case.
\end{proof}

\begin{lem}
\label{lem:tests}
Let $\Omega$ be a basis of interpretation and $\vdash\Gamma$ an MLL sequent. Then, we have $\tests{\vdash\Gamma} \subseteq \interp{\vdash\Gamma}_\Omega^{\bot^\mathrm{fin}}$.
\end{lem}
\begin{proof}
Assume we have $\test{\vdash\Gamma}^\varphi \in \tests{\vdash\Gamma}$ for a switching $\varphi$ of $\vdash\Gamma$. The proof is done by induction on $\vdash\Gamma$.
\begin{itemize}
    \item If $\Gamma = \{A_1, ..., A_n\}$ where the $A_i$ are formulas $X_i$ or $X_i^\bot$, then there is a single switching $\varphi$. Because typing is invariant under execution, we can consider a simplification of tests by fusion $\exec(\test{\vdash\Gamma}^\varphi) = \sum_{i=1}^n [\tlocus[-]{A_i}{t_i},\ulocus{A_i}{X}]$ where $t_i$ is the expected encoding of the address of the atom $A_i$.
    We would like to show that $\test{\vdash\Gamma}^\varphi \in \interp{\vdash A_1, ..., A_n}_\Omega^\bot = \interp{A_1, A_1, t_1}_\Omega^\bot \otimes ... \otimes \interp{A_n, A_n, t_n}_\Omega^\bot$.
    We show that $[\tlocus[-]{A_i}{t_i},\ulocus{A_i}{X}] \in \interp{A_i, A_i, t_i}_\Omega^\bot$.
	Let $\Phi_i \in \interp{A_i, A_i, t_i}_\Omega$. Because $\interp{A_i, A_i, t_i}_\Omega$ is a behaviour, we can use Lemma~\ref{lem:bhterminates} and infer that $|\exec(\Phi_i)| < \infty$. Adding $[\tlocus[-]{A_i}{t_i},\ulocus{A_i}{X}]$ to a strongly normalising constellation cannot cause divergence, hence we must have $[\tlocus[-]{A_i}{t_i},\ulocus{A_i}{X}] \perp \Phi_i$ and $[\tlocus[-]{A_i}{t_i},\ulocus{A_i}{X}] \in \interp{A_i, A_i, t_i}_\Omega^\bot$. Now, since $\test{\vdash\Gamma}^\varphi$ is made of a disjoint union of constellations of $\interp{A_i, A_i, t_i}_\Omega^\bot$, it follows that $\test{\vdash\Gamma}^\varphi \in \interp{\vdash A_1, ..., A_n}_\Omega^\bot$.
	
    \item If $\vdash\Gamma$ is $\vdash\Delta, A\parr B$, then a switching $\varphi$ of $\vdash\Delta, A\parr B$ is a switching $\bar{\varphi}$ of $\vdash\Delta, A, B$ extended with a left or right selection of premise between $A$ and $B$, both linked by a $\parr$ connective. By the induction hypothesis, we have $\test{\vdash \Delta,A,B}^{\bar{\varphi}} \in \interp{\vdash \Delta, A, B}_\Omega = \interp{\vdash\Delta}_\Omega \parr \interp{A}_\Omega \parr \interp{B}_\Omega$ and we would like to show that $\test{\vdash \Delta,A\parr B}^\varphi \in \interp{\vdash \Delta, A \parr B}_\Omega = \interp{\vdash\Delta}_\Omega \parr \interp{A \parr B}_\Omega = \interp{\vdash\Delta}_\Omega \parr \interp{A \parr B, A, \lc\cdot X} \parr \interp{A \parr B, B, \rc\cdot X}_\Omega$.
    The constellation $\test{\vdash \Delta,A,B}^{\bar{\varphi}}$ uses terms $p_A(t)$ and $p_B(u)$ but when we add a $\parr$ link between $A$ and $B$, these terms are relocated relatively to the conclusion $A \parr B$ and we obtain $p_{A \parr B}(\lc\cdot t)$ and $p_{A \parr B}(\rc\cdot u)$. Since they only differ by a conjugation, the two tests will react in the same way with respects to strong normalisation, \ie $\test{\vdash \Delta,A,B}^{\bar{\varphi}} \in (\interp{\vdash\Delta}_\Omega^\bot \otimes \interp{A}_\Omega^\bot \otimes \interp{B}_\Omega^\bot)^\bot$ implies $\test{\vdash \Delta,A\parr B}^\varphi \in (\interp{\vdash\Delta}_\Omega^\bot \otimes \interp{A \parr B}_\Omega^\bot)^\bot$.

    \item If $\vdash \Gamma$ is $\vdash \Delta, A\otimes B$, a switching of $\vdash \Gamma$ is a switching of $\vdash \Delta, A, B$ extended to the additional $\otimes$ connective linking $A$ and $B$, and $\test{\vdash \Delta,A\otimes B}^\varphi$ can be defined from $\test{\vdash \Delta,A, B}^\varphi$ by removing the uncoloured rays $p_A(x)$ and $p_B(x)$, and adding new stars $[\qray[-]{A}, \qray[-]{B}, \qray[+]{A\otimes B}]+[\qray[-]{A\otimes B}, p_{A\otimes B}(X)]$. One can show that $\interp{\vdash \Delta,A\otimes B}_\Omega$ is generated (in the sense of bi-orthogonal closure) by a pre-behaviour $E$, \ie that $\interp{\vdash \Delta,A\otimes B}_\Omega = E^{\bot\bot}$ for some $E$, similarly to how $\mathbf{A} \otimes \mathbf{B}$ is generated by a bi-orthogonal closure on the pre-tensor $\mathbf{A} \odot \mathbf{B}$ (\cf Definition~\ref{def:pretensor}). In this pre-behaviour $E$, the rays coming from $A$ are disjoint from the rays coming from $B$ (because of the requirement of exclusion of interaction). By using the induction hypothesis $\tests{\vdash\Delta, A, B} \subseteq \interp{\vdash\Delta, A, B}_\Omega^{\bot^\mathrm{fin}}$, this shows the result since this implies that $\test{\vdash \Delta,A\otimes B}^\varphi \in E^{\bot^\mathrm{fin}}$ and $\test{\vdash \Delta,A\otimes B}^\varphi \in E^{\bot\bot\bot} = \interp{\vdash \Delta,A\otimes B}_\Omega^\bot$ since it is known that $X^\bot = X^{\bot\bot\bot}$ in general for any pre-behaviour $X$ \cite[Corollary 9]{joinet2021abstraction}.
\end{itemize}
\end{proof}

\begin{defi}[Proof-like constellation]
The syntax tree $ST(\vdash\Gamma)$ of a sequent induces a set of rays by Definition~\ref{def:addr} by computing the address of each atom in $ST(\vdash\Gamma)$. We note this set $\sharp\Gamma$. A constellation $\Phi$ is \emph{proof-like} \wrt an MLL sequent $\vdash\Gamma$ if it is made of binary stars only and $\idrays{\Phi} = \sharp\Gamma$, \ie it is a binary linking of atoms in $\Gamma$.
\end{defi}

\begin{exa}
A constellation which is proof-like \wrt $\vdash X_1^\bot \parr X_2^\bot, X_1 \otimes X_2$ is
\[[\clocus{X_1^\bot \parr X_2^\bot}{\lc\cdot X}, \clocus{X_1 \otimes X_2}{\lc\cdot X}]+[\clocus{X_1^\bot \parr X_2^\bot}{\rc\cdot X}, \clocus{X_1 \otimes X_2}{\rc\cdot X}].\]
However, even the wrong linking
\[[\clocus{X_1^\bot \parr X_2^\bot}{\lc\cdot X}, \clocus{X_1^\bot \parr X_2^\bot}{\rc\cdot X}]+[\clocus{X_1 \otimes X_2}{\lc\cdot X}, \clocus{X_1 \otimes X_2}{\rc\cdot X}]\]
is proof-like as well.
\end{exa}

\begin{thm}[Completeness for MLL+MIX]
\label{thm:completeness}
If a constellation $\Phi \in \interp{\vdash\Gamma}_\Omega$ is proof-like \wrt $\vdash\Gamma$, then there exists an MLL+MIX proof-net $\vdash \mathcal{S} : \Gamma$ such that $\Phi = \axof{\mathcal{S}}$.
\end{thm}
\begin{proof}
A proof-like constellation $\Phi \in \interp{\vdash\Gamma}_\Omega$ can always be considered as the interpretation of a proof-structure with only axioms; we can then construct a proof-structure $\mathcal{S}$ by considering the union of the latter with $ST(\vdash\Gamma)$ by placing the axioms on the right places in $ST(\vdash\Gamma)$ (at this point, the linking can still be wrong). Since $\Phi \in \interp{\vdash\Gamma}_\Omega$ we can use Lemma~\ref{lem:tests} and infer that for all switchings $\varphi$ of $\vdash\Gamma$ (equivalently, of $\mathcal{S}$), $\test{\vdash\Gamma}^\varphi = \testof{\mathcal{S}}{\varphi} \perp \Phi$, excluding ``wrong linking''. By Corollary~\ref{cor:correctness}, it follows that $\mathcal{S}$ is acyclic, \ie satisfies the correctness criterion for MLL+MIX. Therefore, $\mathcal{S}$ must be a proof-net of vehicle $\Phi$.
\end{proof} 

\subsection{A complete model of MLL}\label{subsec:mllmodel}

The soundness property actually holds for MLL with the same arguments as for MLL+MIX whether we use $\perp^1$ or $\perp^R$ as orthogonality relation. In this section, we only mean $\perp^1$ or $\perp^R$ whenever $\perp$ is written.

\begin{thm}[Full soundness for MLL]
\label{thm:soundness2}
Let $\vdash \mathcal{S} : \Gamma$ be an MLL proof-net and $\Omega$ a basis of interpretation. We have $\exec(\axcutof{\mathcal{S}}) \in \interp{\vdash \Gamma}_\Omega$.
\end{thm}
\begin{proof}
The idea of the proof is exactly the same as for Theorem~\ref{thm:soundness}. The only difference is in the axiom case. We need to show that $\exec(\Phi_1 \uplus \Phi_2 \uplus \axof{\mathcal{S}}) = \roots{\Phi_1 \uplus \Phi_2 \uplus \axof{\mathcal{S}}}$ (respectively, $|\exec(\Phi_1 \uplus \Phi_2 \uplus \axof{\mathcal{S}})| = 1$). However, the properties of the basis of interpretation ensures that $\Phi_2 \uplus \axof{\mathcal{S}}$ will be orthogonal to $\Phi_1$. Hence $\exec(\Phi_1 \uplus \Phi_2 \uplus \axof{\mathcal{S}}) = \roots{\Phi_1 \uplus \Phi_2 \uplus \axof{\mathcal{S}}}$ (respectively, $|\exec(\Phi_1 \uplus \Phi_2 \uplus \axof{\mathcal{S}})| = 1$).
\end{proof}

The proof of Lemma~\ref{lem:tests} which is essential for the completeness property does not hold anymore because of a minor technical problem. This is because a general sequent $\vdash A_1, ..., A_n$ for $A_i$ being atomic formulas is used for the base case. This is valid for MLL+MIX proof-nets since we only require acyclicity when testing with the switchings. However, this is not a correct base case for MLL proof-nets which are more demanding by requiring connectedness. We need to start from a single axiom and therefore, induction could be done on the MLL sequent calculus instead by considering provable formulas in MLL. This would be sufficient to get a completeness result. However, instead of restricting to correct formulas, which would identify $\interp{\vdash\Gamma}_\Omega$ and a subset of $\tests{\vdash\Gamma}^\bot$ corresponding to proof-structures, it is sufficient to identify $\interp{\vdash\Gamma}_\Omega$ and $\tests{\vdash\Gamma}^\bot$ directly. We would then have to prove $\tests{\vdash\Gamma} \subseteq \tests{\vdash\Gamma}^{\bot\bot}$ which is always true in general \cite[Proposition 7]{joinet2021abstraction}. We do so by considering a notion of \emph{strict interpretation}.

\begin{defi}[Strict interpretations]
We define the two strict interpretations for a given basis of interpretation $\Omega$ and an MLL sequent $\vdash\Gamma$:
\[\sinterp{\vdash\Gamma}_\Omega^1 = \tests{\vdash\Gamma}^{\bot^1} \text{ and } \sinterp{\vdash\Gamma}_\Omega^R = \tests{\vdash\Gamma}^{\bot^R}.\]
\end{defi}

\begin{thm}[Completeness for MLL]
\label{thm:completeness2}
If a constellation $\Phi \in \sinterp{\vdash\Gamma}_\Omega^R$ (respectively $\Phi \in \sinterp{\vdash\Gamma}_\Omega^1$) is proof-like \wrt a provable sequent $\vdash\Gamma$ of MLL, then there exists an MLL proof-net $\vdash \mathcal{S} : \Gamma$ such that $\Phi = \axof{\mathcal{S}}$.
\end{thm}
\begin{proof}
The proof begins like the proof of completeness for MLL+MIX (\cf Theorem~\ref{thm:completeness}) and reach the construction of a proof-structure with axioms translated into $\Phi$. Now, $\Phi \in \sinterp{\vdash\Gamma}_\Omega^R$ (respectively $\Phi \in \sinterp{\vdash\Gamma}_\Omega^1$) implies that, in particular, $\Phi$ passes the Danos-Regnier correctness test for MLL (by Corollary~\ref{cor:correctness}). Therefore, the proof-structure we constructed must be correct.
\end{proof}

\begin{obs}
\label{obs:testsproofdual}
Notice that if we have a constellation $\Phi \in \sinterp{\vdash\Gamma}_\Omega^X$ for some $\Omega$, MLL sequent $\vdash\Gamma$ and $X \in \{1, R\}$, its Danos-Regnier tests $\Phi_1, ..., \Phi_n$ are constellations of $(\sinterp{\vdash\Gamma}_\Omega^X)^\bot$. This formalises the intuition in proof-nets that tests are sort of proofs of the dual.
\end{obs}

\subsection{Adequacy}\label{subsec:adequacy}

Girard's adequacy property \cite[Section 4.4]{transyn1} is a way to relate type labels/formulas of \Cref{subsec:testing} and behaviours of \Cref{subsec:behaviours}. In realisability interpretations \cite{rieg2014forcing, miquel2009formalisation}, this relation usually take the form of an \emph{adequacy lemma} showing that type labels guarantee membership in some behaviour. The idea is that type labels have the same role as program specification and what we usually want is that passing some tests for a specification ensures that the program has the expected behaviour.

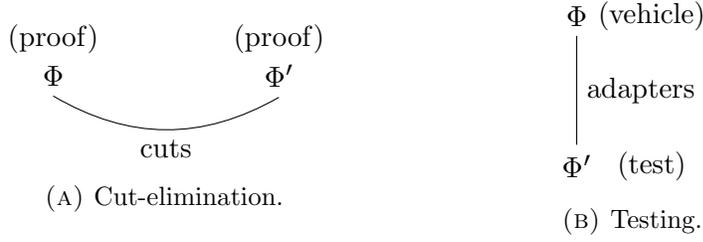
\begin{figure}
    \centering
    \begin{subfigure}{0.4\textwidth}
        \centering
        \begin{tikzpicture}
            \node at (0, 0) (p1) {$\Phi$};
            \node at (0, 0.5) (l1) {(proof)};

            \node at (3, 0) (p2) {$\Phi'$};
            \node at (3, 0.5) (l2) {(proof)};

            \draw (p1.south) edge[bend right] node[below] {cuts} (p2.south);
        \end{tikzpicture}
        \caption{Cut-elimination.}
    \end{subfigure}
    \begin{subfigure}{0.4\textwidth}
        \centering
        \begin{tikzpicture}
            \node at (0, 0) (p1) {$\Phi$};
            \node at (1, 0) (l1) {(vehicle)};
    
            \node at (0, -2) (p2) {$\Phi'$};
            \node at (1, -2) (l2) {(test)};
    
            \draw (p1) -- node[right] {adapters} (p2);
        \end{tikzpicture}
        \caption{Testing.}
    \end{subfigure}
    \caption{Cut-elimination seen as testing. The difference is only the point of view. }
    \label{fig:proofpov}
\end{figure}

This adequacy actually corresponds to a cut-elimination theorem. This is because connecting two constellations with cuts can be seen as making two constellations interact with adapters, \ie as testing a constellation against another constellation (\cf\Cref{fig:proofpov}).

It follows that the cut-elimination theorem (the fact we can eliminates all occurrences of cuts) states all possible interactions between our objects are sound. This is not always the case in general: a cut on an axiom with proof-structures is ill-behaving. But if our objects have the \emph{correct shaope} (passes the right tests) then all interaction must be sound.

Adequacy is therefore a direct consequence of full soundness (\cf Theorem~\ref{thm:soundness2}) since:
\begin{itemize}
    \item we are able to simulate cut-elimination for proof-structures (\cf Theorem~\ref{lem:cutelim});
    \item we can simulate proof-nets with constellations and the cut-elimination is known to hold for proof-nets.
\end{itemize}

\section{Perspectives and Future works}

\paragraph{Alternative definitions of execution} In the simulation of Turing machines \Cref{subsec:turing}, we simulate runs by generating all the linear saturated and correct diagrams, which is an implicit reference to actual infinite. However, this is not a natural way to compute. It is actually possible to follow the usual computation of automata by traversing a dependency graph $\dgraph{M^\bigstar + w^\bigstar}$ itself, seen as the state graph of a Turing machine reading an input word $w$. We then have to handle a unification problem when during the traversal. In case of error, the run is cancelled. In case of non-deterministic choice, we have parallel runs and only runs without unification errors survive. Although we do not explore this idea in this paper, any constellation seems to define a sort of generalised non-deterministic hypergraph automata with hyperedge transitions and graph-like runs where the fusion of stars triggers a propagation of information. This generalises various classes of automata (pushdown, alternating, Turing machines etc) and provide a machine-like execution of constellations. This would also generalise token machines of the geometry of interaction similarly to some existing works of the literature \cite{dal2017geometry}.

\paragraph{Categorical model of constellations} To ensure that we indeed have a model of MLL, it is possible to show that behaviours as categorical objects (providing we choose an orthogonality relation capturing MLL provability) and linear implications as arrows define a $\ast$-autonomous category by following Seiller's categorical model of interaction graphs \cite[Section 3]{seiller2012interaction}. In particular, re-addressing of constellations has to be correctly treated in order to make composition of arrows possible. The associativity of execution (Theorem~\ref{thm:assocexec}) is essential for the associativity of the composition and the property of adjunction (Corollary~\ref{cor:adjunction}) ensures a monoidal closure. As for categorical interpretations of GoI with monoidal traced categories \cite{haghverdi2006categorical, abramsky2002geometry}, the execution of constellation should define a trace. 

\paragraph{Extensions of other fragments of linear logic} We defined a model for both MLL and MLL+MIX but our (re)construction can be extended to other fragments of linear logic. In the first paper of Transcendental Syntax \cite[Section 5]{transyn1}, a reconstruction of intuitionistic exponentials for linear logic is sketched with a new correctness criterion.

The transcendental syntax also claims great improvements going beyond linear logic by suggesting in particular a new computational interpretation for second order logic and predicate calculus \cite[Section 5]{transyn3} (which is seen as part of second order logic) but also Peano arithmetic \cite[Section 5]{transyn4}. Second order logic uses more powerful constellations using internal colours within rays (\eg $+c(+d(x))$) which adds a more complex combinatorics to constellations. This defines two classes of constellations: the ones of this paper, called \emph{objective}, containing no internal colours and the ones with internal colours, called \emph{subjective}. These two classes of constellations will allow for the definition of a non-empty behaviour $\mathbf{0}$ \cite[Section 4.3]{transyn3} from which we can establish non-classical notions of truth \cite[Section 3]{transyn4} but also state coherence, \ie that $\mathbf{0}$ has no correct constellations.

\paragraph{Applications to implicit computational complexity} Several authors \cite{aubert2016logarithmic, aubert2016unary, baillot2001elementary} used flows (which can be seen as binary stars) for an implicit complexity analysis using encodings of automata. Since the stellar resolution is designed to be an extension of this model, we can expect to capture other complexity classes. It would then be possible to extend Seiller's idea of computational complexity \cite{seiller2018interaction} where the concepts of orthogonality and test are central.

Moreover, a reconstruction of predicate calculus may provide a new understanding to descriptive complexity \cite{immerman2012descriptive} results, such as Immerman-Vardi theorem \cite{vardi1982complexity,immermann1986relational} where predicate calculus is essential in the definition of complexity classes such as P and NP. The idea is that if a complexity class $C$ is captured by a logic, then in our framework, formulas would correspond to set of programs bounded by a certain complexity and correctness criteria could also be used to check the complexity of a program (constellation). Formulas would then represent specification certifying complexity for a constellation seen as a program.

\paragraph{Formulas as specification for a computational behaviour} In model checking, the representation of a system (usually an automata or a transition system) is verified automatically in order to check if it satisfies a formula (typically in the temporal logic LTL \cite{pnueli1977temporal}). Since the stellar resolution can provide natural encoding of state machines or labelled transition systems, and that it is possible to design logics and formulas with the transcendental syntax (by using both interactive typing and tests of type labels), we can imagine extensions of model checking to other models of computation or to other (existing or designed) logics capturing more fine-grained properties. The correctness criterion of linear logic can itself be seen as a way to certify a computational object as a proof of some formula $A$ corresponding to a specification. Since the $\lambda$-calculus can be nicely encoded with proof-nets, higher order model checking for $\lambda$-calculus may be investigated as well. 

\paragraph{Towards a ``Logic of Interaction?"} Inspired by computation and linear logic, authors such as Curien \cite{curien2003symmetry} and Abramsky \cite[Section 5]{abramsky2002geometry} exposed a paradigm of interaction where the notion of interaction would be central in logic before anything else. Similarly to how complex behaviours arise from a system of interacting agents as in biology or more generally, the theory of complex systems (see chemical reaction networks for instance \cite{jost2019hypergraph}), can a notion of logic emerge from any system of interactions?

Our stellar resolution is indeed an instance of a system of interacting agents (stars) from which emerge complex concepts (proofs and formulas). But even beyond logic, constellations are able to represent automata and other models of computation as interacting agents transmitting information in a graph-like structure and, by interactive typing, it is possible to design formulas describing their computational behaviour. Logic then appears as a way to describe computation or put constraint on it. This opens the materialistic idea of a ``logic of things'' analysing computational interactions.

\bibliographystyle{plainurl}
\bibliography{references}

\begin{thebibliography}{100}

\bibitem{abramsky2008information}
Samson Abramsky.
\newblock Information, processes and games.
\newblock {\em J. Benthem van \& P. Adriaans (Eds.), Philosophy of
  Information}, pages 483--549, 2008.

\bibitem{abramsky2002geometry}
Samson Abramsky, Esfandiar Haghverdi, and Philip Scott.
\newblock Geometry of interaction and linear combinatory algebras.
\newblock {\em Mathematical Structures in Computer Science}, 12(5):625--665,
  2002.

\bibitem{abramsky1994games}
Samson Abramsky and Radha Jagadeesan.
\newblock Games and full completeness for multiplicative linear logic.
\newblock {\em The Journal of Symbolic Logic}, 59(2):543--574, 1994.

\bibitem{abramsky2000full}
Samson Abramsky, Radha Jagadeesan, and Pasquale Malacaria.
\newblock Full abstraction for pcf.
\newblock {\em Information and Computation}, 163(2):409--470, 2000.

\bibitem{accattoli2018proof}
Beniamino Accattoli.
\newblock Proof nets and the linear substitution calculus.
\newblock In {\em International Colloquium on Theoretical Aspects of
  Computing}, pages 37--61. Springer, 2018.

\bibitem{amadio2016operational}
Roberto Amadio.
\newblock {\em Operational methods in semantics}.
\newblock PhD thesis, Univerist{\'e} Denis Diderot Paris 7, 2016.

\bibitem{asperti1995paths}
Andrea Asperti and Cosimo Laneve.
\newblock Paths, computations and labels in the $\lambda$-calculus.
\newblock {\em Theoretical Computer Science}, 142(2):277--297, 1995.

\bibitem{aubert2014unification}
Cl{\'e}ment Aubert and Marc Bagnol.
\newblock Unification and logarithmic space.
\newblock In {\em Rewriting and Typed Lambda Calculi}, pages 77--92. Springer,
  2014.

\bibitem{aubert2016unary}
Cl{\'e}ment Aubert, Marc Bagnol, and Thomas Seiller.
\newblock Unary resolution: Characterizing ptime.
\newblock In {\em International Conference on Foundations of Software Science
  and Computation Structures}, pages 373--389. Springer, 2016.

\bibitem{aubert2016characterizing}
Cl{\'e}ment Aubert and Thomas Seiller.
\newblock Characterizing co-nl by a group action.
\newblock {\em Mathematical Structures in Computer Science}, 26(4):606--638,
  2016.

\bibitem{aubert2016logarithmic}
Cl{\'e}ment Aubert and Thomas Seiller.
\newblock Logarithmic space and permutations.
\newblock {\em Information and Computation}, 248:2--21, 2016.

\bibitem{baader1999term}
Franz Baader and Tobias Nipkow.
\newblock {\em Term rewriting and all that}.
\newblock Cambridge university press, 1999.

\bibitem{bagnol2015dependencies}
Marc Bagnol, Amina Doumane, and Alexis Saurin.
\newblock On the dependencies of logical rules.
\newblock In {\em International Conference on Foundations of Software Science
  and Computation Structures}, pages 436--450. Springer, 2015.

\bibitem{baier2008principles}
Christel Baier and Joost-Pieter Katoen.
\newblock {\em Principles of model checking}.
\newblock MIT press, 2008.

\bibitem{baillot2001elementary}
Patrick Baillot and Marco Pedicini.
\newblock Elementary complexity and geometry of interaction.
\newblock {\em Fundamenta Informaticae}, 45(1-2):1--31, 2001.

\bibitem{basaldella2009meaning}
Michele Basaldella and Kazushige Terui.
\newblock On the meaning of logical completeness.
\newblock In {\em International Conference on Typed Lambda Calculi and
  Applications}, pages 50--64. Springer, 2009.

\bibitem{blass1992game}
Andreas Blass.
\newblock A game semantics for linear logic.
\newblock {\em Annals of Pure and Applied logic}, 56(1-3):183--220, 1992.

\bibitem{bretto2013hypergraph}
Alain Bretto.
\newblock Hypergraph theory.
\newblock {\em An introduction. Mathematical Engineering. Cham: Springer},
  2013.

\bibitem{chen1991synthesis}
Junghuei Chen and Nadrian~C Seeman.
\newblock Synthesis from dna of a molecule with the connectivity of a cube.
\newblock {\em Nature}, 350(6319):631--633, 1991.

\bibitem{church1940formulation}
Alonzo Church.
\newblock A formulation of the simple theory of types.
\newblock {\em The journal of symbolic logic}, 5(2):56--68, 1940.

\bibitem{church1940}
Alonzo Church.
\newblock On the concept of a random sequence.
\newblock {\em Bulletin of the American Mathematical Society}, 46(2):130--135,
  1940.

\bibitem{curien2003symmetry}
P-L Curien.
\newblock Symmetry and interactivity in programming.
\newblock {\em Bulletin of Symbolic Logic}, pages 169--180, 2003.

\bibitem{curry}
Haskell~B Curry.
\newblock Functionality in combinatory logic.
\newblock {\em Proceedings of the National Academy of Sciences of the United
  States of America}, 20(11):584, 1934.

\bibitem{dal2017geometry}
Ugo Dal~Lago, Ryo Tanaka, and Akira Yoshimizu.
\newblock The geometry of concurrent interaction: Handling multiple ports by
  way of multiple tokens.
\newblock In {\em 2017 32nd Annual ACM/IEEE Symposium on Logic in Computer
  Science (LICS)}, pages 1--12. IEEE, 2017.

\bibitem{danosPHD}
Vincent Danos.
\newblock {\em La Logique Lin{\'e}aire appliqu{\'e}e {\`a} l'{\'e}tude de
  divers processus de normalisation (principalement du Lambda-calcul)}.
\newblock PhD thesis, Paris 7, 1990.

\bibitem{danos1989structure}
Vincent Danos and Laurent Regnier.
\newblock The structure of multiplicatives.
\newblock {\em Archive for Mathematical logic}, 28(3):181--203, 1989.

\bibitem{danos1995proof}
Vincent Danos and Laurent Regnier.
\newblock Proof-nets and the hilbert space.
\newblock {\em London Mathematical Society Lecture Note Series}, pages
  307--328, 1995.

\bibitem{danos1999reversible}
Vincent Danos and Laurent Regnier.
\newblock Reversible, irreversible and optimal $\lambda$-machines.
\newblock {\em Theoretical Computer Science}, 227(1-2):79--97, 1999.

\bibitem{dantsin2001complexity}
Evgeny Dantsin, Thomas Eiter, Georg Gottlob, and Andrei Voronkov.
\newblock Complexity and expressive power of logic programming.
\newblock {\em ACM Computing Surveys (CSUR)}, 33(3):374--425, 2001.

\bibitem{davis1958computability}
Martin Davis.
\newblock Computability and unsolvability. 1982 ed, 1958.

\bibitem{de2011correctness}
Paulin~Jacob{\'e} De~Naurois and Virgile Mogbil.
\newblock Correctness of linear logic proof structures is nl-complete.
\newblock {\em Theoretical Computer Science}, 412(20):1941--1957, 2011.

\bibitem{di2003proof}
Roberto Di~Cosmo, Delia Kesner, and Emmanuel Polonovski.
\newblock Proof nets and explicit substitutions.
\newblock {\em Mathematical Structures in Computer Science}, 13(3):409, 2003.

\bibitem{dijkstra2001solution}
Edsger~W Dijkstra.
\newblock Solution of a problem in concurrent programming control.
\newblock In {\em Pioneers and Their Contributions to Software Engineering},
  pages 289--294. Springer, 2001.

\bibitem{dummett1991logical}
Michael Dummett.
\newblock {\em The logical basis of metaphysics}.
\newblock Harvard university press, 1991.

\bibitem{eisinger1991deduction}
Norbert Eisinger and Hans~J{\"u}rgen Ohlbach.
\newblock Deduction systems based on resolution.
\newblock 1991.

\bibitem{eiter2009answer}
Thomas Eiter, Giovambattista Ianni, and Thomas Krennwallner.
\newblock Answer set programming: A primer.
\newblock In {\em Reasoning Web International Summer School}, pages 40--110.
  Springer, 2009.

\bibitem{fleury1994mix}
Arnaud Fleury and Christian Retor{\'e}.
\newblock The mix rule.
\newblock {\em Mathematical Structures in Computer Science}, 4(2):273--285,
  1994.

\bibitem{gelfond2008answer}
Michael Gelfond.
\newblock Answer sets.
\newblock {\em Foundations of Artificial Intelligence}, 3:285--316, 2008.

\bibitem{gentzen1}
Gerhard Gentzen.
\newblock Untersuchungen {\"u}ber das logische schlie{\ss}en. i.
\newblock {\em Mathematische zeitschrift}, 39(1):176--210, 1935.

\bibitem{gentzen2}
Gerhard Gentzen.
\newblock Untersuchungen {\"u}ber das logische schlie{\ss}en. ii.
\newblock {\em Mathematische Zeitschrift}, 39(1):405--431, 1935.

\bibitem{linearlogic}
Jean-Yves Girard.
\newblock Linear logic.
\newblock {\em Theoretical computer science}, 50(1):1--101, 1987.

\bibitem{goi2}
Jean-Yves Girard.
\newblock Geometry of interaction {II}: deadlock-free algorithms.
\newblock In {\em International Conference on Computer Logic}, pages 76--93.
  Springer, 1988.

\bibitem{girard1988normal}
Jean-Yves Girard.
\newblock Normal functors, power series and $\lambda$-calculus.
\newblock {\em Annals of pure and applied logic}, 37(2):129--177, 1988.

\bibitem{goi1}
Jean-Yves Girard.
\newblock Geometry of interaction {I}: interpretation of system f.
\newblock In {\em Studies in Logic and the Foundations of Mathematics}, volume
  127, pages 221--260. Elsevier, 1989.

\bibitem{goi0}
Jean-Yves Girard.
\newblock Towards a geometry of interaction.
\newblock {\em Contemporary Mathematics}, 92(69-108):6, 1989.

\bibitem{goi3}
Jean-Yves Girard.
\newblock Geometry of interaction {III}: accommodating the additives.
\newblock {\em London Mathematical Society Lecture Note Series}, pages
  329--389, 1995.

\bibitem{girard1996proof}
Jean-Yves Girard.
\newblock Proof-nets: the parallel syntax for proof-theory.
\newblock {\em Lecture Notes in Pure and Applied Mathematics}, pages 97--124,
  1996.

\bibitem{girard2001locus}
Jean-Yves Girard.
\newblock Locus solum: From the rules of logic to the logic of rules.
\newblock {\em Mathematical structures in computer science}, 11(3):301, 2001.

\bibitem{goi4}
Jean-Yves Girard.
\newblock Geometry of interaction {IV}: the feedback equation.
\newblock In {\em Logic Colloquium}, volume~3, pages 76--117, 2006.

\bibitem{goi5}
Jean-Yves Girard.
\newblock Geometry of interaction {V}: logic in the hyperfinite factor.
\newblock {\em Theoretical Computer Science}, 412(20):1860--1883, 2011.

\bibitem{goi6}
Jean-Yves Girard.
\newblock Geometry of interaction {VI}: a blueprint for transcendental syntax.
\newblock {\em preprint}, 2013.

\bibitem{girard2013three}
Jean-Yves Girard.
\newblock Three lightings of logic (invited talk).
\newblock In {\em Computer Science Logic 2013 (CSL 2013)}. Schloss
  Dagstuhl-Leibniz-Zentrum fuer Informatik, 2013.

\bibitem{transyn2}
Jean-Yves Girard.
\newblock Transcendental syntax {II}: non-deterministic case.
\newblock 2016.

\bibitem{transyn3}
Jean-Yves Girard.
\newblock Transcendental syntax {III}: equality.
\newblock 2016.

\bibitem{transyn1}
Jean-Yves Girard.
\newblock Transcendental syntax {I}: deterministic case.
\newblock {\em Mathematical Structures in Computer Science}, 27(5):827--849,
  2017.

\bibitem{girard2018logique}
Jean-Yves Girard.
\newblock La logique 2.0.
\newblock 2018.

\bibitem{transyn4}
Jean-Yves Girard.
\newblock Transcendental syntax {IV}: logic without systems.
\newblock 2020.

\bibitem{haghverdi2006categorical}
Esfandiar Haghverdi and Philip Scott.
\newblock A categorical model for the geometry of interaction.
\newblock {\em Theoretical Computer Science}, 350(2-3):252--274, 2006.

\bibitem{hedman2004first}
Shawn Hedman.
\newblock {\em A First Course in Logic: An introduction to model theory, proof
  theory, computability, and complexity}.
\newblock OUP Oxford, 2004.

\bibitem{herbrand1930recherches}
Jacques Herbrand.
\newblock {\em Recherches sur la th{\'e}orie de la d{\'e}monstration}.
\newblock PhD thesis, Universit{\'e} de Paris, 1930.

\bibitem{horn1951sentences}
Alfred Horn.
\newblock On sentences which are true of direct unions of algebras.
\newblock {\em The Journal of Symbolic Logic}, 16(1):14--21, 1951.

\bibitem{howard}
William~A Howard.
\newblock The formulae-as-types notion of construction.
\newblock {\em To HB Curry: essays on combinatory logic, lambda calculus and
  formalism}, 44:479--490, 1980.

\bibitem{huth2004logic}
Michael Huth and Mark Ryan.
\newblock {\em Logic in Computer Science: Modelling and reasoning about
  systems}.
\newblock Cambridge university press, 2004.

\bibitem{hyland2000full}
J~Martin~E Hyland and C-HL Ong.
\newblock On full abstraction for pcf: I, ii, and iii.
\newblock {\em Information and computation}, 163(2), 2000.

\bibitem{immermann1986relational}
Neil Immerman.
\newblock Relational queries computable in polynomial time.
\newblock {\em Information and Control}, 68(1):86 -- 104, 1986.
\newblock \href {https://doi.org/https://doi.org/10.1016/S0019-9958(86)80029-8}
  {\path{doi:https://doi.org/10.1016/S0019-9958(86)80029-8}}.

\bibitem{immerman2012descriptive}
Neil Immerman.
\newblock {\em Descriptive complexity}.
\newblock Springer Science \& Business Media, 2012.

\bibitem{joinet2021abstraction}
Jean-Baptiste Joinet and Thomas Seiller.
\newblock From abstraction and indiscernibility to classification and types:
  revisiting hermann weyl’s theory of ideal elements.
\newblock {\em Kagaku tetsugaku}, 53(2):65--93, 2021.

\bibitem{jonoska2005computational}
Nata{\v{s}}a Jonoska and Gregory~L McColm.
\newblock A computational model for self-assembling flexible tiles.
\newblock In {\em International Conference on Unconventional Computation},
  pages 142--156. Springer, 2005.

\bibitem{jonoska2006flexible}
Nata{\v{s}}a Jonoska and Gregory~L McColm.
\newblock Flexible versus rigid tile assembly.
\newblock In {\em International Conference on Unconventional Computation},
  pages 139--151. Springer, 2006.

\bibitem{jost2019hypergraph}
J{\"u}rgen Jost and Raffaella Mulas.
\newblock Hypergraph laplace operators for chemical reaction networks.
\newblock {\em Advances in mathematics}, 351:870--896, 2019.

\bibitem{kowalski1974predicate}
Robert Kowalski.
\newblock Predicate logic as programming language.
\newblock In {\em IFIP congress}, volume~74, pages 569--544, 1974.

\bibitem{kowalski1975proof}
Robert Kowalski.
\newblock A proof procedure using connection graphs.
\newblock {\em Journal of the ACM (JACM)}, 22(4):572--595, 1975.

\bibitem{kowalski1971linear}
Robert Kowalski and Donald Kuehner.
\newblock Linear resolution with selection function.
\newblock {\em Artificial Intelligence}, 2(3-4):227--260, 1971.

\bibitem{krivine2009interactive}
JL~Krivine, PL~Curien, H~Herbelin, and PA~Melli{\`e}s.
\newblock Interactive models of computation and program behavior.
\newblock 2009.

\bibitem{lafont1995proof}
Yves Lafont.
\newblock From proof nets to interaction nets.
\newblock {\em London Mathematical Society Lecture Note Series}, pages
  225--248, 1995.

\bibitem{lassez1988unification}
J-L Lassez, Michael~J Maher, and Kim Marriott.
\newblock Unification revisited.
\newblock In {\em Foundations of logic and functional programming}, pages
  67--113. Springer, 1988.

\bibitem{lathrop2009strict}
James~I Lathrop, Jack~H Lutz, and Scott~M Summers.
\newblock Strict self-assembly of discrete sierpinski triangles.
\newblock {\em Theoretical Computer Science}, 410(4-5):384--405, 2009.

\bibitem{leitsch2012resolution}
Alexander Leitsch.
\newblock {\em The resolution calculus}.
\newblock Springer Science \& Business Media, 2012.

\bibitem{lobo1991semantics}
Jorge Lobo, Arcot Rajasekar, and Jack Minker.
\newblock Semantics of horn and disjunctive logic programs.
\newblock {\em Theoretical Computer Science}, 86(1):93--106, 1991.

\bibitem{unifalgo}
Alberto Martelli and Ugo Montanari.
\newblock An efficient unification algorithm.
\newblock {\em ACM Transactions on Programming Languages and Systems (TOPLAS)},
  4(2):258--282, 1982.

\bibitem{mazurkiewicz1988basic}
Antoni Mazurkiewicz.
\newblock Basic notions of trace theory.
\newblock In {\em Workshop/School/Symposium of the REX Project (Research and
  Education in Concurrent Systems)}, pages 285--363. Springer, 1988.

\bibitem{mellies2009categorical}
Paul-Andr{\'e} Mellies.
\newblock Categorical semantics of linear logic.
\newblock {\em Panoramas et syntheses}, 27:15--215, 2009.

\bibitem{minker1994overview}
Jack Minker.
\newblock Overview of disjunctive logic programming.
\newblock {\em Annals of Mathematics and Artificial Intelligence},
  12(1-2):1--24, 1994.

\bibitem{miquel2009formalisation}
Alexandre Miquel.
\newblock De la formalisation des preuves {\`a} l’extraction de programmes.
\newblock {\em HdR thesis, Universit{\'e} Paris}, 7, 2009.

\bibitem{miquey2017classical}
{\'E}tienne Miquey.
\newblock {\em Classical realizability and side-effects}.
\newblock PhD thesis, Universit{\'e} Sorbonne Paris Cit{\'e}-Universit{\'e}
  Paris Diderot (Paris 7~…, 2017.

\bibitem{murawski2000dominator}
Andrzej~S Murawski and C-HL Ong.
\newblock Dominator trees and fast verification of proof nets.
\newblock In {\em Proceedings Fifteenth Annual IEEE Symposium on Logic in
  Computer Science (Cat. No. 99CB36332)}, pages 181--191. IEEE, 2000.

\bibitem{patitz2014introduction}
Matthew~J Patitz.
\newblock An introduction to tile-based self-assembly and a survey of recent
  results.
\newblock {\em Natural Computing}, 13(2):195--224, 2014.

\bibitem{pnueli1977temporal}
Amir Pnueli.
\newblock The temporal logic of programs.
\newblock In {\em 18th Annual Symposium on Foundations of Computer Science
  (sfcs 1977)}, pages 46--57. ieee, 1977.

\bibitem{regnierPHD}
Laurent Regnier.
\newblock {\em Lambda-calcul et r{\'e}seaux}.
\newblock PhD thesis, Paris 7, 1992.

\bibitem{retore2003handsome}
Christian Retor{\'e}.
\newblock Handsome proof-nets: perfect matchings and cographs.
\newblock {\em Theoretical Computer Science}, 294(3):473--488, 2003.

\bibitem{riba2007strong}
Colin Riba.
\newblock Strong normalization as safe interaction.
\newblock In {\em 22nd Annual IEEE Symposium on Logic in Computer Science (LICS
  2007)}, pages 13--22. IEEE, 2007.

\bibitem{rieg2014forcing}
Lionel Rieg.
\newblock {\em On forcing and classical realizability}.
\newblock PhD thesis, Ecole normale sup{\'e}rieure de lyon-ENS LYON, 2014.

\bibitem{robinson1965machine}
John~Alan Robinson et~al.
\newblock A machine-oriented logic based on the resolution principle.
\newblock {\em Journal of the ACM}, 12(1):23--41, 1965.

\bibitem{schalk2004categorical}
Andrea Schalk.
\newblock What is a categorical model of linear logic.
\newblock {\em Manuscript, available from http://www. cs. man. ac. uk/~
  schalk/work. html}, 2004.

\bibitem{schutte1956system}
Kurt Sch{\"u}tte.
\newblock Ein system des verkn{\"u}pfenden schliessens.
\newblock {\em Archiv f{\"u}r mathematische Logik und Grundlagenforschung},
  2(2):55--67, 1956.

\bibitem{seely1987linear}
Robert~AG Seely.
\newblock {\em Linear logic,*-autonomous categories and cofree coalgebras}.
\newblock Ste. Anne de Bellevue, Quebec: CEGEP John Abbott College, 1987.

\bibitem{seeman}
Nadrian~C Seeman.
\newblock Nucleic acid junctions and lattices.
\newblock {\em Journal of theoretical biology}, 99(2):237--247, 1982.

\bibitem{seiller2012interaction}
Thomas Seiller.
\newblock Interaction graphs: multiplicatives.
\newblock {\em Annals of Pure and Applied Logic}, 163(12):1808--1837, 2012.

\bibitem{seiller2016interaction}
Thomas Seiller.
\newblock Interaction graphs: additives.
\newblock {\em Annals of Pure and Applied Logic}, 167(2):95--154, 2016.

\bibitem{seiller2016interaction2}
Thomas Seiller.
\newblock Interaction graphs: Full linear logic.
\newblock In {\em 2016 31st Annual ACM/IEEE Symposium on Logic in Computer
  Science (LICS)}, pages 1--10. IEEE, 2016.

\bibitem{seiller2017interaction}
Thomas Seiller.
\newblock Interaction graphs: Graphings.
\newblock {\em Annals of Pure and Applied Logic}, 168(2):278--320, 2017.

\bibitem{seiller2018interaction}
Thomas Seiller.
\newblock Interaction graphs: Non-deterministic automata.
\newblock {\em ACM Transactions on Computational Logic (TOCL)}, 19(3):1--24,
  2018.

\bibitem{seiller2019interaction}
Thomas Seiller.
\newblock Interaction graphs: Exponentials.
\newblock {\em Logical Methods in Computer Science}, 15, 2019.

\bibitem{seiller2020probabilistic}
Thomas Seiller.
\newblock Probabilistic complexity classes through semantics.
\newblock {\em arXiv preprint arXiv:2002.00009}, 2020.

\bibitem{sickel1976search}
Sharon Sickel.
\newblock A search technique for clause interconnectivity graphs.
\newblock {\em IEEE Transactions on Computers}, (8):823--835, 1976.

\bibitem{sipser1996introduction}
Michael Sipser.
\newblock Introduction to the theory of computation.
\newblock {\em ACM Sigact News}, 27(1):27--29, 1996.

\bibitem{tarnlund1977horn}
Sten-{\AA}ke T{\"a}rnlund.
\newblock Horn clause computability.
\newblock {\em BIT Numerical Mathematics}, 17(2):215--226, 1977.

\bibitem{thomas1991logics}
Wolfgang Thomas.
\newblock On logics, tilings, and automata.
\newblock In {\em International Colloquium on Automata, Languages, and
  Programming}, pages 441--454. Springer, 1991.

\bibitem{vardi1982complexity}
Moshe~Y Vardi.
\newblock The complexity of relational query languages.
\newblock In {\em Proceedings of the fourteenth annual ACM symposium on Theory
  of computing}, pages 137--146, 1982.

\bibitem{wang}
Hao Wang.
\newblock Proving theorems by pattern recognition —{II}.
\newblock {\em Bell system technical journal}, 40(1):1--41, 1961.

\bibitem{winfree1998algorithmic}
Erik Winfree.
\newblock {\em Algorithmic self-assembly of DNA}.
\newblock PhD thesis, Citeseer, 1998.

\bibitem{woods2015intrinsic}
Damien Woods.
\newblock Intrinsic universality and the computational power of self-assembly.
\newblock {\em Philosophical Transactions of the Royal Society A: Mathematical,
  Physical and Engineering Sciences}, 373(2046):20140214, 2015.

\bibitem{woods2019diverse}
Damien Woods, David Doty, Cameron Myhrvold, Joy Hui, Felix Zhou, Peng Yin, and
  Erik Winfree.
\newblock Diverse and robust molecular algorithms using reprogrammable dna
  self-assembly.
\newblock {\em Nature}, 567(7748):366--372, 2019.
\newblock \href {https://doi.org/10.1038/s41586-019-1014-9}
  {\path{doi:10.1038/s41586-019-1014-9}}.

\end{thebibliography}

\appendix

\section{Term unification and first-order resolution}\label{sec:unification}

We recall elementary definitions of term unification \cite{herbrand1930recherches} and first-order resolution \cite{robinson1965machine}. We refer the reader to the article of Lassez et al. \cite{lassez1988unification} for more details which are often omitted in the literature or Baader et al. \cite{baader1999term} for a broader view.

We will use uppercase letters such as $X, Y, Z$ for variables and lowercase letters $a, b, c, f, g$ and $h$ for function symbols.

\begin{itemize}[label=$\triangleright$]
    \item A \emph{signature} $\mathbb{S} = (V, F, \ar)$ consists of a countable set $V$ of variables, a countable set $F$ of function symbols whose arities are given by $\ar : F \rightarrow \nat$. We set a signature for this section.

    \item The set of \emph{terms} $\terms{\mathbb{S}}$ is inductively defined by the following grammar:
    \begin{equation*}\tag{Terms}
    t, u ::= X \mid f(t_1, \dots, t_n)
    \qquad X \in V, f \in F, \ar(f) = n
    \end{equation*}

    \item A \emph{substitution} is a function $\theta: V \rightarrow \terms{\mathbb{S}}$. Substitutions are extended from variables to terms by $\theta(f(u_1, ..., u_k)) = f(\theta(u_1), ..., \theta(u_k))$. The substitution $\theta t$ is often written $\theta(t)$ or explictly as a set of associations $\{X_1 \mapsto t_1, ..., X_n \mapsto t_n\}$ (often written $\{x:=t\}$ when there is only one association).

    From two substitutions $\theta_1, \theta_2$, we can construct their composition $\theta_1 \circ \theta_2$ such that $(\theta_1 \circ \theta_2)t = \theta_1 (\theta_2 t)$. The composition is associative \cite[Corollary 6]{lassez1988unification}.

    \item A \emph{renaming} is a substitution $\alpha$ such that $\alpha(X) \in V$ for all $X \in V$.

    \item An \emph{equation} is an unordered pair $t \eqq u$ of terms in $\terms{\mathbb{S}}$.

    \item A \emph{unification problem} or simply \emph{problem} is a set of equations $P = \{t_1 \eqq u_1, ..., t_n \eqq u_n\}$.

    \item A \emph{solution} for a problem $P$ is a substitution $\theta$ such that for all $t \eqq u \in P$, $\theta t = \theta u$. In this case, we say that the terms $t$ and $u$ are \emph{unifiable} and that $\theta$ is a \emph{unifier} for them.

    \item Two terms $t$ and $u$ are $\alpha$-unifiable if there exists a renaming $\alpha$ such that $\{\alpha t \eqq u\}$ has a solution which is called $\alpha$-unifier. An $\alpha$-unification between two terms is \emph{exact} when it is a renaming.

    \item Two terms $t, u$ are \emph{equivalent up to renaming}, written $t \alphaeq u$, if there exists an exact $\alpha$-unifier between them.
\end{itemize}

These definitions define a preorder on terms. A term $t$ is lesser than another term $u$ when it is more specialised or less general. In terms of substitutions, it means that $t$ can be obtained by instantiating the variables of $u$ with other terms.

\begin{defi}[Order on terms]
\label{def:orderterms}
We define the following partial order: given $t, u$ two terms, $t \preceq u$ if and only if there exists a substitution $\theta$ such that $t = \theta u$. We consider the order up to renaming, \ie $t = u$ when $t \alphaeq u$.
\end{defi}

\begin{prop}
\label{prop:preorderterms}
The relation $\preceq$ defines a preorder.
\end{prop}
\begin{proof}
Let $t$ be a term. If $\theta$ is the identity substitution, we have $t = \theta t$. Let $t_1, t_2, t_3$ be terms. Assume $t_1 = \theta_a t_2$ and $t_2 = \theta_b t_3$. We can compose the two substitutions and obtain $\theta_a \circ \theta_b$. We have $(\theta_a \circ \theta_b) t_3 = \theta_a (\theta_b t_3) = \theta_a t_2 = t_1$.
\end{proof}

Our definition of $\alpha$-unification comes from a simplification of Aubert and Bagnol's definition of \emph{matching} \cite[Definition 6]{aubert2014unification} itself appearing in Girard's definitions \cite[Section 1.1.2]{girard2013three}. However, since \emph{matching} already exists with a different definition in the literature we chose a different name. We show that our simplification is equivalent to Aubert and Bagnol's definition definition.

\begin{prop}
Two terms $t_1$ and $t_2$ are $\alpha$-unifiable if and only if there exists two renamings $\alpha_1$ and $\alpha_2$ such that $\alpha_1 t_1$ and $\alpha_2 t_2$ are unifiable and that $\vars{\alpha_1 t_1} \cap \vars{\alpha_2 t_2} = \emptyset$.
\end{prop}
\begin{proof}
($\Rightarrow$) Assume that $t_1$ and $t_2$ are $\alpha$-unifiable. Hence, there exists $\alpha$ such that $\theta \alpha t_1 = \theta t_2$ for some substitution $\theta$. For $\alpha_1 := \alpha$ and $\alpha_2 := \emptyset$, the empty renaming which is indeed disjoint from $\alpha$, we have $\theta \alpha_1 t_1 = \theta \alpha_2 t_2$.
$(\Leftarrow)$ Now assume that there exists two renamings $\alpha_1$ and $\alpha_2$ such that $\alpha_1 t_1$ and $\alpha_2 t_2$ are unifiable and that $\vars{\alpha_1 t_1} \cap \vars{\alpha_2 t_2} = \emptyset$. We have $\theta\alpha_1 t_1 = \theta\alpha_2 t_2$ for some $\theta$. We can define the substitution $\psi := \theta \circ \alpha_2$ and the renaming $\alpha := \alpha_2^{-1} \circ \alpha_1$ such that $\psi \alpha t_1 = \psi t_2$ since we have $\theta\alpha_1 t_1 = \theta\alpha_2 t_2$. This shows that $t_1$ and $t_2$ are $\alpha$-unifiable.
\end{proof}

The problem of deciding if a solution to a given problem $P$ exists is known to be decidable in the literature. Moreover, there exists a maximal solution $\solution{P}$ \wrt the preorder $\preceq$, which is unique up to renaming. Several algorithms were designed to compute the unique solution when it exists, such that the Martelli-Montanari unification algorithm \cite{unifalgo}.

\begin{defi}[Solved form]
A unification problem $P = \{X_1 \eqq t_1, ..., X_n \eqq t_n\}$ with $X_1, ..., X_n \in V$ is in \emph{solved form} if no variable $X_1, ..., X_2$ appears on the right of an equation, \ie $\{X_1, ..., X_n\} \cap \bigcup_{j=1}^{n} \mathtt{fv}(t_j) = \emptyset$. Its \emph{underlying substitution} is defined by $\overset{\rightarrow}{P} := \substseq{X_1 \mapsto t_1, ..., X_n \mapsto t_n}$.
\end{defi}

\begin{defi}[Martelli-Montanari unification algorithm]
We define a non-deterministic algorithm with inference rules read from top to bottom:

\[
\begin{prooftree}
\hypo{P \cup \{ t \eqq t \}}
\infer1[clear]{P}
\end{prooftree}
\qquad
\begin{prooftree}
\hypo{P \text{ (in solved form)}}
\infer1[success]{}
\end{prooftree}
\qquad
\begin{prooftree}
\hypo{P \text{ (not in solved form)}}
\infer1[fail]{\bot}
\end{prooftree}
\]

\[
\begin{prooftree}
\hypo{P \cup \{ f(t_1, ..., t_n) \eqq f(u_1, ..., u_n) \}}
\infer1[open]{P \cup \{ t_1 \eqq u_1, ..., t_n \eqq u_n \}}
\end{prooftree}
\qquad
\begin{prooftree}
\hypo{P \cup \{ t \eqq X \} \text{ with } t \not\in \vars{t}}
\infer1[orient]{P \cup \{ X \eqq t \}}
\end{prooftree}
\]

\[
\begin{prooftree}
\hypo{P \cup \{ X \eqq t \} \text{ with } X \in \vars{P} \text{ and } X \not\in \vars{t}}
\infer1[replace]{\substseq{X \mapsto t} P \cup \{ X \eqq t \}}
\end{prooftree}
\]

where $\vars{P}$ and $\vars{t}$ are the sets of variables occurring in $P$ and $t$. A sequence constructed by these rules and ending with a success or fail rule when no other rule can be applied is called an \emph{execution} of the unification algorithm. The last step of the sequence is written $\solution{P}$ for a problem $P$. If we can apply more rules, it is called a \emph{partial execution}.
\end{defi}

\begin{thm}[Confluence of the unification algorithm]
Any solvable problem $P$ has a unique solution modulo $\alphaeq$.
\end{thm}
\begin{proof}
Proven in \cite[Theorem 3.17]{lassez1988unification}.
\end{proof}

We complete the section with the resolution rule for first-order logic for which the unification is central.

\begin{defi}[Resolution rule]
\label{def:resolution}
Let $C, D$ be two set of first-order atoms and $P$ a predicate of arity $n$. The \emph{resolution rule} is defined by the following inference rule:
\begin{center}
\begin{prooftree}
    \hypo{C \cup \{P(t_1, ..., t_n)\}}
    \hypo{D \cup \{\lnot P(u_1, ..., u_n)\}}
    \infer2[Res]{\theta (C \cup D)}
\end{prooftree}
\end{center}
where $\theta := \solution{\mathcal{P}(\bigcup_{k=1}^n \{t_i \eqq u_i\})}$.
\end{defi}

\section{Hypergraph theory}\label{sec:hypergraphs}

Hypergraphs are graphs with edges potentially linking several vertices and thus generalise graphs. We use a definition of directed hypergraphs with several targets (we usually consider a single target in the literature \cite{bretto2013hypergraph}).

\begin{defi}[Directed hypergraph]
An \emph{directed hypergraph} is defined by a tuple $H = (V, E, \ein{}, \eout{})$ where:
\begin{itemize}
    \item $V$ and $E$ are respectively the set of vertices and hyperedges;
    \item $\ein{} : E \rightarrow P(V)$ defines the sources of a hyperedge;
    \item $\eout{} : E \rightarrow P(V)$ defines the targets of a hyperdge
\end{itemize}
where $P(V)$ is the powerset of $V$. The hypergraph is \emph{undirected} when we forget $\eout{}$ and $\ein{}$ defines the vertices connected by an hyperedge.
\end{defi}

\begin{defi}[Disjoint union of hypergraphs]
The \emph{disjoint union} of two hypergraphs $H = (V, E, \ein{}, \eout{})$ and $H' = (V', E', \ein{}', \eout{}')$ is defined by $H \uplus H' := (V \uplus V', E \uplus E', \ein{} \uplus \ein{}', \eout{}' \uplus \eout{}')$ where $\uplus$ is the disjoint union of sets and $f \uplus f'$ for two functions $f$ and $f'$ is the function corresponds to the disjoint union of the function graph of $f$ and $f'$. 
\end{defi}

\begin{defi}[(Multi)Graph homomorphism]
A \emph{(multi)graph homomorphism} is a function between graphs $f : G \rightarrow G'$ preserving the structure of graph, \ie for all $(u, v) \in E^G$, we have $(f(u), f(v)) \in E^{G'}$.
\end{defi}

\begin{defi}[(Multi)Graph isomorphism]
A (multi)graph homomorphism $f$ is an \emph{isomorphism} if it is a bijection and that its inverse $f^{-1}$ is also a (multi)graph homomorphism. 
\end{defi}

\begin{defi}[Path]
A \emph{(directed) path} is an ordered alternated sequence of vertices and hyperedges $v_1 e_1 ... v_n, e_n, v_{n+1}$ for $v_i \in V$ and $e_i \in E$.

For an undirected graph, the path has an unordered sequence.
\end{defi}

\begin{defi}[Accessibility relation]
\label{def:accrel}
The notion of path defines an \emph{accessibility relation} on vertices for a hypergraph $H$. For two vertices $x$ and $y$, we have $x \equiv_\mathrm{acc}^H y$ when there is a path from $x$ to $y$ and, if the edges are oriented/directed, from $y$ to $x$ as well. 
\end{defi}

\begin{defi}[Connected components and connectedness]
The set of \emph{connected components} of a hypergraph is the set of equivalence classes defined by the equivalence relation induced by the accessibility of vertices with paths. A hypergraph is connected when it has a unique connected component.
\end{defi}

\begin{defi}[Cycles and acyclicity]
A \emph{cycle} is a path $v_1 e_1 ... v_n, e_n, v_{n+1}$ such that $v_1 = v_{n+1}$. A hypergraph is \emph{acyclic} when it has no cycle.
\end{defi}

\begin{nota}
For a (multi)graph or hypergraph $G$ of vertices $V$ and (hyper)edges $E$, we write $V^G$ for $V$ and $E^G$ for $E$.
\end{nota}

\end{document}